\documentclass{article}

\usepackage{amsmath, amsthm, amssymb}
\usepackage{graphicx}
\usepackage{verbatim}
\usepackage{caption}
\usepackage{subcaption}
\usepackage{fancyvrb}
\usepackage{enumerate}
\usepackage{setspace}

\usepackage{refcount}

\usepackage{relsize}
\usepackage{float}

\usepackage{multirow}
\usepackage{rotating}

\RequirePackage{natbib}

\theoremstyle{plain}

\newtheorem*{theo*}{Theorem}

\newtheorem{prop}{Proposition}

\newtheorem{coro}[prop]{Corollary}
\newtheorem{lemm}[prop]{Lemma}
\newtheorem{theo}[prop]{Theorem}

\theoremstyle{definition}

\newtheorem{defi}{Definition}

\newtheorem{deficmp}{Definition}

\newtheorem{rema}{Remark}

\usepackage{stefan_tex}

\theoremstyle{definition}
\newfloat{algbox}{tbp}{lop}
\newtheorem{alg}{Procedure}
\newcommand{\myalg}[3]{
\begin{center}
\fbox{
\parbox{0.95\textwidth}{
\begin{alg}\label{#1}{\textsc{ #2}}
\vspace{.1cm}\\ #3 
\end{alg}
}}
\end{center}
}

\usepackage{hyperref}
\usepackage[margin=1.49in]{geometry}
\hypersetup{colorlinks,citecolor=blue,urlcolor=blue,linkcolor=blue}

\newcommand{\rf}{\operatorname{RF}}
\newcommand{\cf}{\operatorname{CF}}

\newcommand{\hVIJ}{\widehat{V}_{IJ}}

\newcommand{\bY}{\widebar{Y}}

\newcommand{\diam}{\operatorname{diam}}

\newcommand{\proj}[1]{\mathring{#1}}

\newcommand{\figw}{0.6\textwidth}

\title{Estimation and Inference of Heterogeneous Treatment Effects using Random Forests\footnote{Part of the results developed in this paper were made available as an earlier technical report \emph{``Asymptotic Theory for Random Forests''}, available at \texttt{http://arxiv.org/abs/1405.0352}.}}

\author{Stefan Wager \\ Department of Statistics \\ Stanford University \\ \texttt{swager@stanford.edu}
\and
Susan Athey \\ Graduate School of Business \\ Stanford University \\ \texttt{athey@stanford.edu}}

\begin{document}

\maketitle

\begin{abstract}
Many scientific and engineering challenges---ranging from personalized medicine to customized marketing recommendations---require an understanding of treatment effect heterogeneity. In this paper, we develop a non-parametric \emph{causal forest} for estimating heterogeneous treatment effects that extends Breiman's widely used random forest algorithm. In the potential outcomes framework with unconfoundedness, we show that causal forests are pointwise consistent for the true treatment effect, and have an asymptotically Gaussian and centered sampling distribution. We also discuss a practical method for constructing asymptotic confidence intervals for the true treatment effect that are centered at the causal forest estimates. Our theoretical results rely on a generic Gaussian theory for a large family of random forest algorithms. To our knowledge, this is the first set of results that allows any type of random forest, including classification and regression forests, to be used for provably valid statistical inference. In experiments, we find causal forests to be substantially more powerful than classical methods based on nearest-neighbor matching, especially in the presence of irrelevant covariates.

\vspace{\baselineskip}
\end{abstract}


\section{Introduction}

In many applications, we want to use data to draw inferences about the causal effect of a treatment: Examples include medical studies about the effect of a drug on health outcomes, studies of the impact of advertising or marketing offers on consumer purchases, evaluations of the effectiveness of government programs or public policies, and ``A/B tests'' (large-scale randomized experiments) commonly used by technology firms to select algorithms for ranking search results or making recommendations.
Historically, most datasets have been too small to meaningfully explore heterogeneity of treatment effects beyond dividing the sample into a few subgroups. Recently, however, there has been an explosion of empirical settings where it is potentially feasible to customize estimates for individuals.    

An impediment to exploring heterogeneous treatment effects is the fear that researchers will iteratively
search for subgroups with high treatment levels, and then report only the results for subgroups with extreme effects, thus 
highlighting heterogeneity that may be purely spurious \citep{assmann2000subgroup, cook2004subgroup}.
For this reason, protocols for clinical trials must specify in advance which subgroups will be analyzed, and other disciplines such
as economics have instituted protocols for registering pre-analysis plans for randomized experiments or surveys.
However, such procedural restrictions can make it difficult to discover strong but unexpected treatment effect heterogeneity.
In this paper, we seek to address this challenge by developing a powerful, nonparametric method for heterogeneous treatment effect estimation that yields valid asymptotic confidence intervals for the true underlying treatment effect.

Classical approaches to nonparametric estimation of heterogeneous treatment effects include nearest-neighbor matching, kernel methods, and series estimation; see, e.g., \citet{crump2008nonparametric}, \citet{lee2009non}, and \citet{willke2012concepts}. These methods perform well in applications with a small number of covariates, but quickly break down as the number of covariates increases.
In this paper, we explore the use of ideas from the machine learning literature to improve the performance of these classical methods with many covariates.
We focus on the family of random forest algorithms introduced by \citet{breiman2001random}, which allow for flexible modeling of interactions in high dimensions by building a large number of regression trees and averaging their predictions.
Random forests are related to kernels and nearest-neighbor methods in that they make predictions using a weighted average of ``nearby" observations;
however, random forests differ in that they have a 
data-driven way to determine which nearby observations receive more weight, something that
is especially important in environments with many covariates or complex interactions among covariates. 

Despite their widespread success at prediction and classification, there are important hurdles that need to be cleared before random forests are directly useful to causal inference. Ideally, an estimator should be consistent with a well-understood 
asymptotic sampling distribution, so that a researcher can use it to test hypotheses and establish confidence intervals.   
For example, when deciding to use a drug for an individual, we may wish to test the hypothesis 
that the expected benefit from the treatment is less than the
treatment cost.  Asymptotic normality results are especially important in the causal inference setting, both
because many policy applications require confidence intervals for decision-making, and because
it can be difficult to directly evaluate the model's performance using, e.g., cross validation, when estimating causal effects.  Yet, the asymptotics of random forests have been largely left open, even in the standard regression or classification contexts.

This paper addresses these limitations, developing a forest-based method for treatment effect estimation that allows for a tractable asymptotic theory and valid statistical inference.
Following \citet{athey2015machine}, our proposed forest is composed of \emph{causal trees} that estimate the
effect of the treatment at the leaves of the trees; we thus refer to our algorithm as a \emph{causal forest}.

In the interest of generality, we begin our theoretical analysis by developing the desired consistency and asymptotic normality results in the context of regression forests.
We prove these results for a particular variant of regression forests that uses subsampling to generate a variety of different trees, while it relies on
deeply grown trees that satisfy a condition we call ``honesty'' to reduce bias. An example of an honest tree is one where the tree is grown using one subsample, while the predictions at the leaves of the tree are 
estimated using a different subsample.  We also show that the heuristically motivated infinitesimal jackknife for random forests developed by \citet{efron2013estimation} and \citet{wager2014confidence} is consistent for the asymptotic variance of random forests in this setting.  Our proof builds on classical ideas from \citet{efron1981jackknife}, \citet{hajek1968asymptotic}, and \citet{hoeffding1948class}, as well as the adaptive nearest neighbors interpretation of random forests of \citet{lin2006random}.
Given these general results, we next show that our consistency and asymptotic normality results extend from the regression setting to estimating heterogeneous treatment effects in the potential outcomes framework with unconfoundedness
\citep{neyman1923applications,rubin1974estimating}.

Although our main focus in this paper is causal inference, we note that 
there are a variety of important applications of the asymptotic normality result in a pure prediction context.  For example,
\citet{kleinberg2015prediction} seek to improve the allocation of medicare
funding for hip or knee replacement surgery by detecting patients who had been prescribed such a surgery, but were
in fact likely to die of other causes before the surgery would have been useful to them. Here we need
predictions for the probability that a given patient will survive for more than, say, one year that come with rigorous
confidence statements; our results are the first that enable the use of random forests for this purpose.

Finally, we compare the performance of the causal forest algorithm against classical $k$-nearest neighbor
matching using simulations, finding that the causal forest dominates in terms of both bias and variance in a variety of
settings, and that its advantage increases with the number of covariates.  We also examine coverage
rates of our confidence intervals for heterogeneous treatment effects.

\subsection{Related Work}

There has been a longstanding understanding in the machine learning literature that prediction methods such as random forests ought to be validated empirically \citep{breiman2001statistical}: if the goal is prediction, then we should hold out a test set, and the method will be considered as good as its error rate is on this test set. 
However, there are fundamental challenges with applying
a test set approach in the setting of causal inference.  In the widely used potential outcomes framework we use to formalize our results \citep{neyman1923applications,rubin1974estimating}, a treatment effect is understood as a difference between two potential outcomes,
e.g., would the patient have died if they received the drug vs. if they didn't receive it. Only one of these potential outcomes can ever be observed in practice, and so direct test-set evaluation is in general impossible.\footnote{\citet{athey2015machine} have proposed indirect approaches to mimic test-set evaluation for causal inference. However, these approaches require an estimate of the true treatment effects and/or treatment propensities for all the observations in the test set, which creates a new set of challenges. In the absence of an observable ground truth in a test set, statistical theory plays a more central role in evaluating the noise in estimates of causal effects.} Thus, when evaluating estimators of causal effects, asymptotic theory plays a much more important role than in the standard prediction context.

From a technical point of view, the main contribution of this paper is an asymptotic normality theory enabling us to do statistical inference using random forest predictions. Recent results by \citet{biau2012analysis, meinshausen2006quantile, mentch2014ensemble, scornet2015consistency}, and others have established asymptotic properties of particular variants and simplifications of the random
forest algorithm.  To our knowledge, however, we provide the first set of conditions under which predictions made by random forests are both asymptotically unbiased and Gaussian, thus allowing for classical statistical inference; the extension to the causal forests 
proposed in this paper is also new. We review the existing theoretical literature on random forests in more detail in Section \ref{sec:theor_relworks}.

A small but growing literature, including \citet{green2012modeling}, \citet{hill2011bayesian} and \citet{hill2013assessing}, has considered the use of forest-based algorithms for estimating heterogeneous treatment effects. These papers use the Bayesian Additive Regression Tree (BART) method of \citet{chipman2010bart}, and report posterior credible intervals obtained by Markov-chain Monte Carlo (MCMC) sampling based on a convenience prior.
Meanwhile, \citet{foster2011subgroup} use regression forests to estimate the effect of covariates on outcomes in treated and control groups separately, and then take the difference in predictions as data and project treatment effects onto units' attributes using regression or classification trees (in contrast, we modify the standard random forest algorithm to focus on directly estimating heterogeneity in causal effects).
A limitation of this line of work is that, until now, it has lacked formal statistical inference results.

We view our contribution as complementary to this literature, by showing that forest-based methods need not only be viewed as black-box heuristics, and can instead be used for rigorous asymptotic analysis. We believe that the theoretical tools developed here will be useful beyond the specific class of algorithms studied in our paper. In particular, our tools allow for a fairly direct analysis of variants of the method of \citet{foster2011subgroup}. Using BART for rigorous statistical analysis may prove more challenging since, although BART is often successful in practice, there are currently no results guaranteeing posterior concentration around the true conditional mean function, or convergence of the MCMC sampler in polynomial time. Advances of this type would be of considerable interest.

Several papers use tree-based methods for estimating heterogeneous treatment effects.
In growing trees to build our forest, we follow most closely the approach of \citet{athey2015machine},
who propose honest, causal trees, and obtain valid confidence intervals for average treatment effects for
each of the subpopulations (leaves) identified by the algorithm. (Instead of personalizing predictions for
each individual, this approach only provides treatment effect estimates for leaf-wise subgroups whose
size must grow to infinity.)
Other related approaches  include those of  \citet{su2009subgroup} and \citet{zeileis2008model}, which
build a tree for treatment effects in subgroups and use statistical tests to determine splits; however,
these papers do not analyze bias or consistency properties.  
 
Finally, we note a growing literature on estimating heterogeneous treatment effects using different machine learning methods. \citet{imai2013estimating}, \citet{signorovitch2007identifying}, \citet{tian2014simple} and \citet{weisberg2015post} develop lasso-like methods for causal inference in a sparse high-dimensional linear setting.
\citet{beygelzimer2009offset}, \citet{langford2011doubly}, and others discuss procedures for transforming outcomes that enable off-the-shelf loss minimization methods to be used for optimal treatment policy estimation.
In the econometrics literature,
\citet{bhattacharya2012inferring,dehejia2005program,hirano2009asymptotics,manski2004statistical}
estimate parametric or semi-parametric models for optimal policies, relying on regularization for covariate selection in the case of \citet{bhattacharya2012inferring}.  
\citet{taddy2014heterogeneous} use Bayesian nonparametric methods with Dirichlet priors to flexibly estimate the data-generating process, and then project the estimates of heterogeneous treatment effects down onto the feature space using regularization methods or regression trees to get low-dimensional summaries of the heterogeneity; but again, there are no guarantees about asymptotic properties.

\section{Causal Forests}
\label{sec:main}

\subsection{Treatment Estimation with Unconfoundedness}
\label{sec:unconfoundedness}

Suppose we have access to $n$ independent and identically distributed training examples labeled
$i = 1, \, ..., \, n$, each of which consists of a feature vector $X_i \in [0, \, 1]^d$, a response
$Y_i \in \RR$, and a treatment indicator $W_i \in \cb{0, \, 1}$. Following the potential outcomes framework of
\citet{neyman1923applications} and \citet{rubin1974estimating} (see \citet{imbens2015causal} for a review),
we then posit the existence of potential outcomes $\smash{Y_i^{(1)}}$ and $\smash{Y_i^{(0)}}$ corresponding
respectively to the response the $i$-th subject would have experienced with and without the treatment, and define
the treatment effect at $x$ as
\begin{equation}
\label{eq:tau}
\tau\p{x} = \EE{Y_i^{(1)} - Y_i^{(0)} \cond X_i = x}.
\end{equation}
Our goal is to estimate this function $\tau(x)$. The main difficulty is that we can only ever observe one of the two potential outcomes $\smash{Y_i^{(0)}, \, Y_i^{(1)}}$ for a given training example, and so cannot directly train machine learning methods on differences of the form \smash{$Y_i^{(1)} - Y_i^{(0)}$}.

In general, we cannot estimate $\tau(x)$ simply from the observed data $(X_i, \, Y_i, \, W_i)$ without further restrictions on the data generating distribution. A standard way to make progress is to assume unconfoundedness \citep{rosenbaum1983central}, i.e., that the treatment assignment $W_i$ is independent of the potential outcomes for $Y_i$ conditional on $X_i$:
\begin{equation}
\label{eq:unconf}
\cb{Y_i^{(0)}, \, Y_i^{(1)}} \  \mathlarger{\mathlarger{\raisebox{-0.5mm}{$\indep$}}} \ W_i \ \ \big| \ \ X_i.
\end{equation}
The motivation behind unconfoundedness is that, given continuity assumptions, it effectively implies that we can treat nearby observations in $x$-space as having come from a randomized experiment; thus, nearest-neighbor matching and other local methods will in general be consistent for $\tau(x)$.

An immediate consequence of unconfoundedness is that
\begin{equation}
\label{eq:propensity}
\EE{Y_i \p{\frac{W_i}{e\p{x}} - \frac{1 - W_i}{1 - e\p{x}}} \cond X_i = x} = \tau\p{x}, \ \where \ e\p{x} = \EE{W_i \cond X_i = x}
\end{equation}
is the propensity of receiving treatment at $x$. Thus, if we knew $e(x)$, we would have access to a simple unbiased estimator for $\tau(x)$; this observation lies at the heart of methods based on propensity weighting \citep[e.g.,][]{hirano2003efficient}. Many early applications of machine learning to causal inference effectively reduce to estimating $e(x)$ using, e.g., boosting, a neural network, or even random forests, and then transforming this into an estimate for $\tau(x)$ using \eqref{eq:propensity} \citep[e.g.,][]{mccaffrey2004propensity,westreich2010propensity}.
In this paper, we take a more indirect approach: We show that, under regularity assumptions, causal forests can use the unconfoundedness assumption \eqref{eq:unconf} to achieve consistency without needing to explicitly estimate the propensity $e(x)$.

\subsection{From Regression Trees to Causal Trees and Forests}

At a high level, trees and forests can be thought of as nearest neighbor methods with an adaptive neighborhood metric. Given a test point $x$, classical methods such as $k$-nearest neighbors seek the $k$ closest points to $x$ according to some pre-specified distance measure, e.g., Euclidean distance. In contrast, tree-based methods also seek to find training examples that are close to $x$, but now closeness is defined with respect to a decision tree, and the closest points to $x$ are those that fall in the same leaf as it. The advantage of trees is that their leaves can be narrower along the directions where the signal is changing fast and wider along the other directions, potentially leading a to a substantial increase in power when the dimension of the feature space is even moderately large.

In this section, we seek to build causal trees that resemble their regression analogues as closely as possible. Suppose first that we only observe independent samples $(X_i, \, Y_i)$, and want to build a CART regression tree. We start by recursively splitting the feature space until we have partitioned it into a set of leaves $L$, each of which only contains a few training samples. Then, given a test point $x$, we evaluate the prediction $\hmu(x)$ by identifying the leaf $L(x)$ containing $x$ and setting
\begin{equation}
\label{eq:pred_tree}
\hmu\p{x} = \frac{1}{\abs{\cb{i : X_i \in L(x)}}} \sum_{\cb{i : X_i \in L(x)}} Y_i.
\end{equation}
Heuristically, this strategy is well-motivated if we believe the leaf $L(x)$ to be small enough that the responses $Y_i$ inside the leaf are roughly identically distributed. There are several procedures for how to place the splits in the decision tree; see, e.g., \citet{hastie2009elements}. 

In the context of causal trees, we analogously want to think of the leaves as small enough that the $(Y_i, \, W_i)$ pairs corresponding to the indices $i$ for which $i \in L(x)$ act as though they had come from a randomized experiment. Then, it is natural to estimate the treatment effect for any $x \in L$ as
\begin{equation}
\label{eq:local_balance}
\htau\p{x} = \frac{1}{\abs{\cb{i : W_i = 1, \, X_i \in L}}} \sum_{\cb{i : W_i = 1, \, X_i \in L}} \!\!\!\!\!\!\!\!\!\! Y_i \ \ - \ \  \frac{1}{\abs{\cb{i : W_i = 0, \, X_i \in L}}} \sum_{\cb{i : W_i = 0, \, X_i \in L}} \!\!\!\!\!\!\!\!\!\! Y_i. 
\end{equation}
In the following sections, we will establish that such trees can be used to grow causal forests that are consistent for $\tau(x)$.\footnote{The causal tree algorithm presented above is a simplification of the method of \citet{athey2015machine}. The main difference between our approach and that of \citet{athey2015machine} is that they seek to build a single well-tuned tree; to this end, they use fairly large leaves and apply a form propensity weighting based on \eqref{eq:propensity} within each leaf to correct for variations in $e(x)$ inside the leaf. In contrast, we follow \citet{breiman2001random} and build our causal forest using deep trees. Since our leaves are small, we are not required to apply any additional corrections inside them. However, if reliable propensity estimates are available, using them as weights for our method may improve performance (and would not conflict with the theoretical results).}

Finally, given a procedure for generating a single causal tree, a causal forest generates an ensemble of $B$ such trees, each of which outputs an estimate $\htau_b(x)$. The forest then aggregates their predictions by averaging them: \smash{$\htau(x) = B^{-1} \sum_{b = 1}^B \htau_b(x)$}. We always assume that the individual causal trees in the forest are built using random subsamples of $s$ training examples, where $s/n \ll 1$; for our theoretical results, we will assume that \smash{$s \asymp n^\beta$} for some $\beta < 1$.
The advantage of a forest over a single tree is that it is not always clear what the ``best'' causal tree is. In this case, as shown by \citet{breiman2001random}, it is often better to generate many different decent-looking trees and average their predictions, instead of seeking a single highly-optimized tree. In practice, this aggregation scheme helps reduce variance and smooths sharp decision boundaries \citep{buhlmann2002analyzing}.

\subsection{Asymptotic Inference with Causal Forests}
\label{sec:theory_overview}

Our results require some conditions on the forest-growing scheme: The trees used to build the forest must be grown on subsamples of the training data, and the splitting rule must not ``inappropriately'' incorporate information about the outcomes $Y_i$ as discussed formally in Section \ref{sec:split}. However, given these high level conditions, we obtain a widely applicable consistency result that applies to several different interesting causal forest algorithms.

Our first result is that causal forests are consistent for the true treatment effect $\tau(x)$. 
To achieve pointwise consistency, we need to assume that the conditional mean functions \smash{$\EE{Y^{(0)} \cond X = x}$} and \smash{$\EE{Y^{(1)} \cond X = x}$} are both Lipschitz continuous.
To our knowledge, all existing results on pointwise consistency of regression forests \citep[e.g.,][]{biau2012analysis,meinshausen2006quantile} require an analogous condition on $\EE{Y \cond X = x}$.
This is not particularly surprising, as forests generally have smooth response surfaces \citep{buhlmann2002analyzing}.
In addition to continuity assumptions, we also need to assume that we have overlap, i.e., for some $\varepsilon > 0$ and all $x \in [0, \, 1]^d$,
\begin{equation}
\label{eq:overlap}
\varepsilon < \PP{W = 1 \cond X = x} < 1 - \varepsilon.
\end{equation}
This condition effectively guarantees that, for large enough $n$, there will be enough treatment and control units near any test point $x$ for local methods to work.

Beyond consistency, in order to do statistical inference on the basis of the estimated treatment effects $\htau(x)$, we need to understand their asymptotic sampling distribution.
Using the potential nearest neighbors construction of \citet{lin2006random} and classical analysis tools going back to \citet{hoeffding1948class} and \citet{hajek1968asymptotic}, we show that---provided the sub-sample size $s$ scales appropriately with $n$---the predictions made by a causal forest are asymptotically Gaussian and unbiased.
Specifically, we show that 
\begin{equation}
\label{eq:causal_gauss}
\p{\htau\p{x} - \tau\p{x}} \, \big/ \, {\sqrt{\Var{\htau(x)}}} \Rightarrow \nn\p{0, \, 1}
\end{equation}
under the conditions required for consistency, provided the subsample size $s$ scales
as $s \asymp n^\beta$ for some $\beta_{\min} < \beta < 1$

Moreover, we show that the asymptotic variance of causal forests can be accurately estimated. To do so, we use the infinitesimal jackknife for random forests developed by \citet{efron2013estimation} and \citet{wager2014confidence}, based on the original infinitesimal jackknife procedure of \citet{jaeckel1972infinitesimal}. This method assumes that we have taken the number of trees $B$ to be large enough that the Monte Carlo variability of the forest does not matter; and only measures the randomness in $\htau(x)$ due to the training sample.

To define the variance estimates, let $\htau_b^*(x)$ be the treatment effect estimate given by the $b$-th tree, and let $N_{ib}^* \in \cb{0, \, 1}$ indicate whether or not the $i$-th training example was used for the $b$-th tree.\footnote{For double-sample trees defined in Procedure \ref{alg:split}, $N_{ib}^* = 1$ if the $i$-th example appears in either the $\ii$-sample or the $\jj$-sample.}
Then, we set
\begin{equation}
\label{eq:hvij}
\hVIJ\p{x}  = \frac{n - 1}{n} \p{\frac{n}{n - s}}^2 \sum_{i = 1}^n \Cov[*]{\htau^*_b\p{x}, \, N^*_{ib}}^2,
\end{equation}
where the covariance is taken with respect to the set of all the trees $b=1,..,B$ used in the forest. The term $n (n - 1)/(n - s)^2$ is a finite-sample correction for forests grown by subsampling without replacement; see Proposition \ref{prop:finite_sample}.
We show that this variance estimate is consistent, in the sense that $\hVIJ\p{x}/\Var{\htau(x)} \rightarrow_p 1$.

\subsection{Honest Trees and Forests}
\label{sec:split}

In our discussion so far, we have emphasized the flexible nature of our results: for a wide variety of causal forests that can be tailored to the application area, we achieve both consistency and centered asymptotic normality, provided the sub-sample size $s$ scales at an appropriate rate.
Our results do, however, require the individual trees to satisfy a fairly strong condition, which we call honesty: a tree is honest if, for each training example $i$, it only uses the response $Y_i$ to estimate the within-leaf treatment effect $\tau$ using \eqref{eq:local_balance} or to decide where to place the splits, but not both. We discuss two causal forest algorithms that satisfy this condition.

Our first algorithm, which we call a double-sample tree, achieves honesty by dividing its training subsample into two halves $\ii$ and $\jj$. Then, it uses the $\jj$-sample to place the splits, while holding out the $\ii$-sample to do within-leaf estimation; see Procedure \ref{alg:split} for details.
In our experiments, we set the minimum leaf size to $k = 1$.
A similar family of algorithms was discussed in detail by \citet{denil2014narrowing}, who showed that such forests could achieve competitive performance relative to standard tree algorithms that do not divide their training samples.
In the semiparametric inference literature, related ideas go back at least to the work of \citet{schick1986asymptotically}.

We note that sample splitting procedures are sometimes criticized as inefficient because
they ``waste'' half of the training data at each step of the estimation procedure. However, in our case, the forest
subampling mechanism enables us to achieve honesty without wasting any data in this sense,
because we re-randomize the $\ii/\jj$-data splits over each subsample.
Thus, although no data point can be used for
split selection and leaf estimation in a single tree, each data point will participate in both
$\ii$ and $\jj$ samples of some trees, and so will be used for both specifying the
structure and treatment effect estimates of the forest.  Although our original motivation for considering double-sample trees was to eliminate bias
and thus enable centered confidence intervals, we find that in practice, double-sample trees can improve upon standard random forests in terms of mean-squared error as well.

\begin{algbox}[t]
\myalg{alg:split}{Double-Sample Trees}{

\vspace{-0.8\baselineskip}

Double-sample trees split the available training data into two parts: one half for estimating the desired response inside each leaf, and another half for placing splits.

\vspace{0.5\baselineskip}

Input: $n$ training examples of the form $(X_i, \, Y_i)$ for regression trees or $(X_i, \, Y_i, \, W_i)$ for causal trees, where $X_i$ are features, $Y_i$ is the response, and $W_i$ is the treatment assignment. A minimum leaf size $k$.

\begin{enumerate}
\item  Draw a random subsample of size $s$ from $\cb{1, \, ..., \, n}$ without replacement, and then divide it into two disjoint sets of size $\abs{\ii} = \lfloor s/2 \rfloor$ and $\abs{\jj} = \lceil s/2 \rceil$.
\item Grow a tree via recursive partitioning. The splits are chosen using any data from the $\jj$ sample and $X$- or $W$-observations from the $\ii$ sample, but without using $Y$-observations from the $\ii$-sample.
\item Estimate leaf-wise responses using only the $\ii$-sample observations.
\end{enumerate}

\vspace{0.5\baselineskip}

Double-sample \emph{regression} trees make predictions $\hmu(x)$ using \eqref{eq:pred_tree} on the leaf containing $x$, only using the $\ii$-sample observations. The splitting criteria is the standard for CART regression trees (minimizing mean-squared error of predictions). Splits are restricted so that each leaf of the tree must contain $k$ or more $\ii$-sample observations.

\vspace{0.5\baselineskip}
Double-sample \emph{causal} trees are defined similarly, except that for prediction we estimate $\htau(x)$ using \eqref{eq:local_balance} on the $\ii$ sample.   Following \citet{athey2015machine}, the splits of the tree are chosen by maximizing the variance of $\htau(X_i)$ for $i \in \jj$; see Remark \ref{rema:CT} for details.  In addition, each leaf of the tree must contain $k$ or more $\ii$-sample observations of \emph{each} treatment class.
}
\end{algbox}

Another way to build honest trees is to ignore the outcome data $Y_i$ when placing splits, and instead first train a classification tree for the treatment assignments $W_i$ (Procedure \ref{alg:prop}). Such propensity trees can be particularly useful in observational studies, where we want to minimize bias due to variation in $e(x)$. Seeking estimators that match training examples based on estimated propensity is a longstanding idea in causal inference, going back to \citet{rosenbaum1983central}.

\begin{algbox}[t]
\myalg{alg:prop}{Propensity Trees}{

\vspace{-0.8\baselineskip}

Propensity trees use only the treatment assignment indicator $W_i$ to place splits, and save the responses $Y_i$ for estimating $\tau$.

\vspace{0.5\baselineskip}

Input: $n$ training examples $(X_i, \, Y_i, \, W_i)$, where $X_i$ are features, $Y_i$ is the response, and $W_i$ is the treatment assignment. A minimum leaf size $k$.
\begin{enumerate}
\item Draw a random subsample $\ii \in \cb{1, \, ..., \, n}$ of size $\abs{\ii} = s$ (no replacement).
\item Train a classification tree using sample $\ii$ where the outcome is the treatment assignment, i.e., on the $(X_i, \, W_i)$ pairs with $i \in \ii$. Each leaf of the tree must have $k$ or more observations of \emph{each} treatment class.
\item Estimate $\tau(x)$ using \eqref{eq:local_balance} on the leaf containing $x$.
\end{enumerate}
In step 2, the splits are chosen by optimizing, e.g., the Gini criterion used by CART for classification \citep{breiman1984classification}.
}
\end{algbox}

\begin{rema}
\label{rema:CT}
For completeness, we briefly outline the motivation for the splitting rule of \citet{athey2015machine} we use for our double-sample trees. This method is motivated by an algorithm for minimizing the squared-error loss in regression trees. Because regression trees compute predictions $\hmu$ by averaging training responses over leaves, we can verify that
\begin{equation}
\sum_{i \in \jj} \p{\hmu\p{X_i} - Y_i}^2 = \sum_{i \in \jj} Y_i^2 - \sum_{i \in \jj} \hmu\p{X_i}^2.
\end{equation}
Thus, finding the squared-error minimizing split is equivalent to maximizing the variance of $\hmu(X_i)$ for $i \in \jj$; note that \smash{$\sum_{i \in \jj} \hmu(X_i) = \sum_{i \in \jj} Y_i$} for all trees, and so maximizing variance is equivalent to maximizing the sum of the $\hmu(X_i)^2$. In Procedure \ref{alg:split}, we emulate this algorithm by picking splits that maximize the variance of $\htau(X_i)$ for $i \in \jj$.\footnote{\citet{athey2015machine} also consider ``honest splitting rules'' that
anticipate honest estimation, and correct for the additional sampling variance in small leaves using an
idea closely related to the $C_p$ penalty of \citet{mallows1973some}. Although it could be of interest
for further work, we do not study the effect of such splitting rules here.}
\end{rema}

\begin{rema}
\label{rema:honesty}
In Appendix \ref{sec:honesty}, we present evidence that adaptive forests with small leaves
can overfit to outliers in ways that make them inconsistent near the edges of sample space.
Thus, the forests of \citet{breiman2001random} need to be modified in some way to get
pointwise consistency results; here, we use honesty following, e.g., \citet{wasserman2009high}.
We note that there have been some recent theoretical investigations of non-honest forests,
including \citet{scornet2015consistency} and \citet{wager2015uniform}. However,
\citet{scornet2015consistency} do not consider pointwise properties of forests; whereas
\citet{wager2015uniform} show consistency of adaptive forests with larger leaves, but
their bias bounds decay slower than the sampling variance of the forests and so cannot
be used to establish centered asymptotic normality.
\end{rema}

\section{Asymptotic Theory for Random Forests}
\label{sec:theor}

In order to use random forests to provide formally valid statistical inference, we need an asymptotic normality theory for random forests. In the interest of generality, we first develop such a theory in the context of classical regression forests, as originally introduced by \citet{breiman2001random}.  In this section, we assume that we have training examples $Z_i = \p{X_i, \, Y_i}$ for $i = 1, \, ..., \, n$, a test point $x$, and we want to estimate true conditional mean function
\begin{equation}
\label{eq:cond_mean}
\mu\p{x} = \EE{Y \cond X = x}.
\end{equation}
We also have access to a regression tree $T$ which can be used to get estimates of the conditional mean function at $x$ of the form $T\p{x; \, \xi, \, Z_1, \, ..., \, Z_n}$, where $\xi \sim \Xi$ is a source of auxiliary randomness. Our goal is to use this tree-growing scheme to build a random forest that can be used for valid statistical inference about $\mu(x)$.

We begin by precisely describing how we aggregate individual trees into a forest. For us, a random forest is an average of trees trained over all possible size-$s$ subsamples of the training data, marginalizing over the auxiliary noise $\xi$.
In practice, we compute such a random forest by Monte Carlo averaging, and set
\begin{equation}
\label{eq:rfm}
\rf \p{x; \, Z_1, \, ..., \, Z_n} \approx \frac{1}{B} \sum_{b = 1}^B  T\p{x; \, \xi^*_b, \, Z^*_{b1}, \, ..., \, Z^*_{bs}},
\end{equation}
where $\{Z^*_{b1}, \, ..., \, Z^*_{bs}\}$ is drawn without replacement from $\{Z_1, \, ..., \, Z_n\}$, $\xi^*_b$ is a random draw from $\Xi$, and $B$ is the number of Monte Carlo replicates we can afford to perform. The formulation \eqref{eq:rf_defi} arises as the $B \rightarrow \infty$ limit of \eqref{eq:rfm}; thus, our theory effectively assumes that $B$ is large enough for Monte Carlo effects not to matter. The effects of using a finite $B$ are studied in detail by \citet{mentch2014ensemble}; see also \citet{wager2014confidence}, who recommend taking $B$ on the order of $n$.

\begin{defi}
\label{defi:rf}
The \emph{random forest} with base learner $T$ and subsample size $s$ is
\begin{equation}
\label{eq:rf_defi}
\rf\p{x; \, Z_1, \, ..., \, Z_n} = \binom{n}{s}^{-1} \sum_{1 \leq i_1 < i_2 < ... < i_s \leq n} \EE[\xi \sim \Xi]{T\p{x; \xi, \,  \, Z_{i_1}, \, ..., \, Z_{i_s}}}.
\end{equation}
\end{defi}

Next, as described in Section \ref{sec:main}, we require that the trees $T$ in our forest be honest. Double-sample trees, as defined in Procedure \ref{alg:split}, can always be used to build honest trees with respect to the $\ii$-sample. In the context of causal trees for observational studies, propensity trees (Procedure \ref{alg:prop}) provide a simple recipe for building honest trees without sample splitting.

\begin{defi}
\label{defi:honest}
A tree grown on a training sample $\smash{\p{Z_1 = \p{X_1, \, Y_1}, \, ..., \, Z_s = \p{X_s, \, Y_s}}}$ is \emph{honest} if
(a) (\emph{standard case}) the tree does not use the responses $Y_1, \, ..., \, Y_s$ in choosing where to place its splits; or
(b) (\emph{double sample case}) the tree does not use the $\ii$-sample responses for placing splits.
\end{defi}

In order to guarantee consistency, we also need to enforce that the leaves of the trees become small in \emph{all} dimensions of the feature space as $n$ gets large.\footnote{\citet{biau2012analysis} and \citet{wager2015uniform} consider the estimation of low-dimensional signals embedded in a high-dimensional ambient space using random forests; in this case, the variable selection properties of trees also become important. We leave a study of asymptotic normality of random forests in high dimensions to future work.} Here, we follow \citet{meinshausen2006quantile}, and achieve this effect by enforcing some randomness in the way trees choose the variables they split on: At each step, each variable is selected with probability at least $\pi/d$ for some $0 < \pi \leq 1$ (for example, we could satisfy this condition by completely randomizing the splitting variable with probability $\pi$). Formally, the randomness in how to pick the splitting features is contained in the auxiliary random variable $\xi$.

\begin{defi}
\label{defi:randomsplit}
A tree is a \emph{random-split} tree if at every step of the tree-growing procedure, marginalizing over $\xi$, the probability that the next split occurs along the $j$-th feature is bounded below by $\pi/d$ for some $0 < \pi \leq 1$, for all $j = 1, \, ..., \, d$.
\end{defi}

The remaining definitions are more technical. We use regularity to control the shape of the tree leaves, while symmetry is used to apply classical tools in establishing asymptotic normality.

\begin{defi}
\label{defi:regular}
A tree predictor grown by recursive partitioning is \emph{$\alpha$-regular} for some $\alpha > 0$ if either (a) (\emph{standard case}) each split leaves at least a fraction $\alpha$ of the available training examples on each side of the split and, moreover, the trees are fully grown to depth $k$ for some $k \in \NN$, i.e., there are between $k$ and $2k - 1$ observations in each terminal node of the tree; or (b) (\emph{double sample case}) if the predictor is a double-sample tree as in Procedure 1, the tree satisfies part (a) for the $\ii$ sample.
\end{defi}

\begin{defi}
\label{defi:symmetric}
A predictor is \emph{symmetric} if the (possibly randomized) output of the predictor does not depend on the order ($i = 1, \, 2, \, ...$) in which the training examples are indexed.
\end{defi}

Finally, in the context of classification and regression forests, we estimate the asymptotic variance of random forests using the original infinitesimal jackknife of \citet{wager2014confidence}, i.e.,
\begin{equation}
\label{eq:hvij_plain}
\hVIJ\p{x}  =  \frac{n - 1}{n} \p{\frac{n}{n - s}}^2 \sum_{i = 1}^n \Cov[*]{\hmu^*_b\p{x}, \, N^*_{ib}}^2,
\end{equation}
where $\hmu^*_b(x)$ is the estimate for $\mu(x)$ given by a single regression tree. We note that the finite-sample correction $n (n - 1)/(n - s)^2$ did not appear in \citet{wager2014confidence}, as their paper focused on subsampling with replacement, whereas this correction is only appropriate for subsampling without replacement.

Given these preliminaries, we can state our main result on the asymptotic normality of random forests. As discussed in Section \ref{sec:theory_overview}, we require that the conditional mean function $\mu\p{x} = \EE{Y \cond X = x}$ be Lipschitz continuous. The asymptotic normality result requires for the subsample size $s$ to scale within the bounds given in \eqref{eq:intro_cond}. If the subsample size grows slower than this, the forest will still be asymptotically normal, but the forest may be asymptotically biased. For clarity, we state the following result with notation that makes the dependence of \smash{$\hmu_n(x)$} and $s_n$ on $n$ explicit; in most of the paper, however, we drop the subscripts to \smash{$\hmu_n(x)$} and $s_n$ when there is no risk of confusion.

\begin{theo}
\label{theo:intro}
Suppose that we have $n$ independent and identically distributed training examples $Z_i = \p{X_i, \, Y_i} \in [0, \, 1]^d \times \RR$. Suppose moreover that the features are independently and uniformly distributed\,\footnote{The result also holds with a density that is bounded away from 0 and infinity; however, we assume uniformity for simpler exposition.} $X_i \sim U([0, \, 1]^d)$, that $\mu(x) = \EE{Y \cond X = x}$ and $\mu_2(x) = \EE{Y^2 \cond X = x}$ are Lipschitz-continuous,
\sloppy{and finally that $\Var{Y \cond X = x} > 0$ and
 \smash{$\mathbb{E}[\lvert Y - \mathbb{E}[Y \cond X = x]\rvert^{2 + \delta} \cond X = x] \leq M$}
for some constants $\delta, \, M > 0$, uniformly over all $x \in [0, \, 1]^d$.}
Given this data-generating process, let $T$ be an honest, $\alpha$-regular with $\alpha \leq 0.2$, and symmetric random-split tree in the sense of Definitions \ref{defi:honest}, \ref{defi:randomsplit}, \ref{defi:regular}, and \ref{defi:symmetric}, and let $\hmu_n(x)$ be the estimate for $\mu(x)$ given by a random forest with base learner $T$ and a subsample size $s_n$. Finally, suppose that the subsample size $s_n$ scales as
\begin{equation}
\label{eq:intro_cond}
s_n \asymp n^\beta \ \ \text{for some} \ \ \beta_{\min} := 1 - \p{1 + \frac{d}{\pi} \, \frac{\log\p{\alpha^{-1}}}{\log\p{\p{1 - \alpha}^{-1}}}}^{-1} < \beta < 1.
\end{equation}
Then, random forest predictions are asymptotically Gaussian:
\begin{equation}
\label{eq:intro_gauss}
\frac{\hmu_n(x) -  \mu\p{x}}{\sigma_n(x)}  \Rightarrow \nn\p{0, \, 1} \; \text{for a sequence} \; \sigma_n(x) \rightarrow 0.
\end{equation}
Moreover, the asymptotic variance $\sigma_n$ can be consistently estimated using the infinitesimal jackknife \eqref{eq:hvij}:
\begin{equation}
\label{eq:intro_ij}
\hVIJ\p{x}  \big/ \sigma_n^2(x) \rightarrow_p 1.
\end{equation}
\end{theo}

\begin{rema}[binary classification]
We note that Theorem \ref{theo:intro} also holds for binary classification forests with leaf size $k = 1$, as is default in the \texttt{R}-package \texttt{randomForest} \citep{liaw2002classification}. Here, we treat the output $\rf(x)$ of the random forests as an estimate for the probability $\PP{Y = 1 \cond X = x}$; Theorem \ref{theo:intro} then lets us construct valid confidence intervals for this probability. For classification forests with $k > 1$, the proof of Theorem \ref{theo:intro} still holds if the individual classification trees are built by \emph{averaging} observations within a leaf, but not if they are built by \emph{voting}. Extending our results to voting trees is left as further work.
\end{rema}

The proof of this result is organized as follows. In Section \ref{sec:bias}, we provide bounds for the bias \smash{$\EE{\hmu_n(x) - \mu\p{x}}$} of random forests, while Section \ref{sec:gauss} studies the sampling distributions of \smash{$\hmu_n(x) - \EE{\hmu_n(x)}$} and establishes Gaussianity. Given a subsampling rate satisfying \eqref{eq:intro_cond}, the bias decays faster than the variance, thus allowing for \eqref{eq:intro_gauss}. Before beginning the proof, however, we relate our result to existing results about random forests in Section \ref{sec:theor_relworks}.

\subsection{Theoretical Background}
\label{sec:theor_relworks}

There has been considerable work in understanding the theoretical properties of random forests. The convergence and consistency properties of trees and random forests have been studied by, among others, \citet{biau2012analysis}, \citet{biau2008consistency}, \citet{breiman2004consistency}, \citet{breiman1984classification}, \citet{meinshausen2006quantile}, \citet{scornet2015consistency}, \citet{wager2015uniform}, and \citet{zhu2015reinforcement}. Meanwhile, their sampling variability has been analyzed by \citet{duan2011bootstrap}, \citet{lin2006random}, \citet{mentch2014ensemble}, \citet{sexton2009standard}, and \citet{wager2014confidence}. However, to our knowledge, our Theorem \ref{theo:intro} is the first result establishing conditions under which predictions made by random forests are asymptotically unbiased and normal.

Probably the closest existing result is that of \citet{mentch2014ensemble}, who showed that random forests based on subsampling are asymptotically normal under substantially stronger conditions than us: they require that the subsample size $s$ grows slower than $\sqrt{n}$, i.e., that $s_n/\sqrt{n} \rightarrow 0$. However, under these conditions, random forests will not in general be asymptotically unbiased. As a simple example, suppose that $d = 2$, that $\mu(x) = \Norm{x}_1$, and that we evaluate an honest random forest at $x = 0$.
A quick calculation shows that the bias of the random forest decays as $1/\sqrt{s_n}$, while its variance decays as $s_n/n$. If $s_n/\sqrt{n} \rightarrow 0$, the squared bias decays slower than the variance, and so confidence intervals built using the resulting Gaussian limit distribution will not cover $\mu(x)$.
Thus, although the result of \citet{mentch2014ensemble} may appear qualitatively similar to ours, it cannot be used for valid asymptotic statistical inference about $\mu(x)$.

The variance estimator $\smash{\hVIJ}$ was studied in the context of random forests by \citet{wager2014confidence}, who showed empirically that the method worked well for many problems of interest. \citet{wager2014confidence} also emphasized that, when using $\smash{\hVIJ}$ in practice, it is important to account for Monte Carlo bias.
Our analysis provides theoretical backing to these results, by showing that $\smash{\hVIJ}$ is in fact a consistent estimate for the variance $\sigma_n^2(x)$ of random forest predictions. The earlier work on this topic \citep{efron2013estimation,wager2014confidence} had only motivated the estimator $\smash{\hVIJ}$ by highlighting connections to classical statistical ideas, but did not establish any formal justification for it.

Instead of using subsampling, Breiman originally described random forests in terms of bootstrap sampling, or bagging \citep{breiman1996bagging}.
Random forests with bagging, however, have proven to be remarkably resistant to classical statistical analysis. As observed by \citet{buja2006observations}, \citet{chen2003effects}, \citet{friedman2007bagging} and others, estimators of this form can exhibit surprising properties even in simple situations; meanwhile, using subsampling rather than bootstrap sampling has been found to avoid several pitfalls \citep[e.g.,][]{politis1999subsampling}.
Although they are less common in the literature, random forests based on subsampling have also been occasionally studied and found to have good practical and theoretical properties \citep[e.g.,][]{buhlmann2002analyzing,mentch2014ensemble,scornet2015consistency,strobl2007bias}.

Finally, an interesting question for further theoretical study is to understand the optimal scaling of the subsample size $s_n$ for minimizing the mean-squared error of random forests. For subsampled nearest-neighbors estimation, the optimal rate for $s_n$ is \smash{$s_n \asymp n^{1 - \p{1 + d/4}^{-1}}$} \citep{biau2010rate,samworth2012optimal}. Here, our specific value for $\beta_{\min}$ depends on the upper bounds for bias developed in the following section. Now, as shown by \citet{biau2012analysis}, under some sparsity assumptions on $\mu(x)$, it is possible to get substantially stronger bounds for the bias of random forests; thus, it is plausible that under similar conditions we could push back the lower bound $\beta_{\min}$ on the growth rate of the subsample size.

\subsection{Bias and Honesty}
\label{sec:bias}

We start by bounding the bias of regression trees. Our approach relies on showing that as the sample size $s$ available to the tree gets large, its leaves get small; Lipschitz-continuity of the conditional mean function and honesty then let us bound the bias. In order to state a formal result, define the \emph{diameter} $\diam(L(x))$ of a leaf $L(x)$ as the length of the longest segment contained inside $L(x)$, and similarly let $\diam_j(L(x))$ denote the length of the longest such segment that is parallel to the $j$-th axis. The following lemma is a refinement of a result of \citet{meinshausen2006quantile}, who showed that $\diam(L(x)) \rightarrow_p 0$ for regular trees.

\begin{lemm}
\label{lemm:diameter}
Let $T$ be a regular, random-split tree and let $L(x)$ denote its leaf containing $x$. Suppose that $X_1, \, ..., \, X_s \sim U\p{[0, \, 1]^d}$ independently. Then, for any $0 < \eta < 1$, and for large enough $s$,
$$\PP{\diam_j \p{L(x)} \geq \p{\frac{s}{2k - 1}}^{-\frac{0.99 \, \p{1 - \eta} \log\p{\p{1 - \alpha}^{-1}}}{\log\p{\alpha^{-1}}} \, \frac{\pi}{d}}} \leq \p{\frac{s}{2k - 1}}^{-\frac{\eta^2}{2} \frac{1}{\log\p{\alpha^{-1}}} \, \frac{\pi}{d}}.$$
\end{lemm}

This lemma then directly translates into a bound on the bias of a single regression tree. Since a forest is an average of independently-generated trees, the bias of the forest is the same as the bias of a single tree.

\begin{theo}
\label{theo:bias}
Under the conditions of Lemma \ref{lemm:diameter}, suppose moreover that $\mu\p{x}$ is Lipschitz continuous and that the trees $T$ in the random forest are honest. Then, provided that $\alpha \leq 0.2$, the bias of the random forest at $x$ is bounded by
$$ \abs{\EE{\hmu\p{x}} - \mu\p{x}} = \oo\p{s^{-\frac{1}{2} \frac{\log\p{\p{1 - \alpha}^{-1}}}{\log\p{\alpha^{-1}}} \frac{\pi}{d}}}; $$
the constant in the $\oo$-bound is given in the proof.
\end{theo}

\subsection{Asymptotic Normality of Random Forests}
\label{sec:gauss}

Our analysis of the asymptotic normality of random forests builds on ideas developed by \citet{hoeffding1948class} and \citet{hajek1968asymptotic} for understanding classical statistical estimators such as $U$-statistics. We begin by briefly reviewing their results to give some context to our proof. Given a predictor $T$ and independent training examples $Z_1$, ..., $Z_n$, the H\'ajek projection of $T$ is defined as
\begin{equation}
\label{eq:hajek}
\proj{T} = \EE{T} + \sum_{i = 1}^n \p{\EE{T \cond Z_i} - \EE{T}}.
\end{equation}
In other words, the H\'ajek projection of $T$ captures the first-order effects in $T$.
Classical results imply that $\smash{\Var{\proj{T}} \leq \Var{T}}$, and further:
\begin{equation}
\label{eq:hajek_l2}
\limn {\Var{\proj{T}}} \big/ {\Var{T}} = 1 \text{ implies that } \limn {\EE{\Norm{\proj{T} - T}_2^2}} \big/ {\Var{T}} = 0. 
\end{equation}
Since the H\'ajek projection $\proj{T}$ is a sum of independent random variables, we should expect it to be asymptotically normal under weak conditions. Thus whenever the ratio of the variance of $\proj{T}$ to that of $T$ tends to 1, the theory of H\'ajek projections almost automatically guarantees that $T$ will be asymptotically normal.\footnote{The moments defined in \eqref{eq:hajek} depend on the data-generating process for the $Z_i$, and so cannot be observed in practice. Thus, the H\'ajek projection is mostly useful as an abstract theoretical tool. For a review of classical projection arguments, see Chapter 11 of \citet{van2000asymptotic}.}

If $T$ is a regression tree, however, the condition from \eqref{eq:hajek_l2} does not apply, and we cannot use the classical theory of H\'ajek projections directly. Our analysis is centered around a weaker form of this condition, which we call $\nu$-incrementality. With our definition, predictors $T$ to which we can apply the argument \eqref{eq:hajek_l2} directly are 1-incremental.

\begin{defi}
\label{defi:incr}
The predictor $T$ is $\nu(s)$-incremental at $x$ if
$$ {\Var{\proj{T}\p{x; \, Z_1,\, ..., \, Z_s}}} \big / {\Var{T\p{x; \, Z_1,\, ..., \, Z_s}}} \gtrsim \nu(s), $$
where $\proj{T}$ is the H\'ajek projection of $T$ \eqref{eq:hajek}. In our notation,
$$f(s) \gtrsim g(s) \text{ means that } \liminf_{s \rightarrow \infty} {f(s)} \big/ {g(s)} \geq 1. $$
\end{defi}

Our argument proceeds in two steps. First we establish lower bounds for the incrementality of regression trees in Section \ref{sec:incremental}. Then, in Section \ref{sec:hajek} we show how we can turn weakly incremental predictors $T$ into 1-incremental ensembles by subsampling (Lemma \ref{lemm:hajek}), thus bringing us back into the realm of classical theory. We also establish the consistency of the infinitesimal jackknife for random forests.
Our analysis of regression trees is motivated by the ``potential nearest neighbors'' model for random forests introduced by \citet{lin2006random}; the key technical device used in Section \ref{sec:hajek} is the ANOVA decomposition of \citet{efron1981jackknife}. The discussion of the infinitesimal jackknife for random forest builds on results of \citet{efron2013estimation} and \citet{wager2014confidence}.

\subsubsection{Regression Trees and Incremental Predictors}
\label{sec:incremental}

Analyzing specific greedy tree models such as CART trees can be challenging. We thus follow the lead of \citet{lin2006random}, and analyze a more general class of predictors---potential nearest neighbors predictors---that operate by doing a nearest-neighbor search over rectangles; see also \citet{biau2010layered}. The study of potential (or layered) nearest neighbors goes back at least to \citet{barndorff1966distribution}.

\begin{defi}
Consider a set of points $X_1, \, ..., \, X_s \in \RR^d$ and a fixed $x \in \RR^d$. A point $X_i$ is a \emph{potential nearest neighbor} (PNN) of $x$ if the smallest axis-aligned hyperrectangle with vertices $x$ and $X_i$ contains no other points $X_j$. Extending this notion, a \emph{PNN $k$-set} of $x$ is a set of points $\Lambda \subseteq \cb{X_1, \, ..., \, X_s}$ of size $k \leq |L| < 2k-1$ such that there exists an axis aligned hyperrectangle $L$ containing $x$, $\Lambda$, and no other training points. A training example $X_i$ is called a $k$-PNN of $x$ if there exists a PNN $k$-set of $x$ containing $X_i$.
Finally, a predictor $T$ is a \emph{$k$-PNN predictor} over $\cb{Z}$ if, given a training set
$$\cb{Z} = \cb{\p{X_1, \, Y_1}, \, ..., \, \p{X_s, \, Y_s}} \in \cb{\RR^d \times \yy}^s $$
and a test point $x \in \RR^d$, $T$ always outputs the average of the responses $Y_i$ over a $k$-PNN set of $x$.
\end{defi}

This formalism allows us to describe a wide variety of tree predictors. For example, as shown by \citet{lin2006random},
any decision tree $T$ that makes axis-aligned splits and has leaves of size between $k$ and $2k - 1$ is a $k$-PNN predictor. In particular, the base learners originally used by \citet{breiman2001random}, namely CART trees grown up to a leaf size $k$ \citep{breiman1984classification}, are $k$-PNN predictors.
Predictions made by $k$-PNN predictors can always be written as
\begin{equation}
\label{eq:pnn_sum}
T\p{x; \, \xi, \, Z_1, \, ..., \, Z_s} = \sum_{i = 1}^s S_i Y_i,
\end{equation}
where $S_i$ is a selection variable that takes the value $1/|\{i: X_i \in L(x)\}|$ for indices $i$ in the selected leaf-set $L(x)$ and 0 for all other indices. If the tree is honest, we know in addition that, for each $i$, $S_i$ is independent of $Y_i$ conditional on $X_i$.

An important property of $k$-PNN predictors is that we can often get a good idea about whether $S_i$ is non-zero even if we only get to see $Z_i$; more formally, as we show below, the quantity $s\Var{\EE{S_1\cond Z_1}}$ cannot get too small. Establishing this fact is a key step in showing that $k$-PNNs are incremental.
In the following result, $T$ can be an arbitrary symmetric $k$-PNN predictor.

\begin{lemm}
\label{lemm:pnn}
Suppose that the observations $X_1, \, X_2, \, \ldots$ are independent and identically distributed on $[0, \, 1]^d$ with a density $f$ that is bounded away from infinity, and let $T$ be any symmetric $k$-PNN predictor.
Then, there is a constant $C_{f, \, d}$ depending only on $f$ and $d$ such that, as $s$ gets large,
\begin{equation}
\label{eq:pnn_bound}
s\Var{\EE{S_1 \cond Z_1}} \gtrsim \frac{1}{k} \, C_{f, \, d} \big/ \log\p{s}^{d},
\end{equation}
where $S_i$ is the indicator for whether the observation is selected in the subsample.
When $f$ is uniform over $[0, \, 1]^d$, the bound holds with $C_{f, \, d} = 2^{-(d + 1)}\p{d - 1}!$.
\end{lemm}

When $k = 1$ we see that, marginally, $S_1 \sim \text{Bernoulli}(1/s)$ and so $s\Var{S_1} \sim 1$;
more generally, a similar calculation shows that $1/(2k - 1) \lesssim s\Var{S_1} \lesssim 1/k$.
Thus, \eqref{eq:pnn_bound} can be interpreted as a lower bound on how much
information $Z_1$ contains about the selection event $S_1$.

Thanks to this result, we are now ready to show that all honest and regular random-split trees are incremental. Notice that any symmetric $k$-regular tree following Definition \ref{defi:regular} is also a symmetric $k$-PNN predictor.

\begin{theo}
\label{theo:pnn}
Suppose that the conditions of Lemma \ref{lemm:pnn} hold and that $T$ is an honest $k$-regular symmetric tree in the sense of Definitions \ref{defi:honest} (part a), \ref{defi:regular} (part a), and \ref{defi:symmetric}. 
Suppose moreover that the conditional moments $\mu\p{x}$ and $\mu_2(x)$ are both Lipschitz continuous at $x$.
Finally, suppose that $\Var{Y \cond X = x} > 0$.
Then $T$ is $\nu\p{s}$-incremental at $x$ with
\begin{equation}
\label{eq:pnn_inc}
\nu\p{s} = {C_{f, \, d}} \big / {\log\p{s}^{d}},
\end{equation}
where $C_{f, \, d}$ is the constant from Lemma \ref{lemm:pnn}.
\end{theo}

Finally, the result of Theorem \ref{theo:pnn} also holds for double-sample trees of the form described in Procedure \ref{alg:split}. To establish the following result, we note that a double-sample tree is an honest, symmetric $k$-PNN predictor with respect to the $\ii$-sample, while all the data in the $\jj$-sample can be folded into the auxiliary noise term $\xi$; the details are worked out in the proof.

\begin{coro}
\label{coro:doublesample}
Under the conditions of Theorem \ref{theo:pnn}, suppose that $T$ is instead a double-sample tree (Procedure \ref{alg:split}) satisfying Definitions \ref{defi:honest} (part b), \ref{defi:regular} (part b), and \ref{defi:symmetric}. Then, $T$ is $\nu$-incremental, with $\nu\p{s} =  {C_{f, \, d}} / (4\log\p{s}^{d})$.
\end{coro}

\subsubsection{Subsampling Incremental Base Learners}
\label{sec:hajek}

In the previous section, we showed that decision trees are $\nu$-incremental, in that the H\'ajek projection $\proj{T}$ of $T$ preserves at least some of the variation of $T$. In this section, we show that randomly subsampling $\nu$-incremental predictors makes them 1-incremental; this then lets us proceed with a classical statistical analysis. The following lemma, which flows directly from the ANOVA decomposition of \citet{efron1981jackknife}, provides a first motivating result for our analysis.

\begin{lemm}
\label{lemm:hajek}
Let $\hmu(x)$ be the estimate for $\mu(x)$ generated by a random forest with base learner $T$ as defined in \eqref{eq:rf_defi}, and let $\proj{\hmu}$ be the H\'ajek projection of $\hmu$ \eqref{eq:hajek}. Then
$$\EE{\p{\hmu\p{x} - \proj{\hmu}\p{x}}^2} \leq \p{\frac{s}{n}}^2 \, \Var{T\p{x; \, \xi,  \, Z_1, \, ..., \, Z_s}} $$
whenever the variance $\Var{T}$ of the base learner is finite.
\end{lemm}

This technical result paired with Theorem \ref{theo:pnn} or Corollary \ref{coro:doublesample} leads to an asymptotic Gaussianity result; from a technical point of view, it suffices to check Lyapunov-style conditions for the central limit theorem.

\begin{theo}
\label{theo:gauss}
Let $\hmu(x)$ be a random forest estimator trained according the conditions of Theorem \ref{theo:pnn} or Corollary \ref{coro:doublesample}. Suppose, moreover, that the subsample size $s_n$ satisfies
$$ \limn s_n = \infty \eqand \limn {s_n \log\p{n}^{d}} \big/ {n} = 0, $$
and that  \smash{$\mathbb{E}[\lvert Y - \mathbb{E}[Y \cond X = x]\rvert^{2 + \delta} \cond X = x] \leq M$}
for some constants $\delta, \, M > 0$, uniformly over all $x \in [0, \, 1]^d$.
Then, there exists a sequence $\sigma_n(x) \rightarrow 0$ such that
\begin{equation}
\label{eq:gauss}
 \frac{\hmu_n\p{x} - \EE{\hmu_n\p{x}}}{\sigma_n(x)} \Rightarrow \nn\p{0, \, 1},
\end{equation}
where $\nn\p{0, \, 1}$ is the standard normal distribution.
\end{theo}

Moreover, as we show below, it is possible to accurately estimate the variance of a random forest using the infinitesimal jackknife for random forests \citep{efron2013estimation,wager2014confidence}.

\begin{theo}
\label{theo:ij}
Let $\hVIJ\p{x;, \, Z_1,\, ...,\, Z_n}$ be the infinitesimal jackknife for random forests as defined in \eqref{eq:hvij}. Then, under the conditions of Theorem \ref{theo:gauss},
\begin{equation}
\label{eq:ij_consistent}
\hVIJ\p{x; \, Z_1,\, ...,\, Z_n} \big/ \sigma_n^2(x) \rightarrow_p 1.
\end{equation}
\end{theo}

Finally, we end this section by motivating the finite sample correction $n (n - 1)/(n - s)^2$ appearing in \eqref{eq:hvij_plain} by considering the simple case where we have trivial trees that do not make any splits: $T(x; \, \xi, \, Z_{i_1}, \, ..., \, Z_{i_s}) = s^{-1} \sum_{j = 1}^s Y_{i_j}$. In this case, we can verify that the full random forest is nothing but $\hmu = n^{-1} \sum_{i = 1}^n Y_i$, and the standard variance estimator
$$ \hV_{simple} = \frac{1}{n \, \p{n - 1}} \sum_{i = 1}^n \p{Y_i - \bY}^2, \ \ \bY = \frac{1}{n} \sum_{i = 1}^n Y_i $$
is well-known to be unbiased for $\smash{\Var{\hmu}}$. We show below that, for trivial trees \smash{$\hVIJ = \hV_{simple}$},  implying that our correction makes \smash{$\hVIJ$} exactly unbiased in finite samples for trivial trees.
Of course, $n (n - 1)/(n - s)^2 \rightarrow 1$, and so Theorem \ref{theo:ij} would hold even without this finite-sample correction; however, we find it to substantially improve the performance of our method in practice.

\begin{prop}
\label{prop:finite_sample}
For trivial trees \smash{$T(x; \, \xi, \, Z_{i_1}, \, ..., \, Z_{i_s}) = s^{-1} \sum_{j = 1}^s Y_{i_j}$}, the variance estimate \smash{$\hVIJ$} \eqref{eq:hvij_plain} is equivalent to the standard variance estimator $\hV_{simple}$, and \smash{$\mathbb{E}[\hVIJ] = \Var{\hmu}$}.
\end{prop}

\section{Inferring Heterogeneous Treatment Effects}
\label{sec:treat_eff}

We now return to our main topic, namely estimating heterogeneous treatment effects using random forests in the potential outcomes framework with unconfoundedness, and adapt our asymptotic theory for regression forests to the setting of causal inference.
Here, we again work with training data consisting of tuples $Z_i = \p{X_i, \, Y_i, \, W_i}$ for $i = 1, \, ..., \, n$, where $X_i$ is a feature vector,  $Y_i$ is the response, and $W_i$ is the treatment assignment. Our goal is to estimate the conditional average treatment effect $\tau(x) = \EE{Y^{(1)} - Y^{(0)} \cond X = x}$ at a pre-specified test point $x$. By analogy to Definition \ref{defi:rf}, we build our causal forest $\cf$ by averaging estimates for $\tau$ obtained by training causal trees $\Gamma$ over subsamples:
\begin{equation}
\label{eq:rf_cmp}
\cf\p{x; \, Z_1, \, ..., \, Z_n} = \binom{n}{s}^{-1} \sum_{1 \leq i_1 < i_2 < ... < i_s \leq n} \EE[\xi \sim \Xi]{\Gamma\p{x; \, \xi, \, Z_{i_1}, \, ..., \, Z_{i_s}}}.
\end{equation}
We seek an analogue to Theorem \ref{theo:intro} for such causal forests.

Most of the definitions used to state Theorem \ref{theo:intro} apply directly to this context; however, the notions of honesty and regularity need to be adapted slightly. Specifically, an honest causal tree is not allowed to look at the responses $Y_i$ when making splits but can look at the treatment assignments $W_i$. Meanwhile, a regular causal tree must have at least $k$ examples from both treatment classes in each leaf; in other words, regular causal trees seek to act as fully grown trees for the rare treatment assignment, while allowing for more instances of the common treatment assignment.

\setcounterref{deficmp}{defi:honest}
\addtocounter{deficmp}{-1}

\begin{deficmp}
\label{defi:honest_cmp}
\sloppy{
A causal tree grown on a training sample \smash{$(Z_1 = \p{X_1, \, Y_1, \, W_1}$}, ..., 
\smash{$Z_s = \p{X_s, \, Y_s, \, W_s})$} is \emph{honest} if
(a) (\emph{standard case}) the tree does not use the responses $Y_1, \, ..., \, Y_s$ in choosing where to place its splits; or
(b) (\emph{double sample case}) the tree does not use the $\ii$-sample responses for placing splits.
}
\end{deficmp}

\setcounterref{deficmp}{defi:regular}
\addtocounter{deficmp}{-1}

\begin{deficmp}
\label{defi:regular_cmp}
A causal tree grown by recursive partitioning is \emph{$\alpha$-regular}  at $x$ for some $\alpha > 0$ if either:  (a) (\emph{standard case}) (1) Each split leaves at least a fraction $\alpha$ of the available training examples on each side of the split, (2) The leaf containing $x$ has at least $k$ observations from each treatment group ($W_i \in \cb{0, \, 1}$) for some $k \in \NN$, and (3) The leaf containing $x$ has either less than $2k - 1$ observations with $W_i = 0$ or $2k - 1$ observations with $W_i = 1$; or (b) 
(\emph{double-sample case}) for a double-sample tree as defined in Procedure 1, (a) holds for the $\ii$ sample.
\end{deficmp}

Given these assumptions, we show a close analogue to Theorem \ref{theo:intro}, given below.
The main difference relative to our first result about regression forests is that we now rely on unconfoundedness and overlap to achieve consistent estimation of $\tau(x)$.
To see how these assumptions enter the proof, recall that an honest causal tree uses the features $X_i$ and the treatment assignments $W_i$ in choosing where to place its splits, but not the responses $Y_i$.
Writing $\ii^{(1)}(x)$ and $\ii^{(0)}(x)$ for the indices of the treatment and control units in the leaf around $x$, we then find that after the splitting stage
\begin{align}
\label{eq:unconf_use}
&\EE{\Gamma\p{x} \cond X, \, W} \\
\notag
&\ \ \ \ \ \ = \frac{\sum_{\cb{i \in \ii^{(1)}(x)}} \EE{Y^{(1)} \cond X = X_i, \, W = 1}}{\abs {\ii^{(1)}(x)}}  - \frac{\sum_{\cb{i \in \ii^{(0)}(x)}} \EE{Y^{(0)} \cond X = X_i, \, W = 0}}{\abs {\ii^{(0)}(x)}} \\
\notag
&\ \ \ \ \ \ = \frac{\sum_{\cb{i \in \ii^{(1)}(x)}} \EE{Y^{(1)} \cond X = X_i}}{\abs {\ii^{(1)}(x)}}  - \frac{\sum_{\cb{i \in \ii^{(0)}(x)}} \EE{Y^{(0)} \cond X = X_i}}{\abs {\ii^{(0)}(x)}},
\end{align}
where the second equality follows by unconfoundedness \eqref{eq:unconf}.
Thus, it suffices to show that the two above terms are consistent for estimating \smash{$\EE{Y^{(0)} \cond X = x}$} and \smash{$\EE{Y^{(1)} \cond X = x}$}.
To do so, we can essentially emulate the argument leading to Theorem \ref{theo:intro}, provided we can establish an analogue to Lemma \ref{lemm:diameter} and give a fast enough decaying upper bound to the diameter of $L(x)$; this is where we need the overlap assumption.
A proof of Theorem \ref{theo:cmp_forest} is given in the appendix.

\begin{theo}
\label{theo:cmp_forest}
Suppose that we have $n$ independent and identically distributed training examples $Z_i = \p{X_i, \, Y_i, \, W_i} \in [0, \, 1]^d \times \RR \times \cb{0, \, 1}$. Suppose, moreover, that the treatment assignment is unconfounded \eqref{eq:unconf} and has overlap \eqref{eq:overlap}. Finally, suppose that both potential outcome distributions \smash{$(X_i, \, Y_i^{(0)})$} and \smash{$(X_i, \, Y_i^{(1)})$} satisfy the same regularity assumptions as the pair \smash{$(X_i, \, Y_i)$} did in the statement of Theorem \ref{theo:intro}.
Given this data-generating process, let $\Gamma$ be an honest, $\alpha$-regular with $\alpha \leq 0.2$, and symmetric random-split causal forest in the sense of Definitions \ref{defi:honest_cmp}, \ref{defi:randomsplit}, \ref{defi:regular_cmp}, and \ref{defi:symmetric}, and let $\htau(x)$ be the estimate for $\tau(x)$ given by a causal forest with base learner $\Gamma$ and a subsample size $s_n$ scaling as in \eqref{eq:intro_cond}.
Then, the predictions $\htau(x)$ are consistent and asymptotically both Gaussian and centered, and the variance of the causal forest can be consistently estimated using the infinitesimal jackknife for random forests, i.e., \eqref{eq:causal_gauss} holds.
\end{theo}

\begin{rema}(testing at many points)
\label{rema:regular_cmp}
We note that it is not in general possible to construct causal trees that are regular in the sense of Definition \ref{defi:regular_cmp} for all $x$ simultaneously. As a simple example, consider the situation where $d = 1$, and $W_i = 1\p{\cb{X_i \geq 0}}$; then, the tree can have at most 1 leaf for which it is regular. In the proof of Theorem \ref{theo:cmp_forest}, we avoided this issue by only considering a single test point $x$, as it is always possible to build a tree that is regular at a single given point $x$.
In practice, if we want to build a causal tree that can be used to predict at many test points, we may need to assign different trees to be valid for different test points. Then, when predicting at a specific $x$, we treat the set of trees that were assigned to be valid at that $x$ as the relevant forest and apply Theorem \ref{theo:cmp_forest} to it.
\end{rema}

\section{Simulation Experiments}
\label{sec:simu}

In observational studies, accurate estimation of heterogeneous treatment effects requires overcoming two potential sources of bias. First, we need to identify neighborhoods over which the actual treatment effect $\tau(x)$ is reasonably stable and, second, we need to make sure that we are not biased by varying sampling propensities $e(x)$. The simulations here aim to test the ability of causal forests to respond to both of these factors.

Since causal forests are adaptive nearest neighbor estimators, it is natural to use a non-adaptive nearest neighborhood method as our baseline. We compare our method to the standard $k$ nearest neighbors ($k$-NN) matching procedure, which estimates the treatment effect as
\begin{equation}
\label{eq:knn}
\htau_{KNN}(x) = \frac{1}{k} \sum_{i \in \set_1\p{x}} Y_i - \frac{1}{k} \sum_{i \in \set_0\p{x}} Y_i,
\end{equation}
where $\set_1$ and $\set_0$ are the $k$ nearest neighbors to $x$ in the treatment ($W = 1$) and control ($W = 0$) samples respectively. We generate confidence intervals for the $k$-NN method by modeling $\htau_{KNN}(x)$ as Gaussian with mean $\tau(x)$ and variance \smash{$(\hV\p{\set_0} + \hV\p{\set_1}) / (k \, (k - 1))$}, where \smash{$\hV\p{\set_{0/1}}$} is the sample variance for $\set_{0/1}$.

The goal of this simulation study is to verify that forest-based methods
can be used build rigorous, asymptotically valid confidence intervals that improve over non-adaptive
methods like $k$-NN in finite samples. The fact that forest-based methods hold promise for treatment
effect estimation in terms of predictive error has already been conclusively
established elsewhere; for example, BART methods following
\citet{hill2011bayesian} won the recent Causal Inference Data Analysis Challenge at the 2016 Atlantic Causal
Inference Conference. We hope that the conceptual tools developed in this paper will prove to be helpful
in analyzing a wide variety of forest-based methods.

\subsection{Experimental Setup}

We describe our experiments in terms of the sample size $n$, the ambient dimension $d$, as well as the following functions:
\begin{align*}
&\text{main effect: } m(x) = 2^{-1} \, \EE{Y^{(0)} + Y^{(1)} \cond X = x}, \\
&\text{treatment effect: } \tau(x) = \EE{Y^{(1)} - Y^{(0)} \cond X = x}, \\
&\text{treatment propensity: } e(x) = \PP{W = 1 \cond X = x}.
\end{align*}
In all our examples, we respect unconfoundedness \eqref{eq:unconf},  use $X \sim U([0, \, 1]^d)$, and have homoscedastic noise \smash{$Y^{(0/1)} \sim \nn\p{\mathbb{E}[Y^{(0/1)} \cond X], \, 1}$}.
We evaluate performance in terms of expected mean-squared error for estimating $\tau(X)$ at a random test example $X$, as well as expected coverage of $\tau(X)$ with a target coverage rate of 0.95.

In our first experiment, we held the treatment effect fixed at $\tau(x)=0$, and tested the ability of our method to resist bias due to an interaction between $e(x)$ and $m(x)$.
This experiment is intended to emulate the problem that in observational studies, a treatment assignment is often correlated with potential outcomes, creating bias unless the statistical method accurately adjusts for covariates.  $k$-NN matching is a popular approach for performing this adjustment in practice.
Here, we set
\begin{equation}
\label{eq:prop_setup}
e(X) = \frac{1}{4} \p{1 + \beta_{2, \, 4}(X_1)}, \ \ m(X) = 2X_1 - 1,
\end{equation}
where $\beta_{a, \, b}$ is the $\beta$-density with shape parameters $a$ and $b$.
We used $n = 500$ samples and varied $d$ between 2 and 30. Since our goal is accurate propensity matching, we use propensity trees (Procedure \ref{alg:prop}) as our base learner; we grew $B = 1000$ trees with $s = 50$.

For our second experiment, we evaluated the ability of causal forests to adapt to heterogeneity in $\tau(x)$, while holding $m(x) = 0$ and $e(x) = 0.5$ fixed. Thanks to unconfoundedness, the fact that $e(x)$ is constant means that we are in a randomized experiment. We set $\tau$ to be a smooth function supported on the first two features:
\begin{equation}
\label{eq:tau0_setup}
\tau\p{X} = \varsigma\p{X_1} \, \varsigma\p{X_2}, \ \ \ \varsigma\p{x} = 1 + \frac{1}{1 + e^{-20\p{x - 1/3}}}.
\end{equation}
We took $n = 5000$ samples, while varying the ambient dimension $d$ from 2 to 8. For causal forests, we used double-sample trees with the splitting rule of \citet{athey2015machine} as our base learner (Procedure \ref{alg:split}). We used $s = 2500$ (i.e., $\abs{\ii} = 1250$) and grew $B = 2000$ trees.

One weakness of nearest neighbor approaches in general, and random forests in particular, is that they can fill the valleys and flatten the peaks of the true $\tau(x)$ function, especially near the edge of feature space.
We demonstrate this effect using an example similar to the one studied above, except now $\tau(x)$ has a sharper spike in the $x_1, \, x_2 \approx 1$ region:
\begin{equation}
\label{eq:tau_setup}
\tau\p{X} = \varsigma\p{X_1} \, \varsigma\p{X_2}, \ \ \ \varsigma\p{x} = \frac{2}{1 + e^{-12 \p{x - 1/2}}}.
\end{equation}
We used the same training method as with \eqref{eq:tau0_setup}, except with $n = 10000$, $s = 2000$, and $B = 10000$.

We implemented our simulations in \texttt{R}, using the packages 
\texttt{causalTree} \citep{athey2015machine} for building individual trees,
\texttt{randomForestCI} \citep{wager2014confidence} for computing \smash{$\hVIJ$},
and \texttt{FNN} \citep{beygelzimer2013fnn} for $k$-NN regression.
All our trees had a minimum leaf size of $k = 1$.
Software replicating the above simulations is available from the
authors.\footnote{The \texttt{R} package \texttt{grf} \citep{athey2016generalized} provides a
newer, high-performance implementation of causal forests, available on \texttt{CRAN}.}

\subsection{Results}

\begin{table}[t]
\centering
\begin{tabular}{|r|ccc|ccc|}
\hline
 & \multicolumn{3}{c|}{mean-squared error} & \multicolumn{3}{c|}{coverage} \\
 \hline
$d$ & CF & 10-NN & 100-NN & CF & 10-NN & 100-NN \\ 
  \hline
  2 & 0.02 (0) & 0.21 (0) & 0.09 (0) & 0.95 (0) & 0.93 (0) & 0.62 (1) \\ 
  5 & 0.02 (0) & 0.24 (0) & 0.12 (0) & 0.94 (1) & 0.92 (0) & 0.52 (1) \\ 
  10 & 0.02 (0) & 0.28 (0) & 0.12 (0) & 0.94 (1) & 0.91 (0) & 0.51 (1) \\ 
  15 & 0.02 (0) & 0.31 (0) & 0.13 (0) & 0.91 (1) & 0.90 (0) & 0.48 (1) \\ 
  20 & 0.02 (0) & 0.32 (0) & 0.13 (0) & 0.88 (1) & 0.89 (0) & 0.49 (1) \\ 
  30 & 0.02 (0) & 0.33 (0) & 0.13 (0) & 0.85 (1) & 0.89 (0) & 0.48 (1) \\ 
   \hline
\end{tabular}
\caption{Comparison of the performance of a causal forests (CF) with that of the $k$-nearest neighbors ($k$-NN) estimator with $k = 10,\, 100$, on the setup \eqref{eq:prop_setup}. The numbers in parentheses indicate the (rounded) standard sampling error for the last printed digit, obtained by aggregating performance over 500 simulation replications.}
\label{tab:prop_simu}
\end{table}

In our first setup \eqref{eq:prop_setup}, causal forests present a striking improvement over $k$-NN matching; see Table \ref{tab:prop_simu}. Causal forests succeed in maintaining a mean-squared error of 0.02 as $d$ grows from 2 to 30, while 10-NN and 100-NN do an order of magnitude worse. We note that the noise of $k$-NN due to variance in $Y$ after conditioning on $X$ and $W$ is already $2/k$, implying that $k$-NN with $k \leq 100$ cannot hope to match the performance of causal forests. Here, however, 100-NN is overwhelmed by bias, even with $d = 2$.
Meanwhile, in terms of uncertainty quantification, our method achieves nominal coverage up to $d = 10$, after which the performance of the confidence intervals starts to decay.
The 10-NN method also achieves decent coverage; however, its confidence intervals are much wider than ours as evidenced by the mean-squared error.

\begin{figure}[t]
\centering
\begin{tabular}{ccc}
\includegraphics[width=0.3\textwidth]{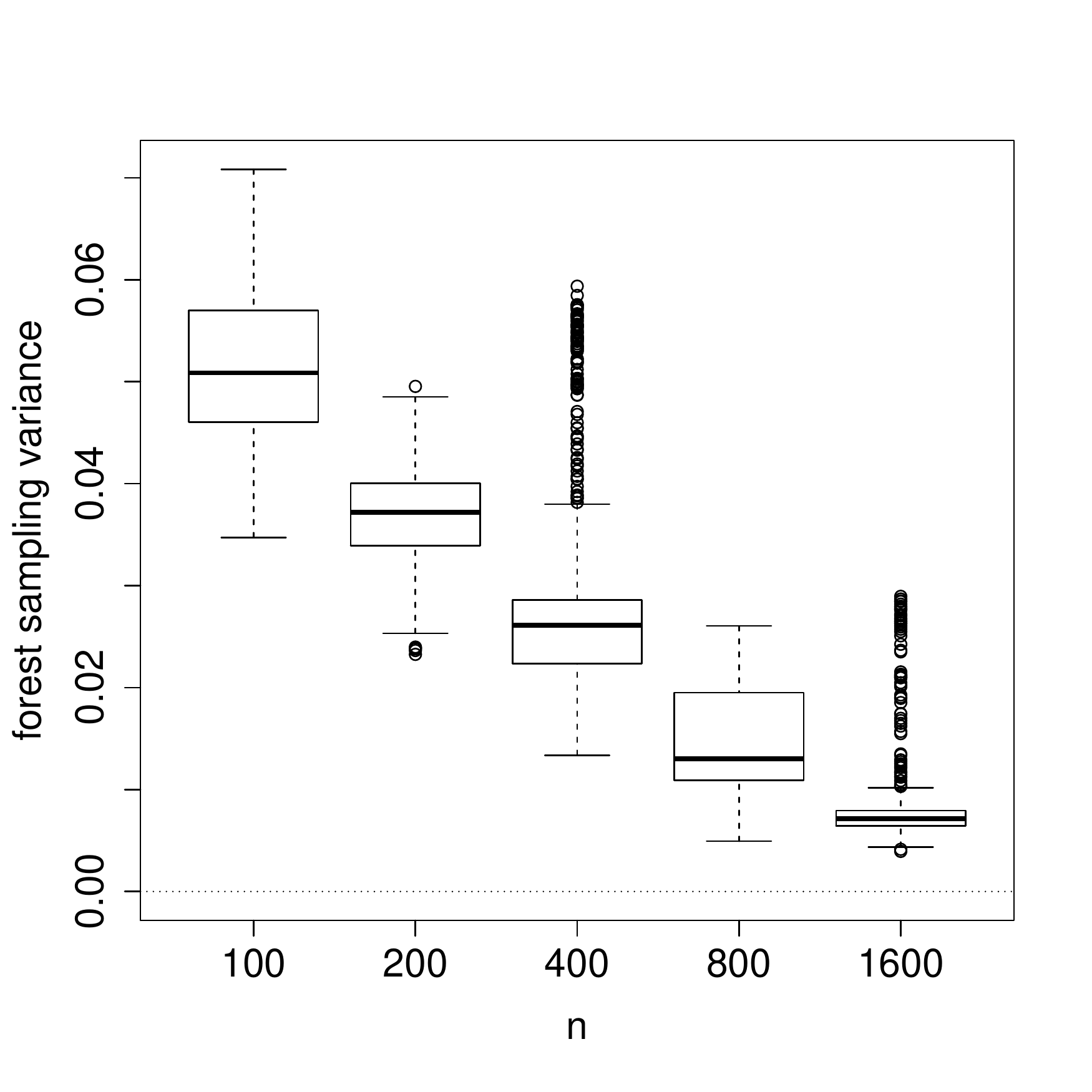} &
\includegraphics[width=0.3\textwidth]{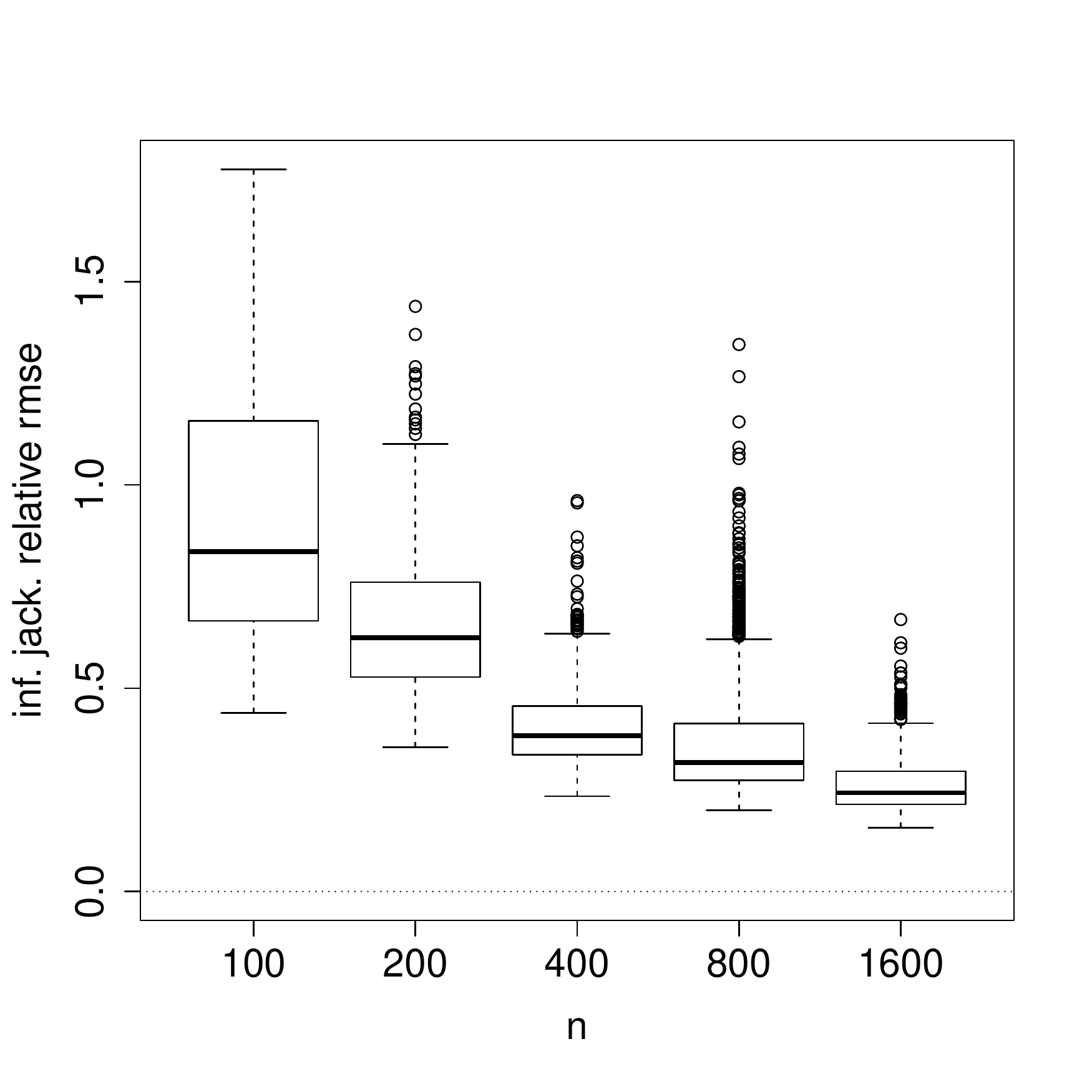} &
\includegraphics[width=0.3\textwidth]{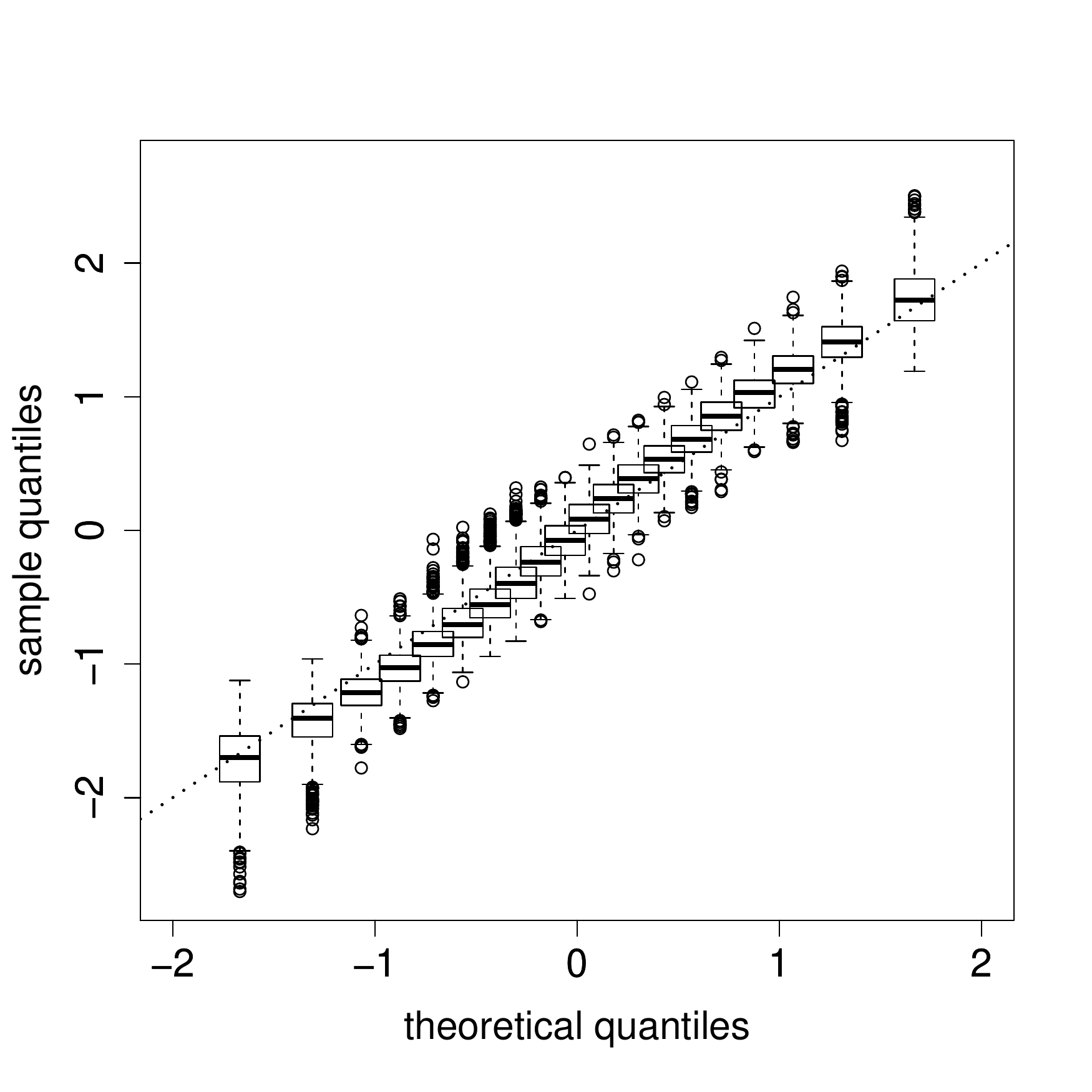} \\
Causal forest variance & Inf.~jack.~relative RMSE & Prediction QQ-plot
\end{tabular}
\caption{Graphical diagnostics for causal forests in the setting of \eqref{eq:prop_setup}.
The first two panels evaluate the sampling error of causal forests and our infinitesimal
jackknife estimate of variance over 1,000 randomly draw test points, with $d = 20$.
The right-most panel shows standardized Gaussian QQ-plots for predictions at the same 1000 test points,
with $n = 800$ and $d = 20$. The first two panels are computed over 50 randomly drawn training
sets, and the last one over 20 training sets.}
\label{fig:variance}
\end{figure}

Figure \ref{fig:variance} offers some graphical diagnostics for causal forests in the setting of \eqref{eq:prop_setup}.
In the left panel, we observe how the causal forest sampling variance $\sigma^2_n(x)$  goes to zero with $n$; while
the center panel depicts the decay of the relative root-mean squared error of the infinitesimal jackknife estimate of variance,
i.e., $\mathbb{E}[(\hsigma^2_n(x) - \sigma_n^2(x))^2]^{1/2} / \sigma_n^2(x)$.
The boxplots display aggregate results for 1,000 randomly sampled test points $x$.
Finally, the right-most panel evaluates the Gaussianity of the forest predictions. Here, we first drew 1,000 random test points $x$, and computed $\htau(x)$ using forests grown on many different training sets. The plot shows standardized Gaussian QQ-plots aggregated over all these $x$; i.e., for each $x$, we plot Gaussian theoretical quantiles against sample quantiles of $(\htau(x) - \mathbb{E}(\htau(x))) / \sqrt{\text{Var}[\htau(x)]}$. 

\begin{table}[t]
\centering
\begin{tabular}{|r|ccc|ccc|}
\hline
 & \multicolumn{3}{c|}{mean-squared error} & \multicolumn{3}{c|}{coverage} \\
 \hline
$d$ & CF & 7-NN & 50-NN & CF & 7-NN & 50-NN \\ 
  \hline
  2 & 0.04 (0) & 0.29 (0) & 0.04 (0) & 0.97 (0) & 0.93 (0) & 0.94 (0) \\ 
  3 & 0.03 (0) & 0.29 (0) & 0.05 (0) & 0.96 (0) & 0.93 (0) & 0.92 (0) \\ 
  4 & 0.03 (0) & 0.30 (0) & 0.08 (0) & 0.94 (0) & 0.93 (0) & 0.86 (1) \\ 
  5 & 0.03 (0) & 0.31 (0) & 0.11 (0) & 0.93 (1) & 0.92 (0) & 0.77 (1) \\ 
  6 & 0.02 (0) & 0.34 (0) & 0.15 (0) & 0.93 (1) & 0.91 (0) & 0.68 (1) \\ 
  8 & 0.03 (0) & 0.38 (0) & 0.21 (0) & 0.90 (1) & 0.90 (0) & 0.57 (1) \\ 
   \hline
\end{tabular}
\caption{Comparison of the performance of a causal forests (CF) with that of the $k$-nearest neighbors ($k$-NN) estimator with $k = 7,\, 50$, on the setup \eqref{eq:tau0_setup}. The numbers in parentheses indicate the (rounded) standard sampling error for the last printed digit, obtained by aggregating performance over 25 simulation replications.}
\label{tab:tau0_simu}
\end{table}

In our second setup \eqref{eq:tau0_setup}, causal forests present a similar improvement over $k$-NN matching when $d > 2$, as seen in Table \ref{tab:tau0_simu}.\footnote{When $d = 2$, we do not expect causal forests to have a particular advantage over $k$-NN since the true $\tau$ also has 2-dimensional support; our results mirror this, as causal forests appear to have comparable performance to 50-NN.}
Unexpectedly, we find that the performance of causal forests \emph{improves} with $d$, at least when $d$ is small.
To understand this phenomenon, we note that the variance of a forest depends on the product of the variance of individual trees times the correlation between different trees \citep{breiman2001random,hastie2009elements}. Apparently, when $d$ is larger, the individual trees have more flexibility in how to place their splits, thus reducing their correlation and decreasing the variance of the full ensemble.

\begin{table}[t]
\centering
\begin{tabular}{|r|ccc|ccc|}
\hline
 & \multicolumn{3}{c|}{mean-squared error} & \multicolumn{3}{c|}{coverage}\\
 \hline
$d$ & CF & 10-NN & 100-NN & CF & 10-NN & 100-NN \\ 
  \hline
  2 & 0.02 (0) & 0.20 (0) & 0.02 (0) & 0.94 (0) & 0.93 (0) & 0.94 (0) \\ 
  3 & 0.02 (0) & 0.20 (0) & 0.03 (0) & 0.90 (0) & 0.93 (0) & 0.90 (0) \\ 
  4 & 0.02 (0) & 0.21 (0) & 0.06 (0) & 0.84 (1) & 0.93 (0) & 0.78 (1) \\ 
  5 & 0.02 (0) & 0.22 (0) & 0.09 (0) & 0.81 (1) & 0.93 (0) & 0.67 (0) \\ 
  6 & 0.02 (0) & 0.24 (0) & 0.15 (0) & 0.79 (1) & 0.92 (0) & 0.58 (0) \\ 
  8 & 0.03 (0) & 0.29 (0) & 0.26 (0) & 0.73 (1) & 0.90 (0) & 0.45 (0) \\ 
   \hline
\end{tabular}
\caption{Comparison of the performance of a causal forests (CF) with that of the $k$-nearest neighbors ($k$-NN) estimator with $k = 10,\, 100$, on the setup \eqref{eq:tau_setup}. The numbers in parentheses indicate the (rounded) standard sampling error for the last printed digit, obtained by aggregating performance over 40 simulation replications.}
\label{tab:tau_simu}
\end{table}

\begin{figure}[t]
\centering
\begin{tabular}{cccc}
{\begin{sideways}\parbox{0.25\textwidth}{\centering $d = 6$}\end{sideways}} &
\includegraphics[width = 0.25\textwidth]{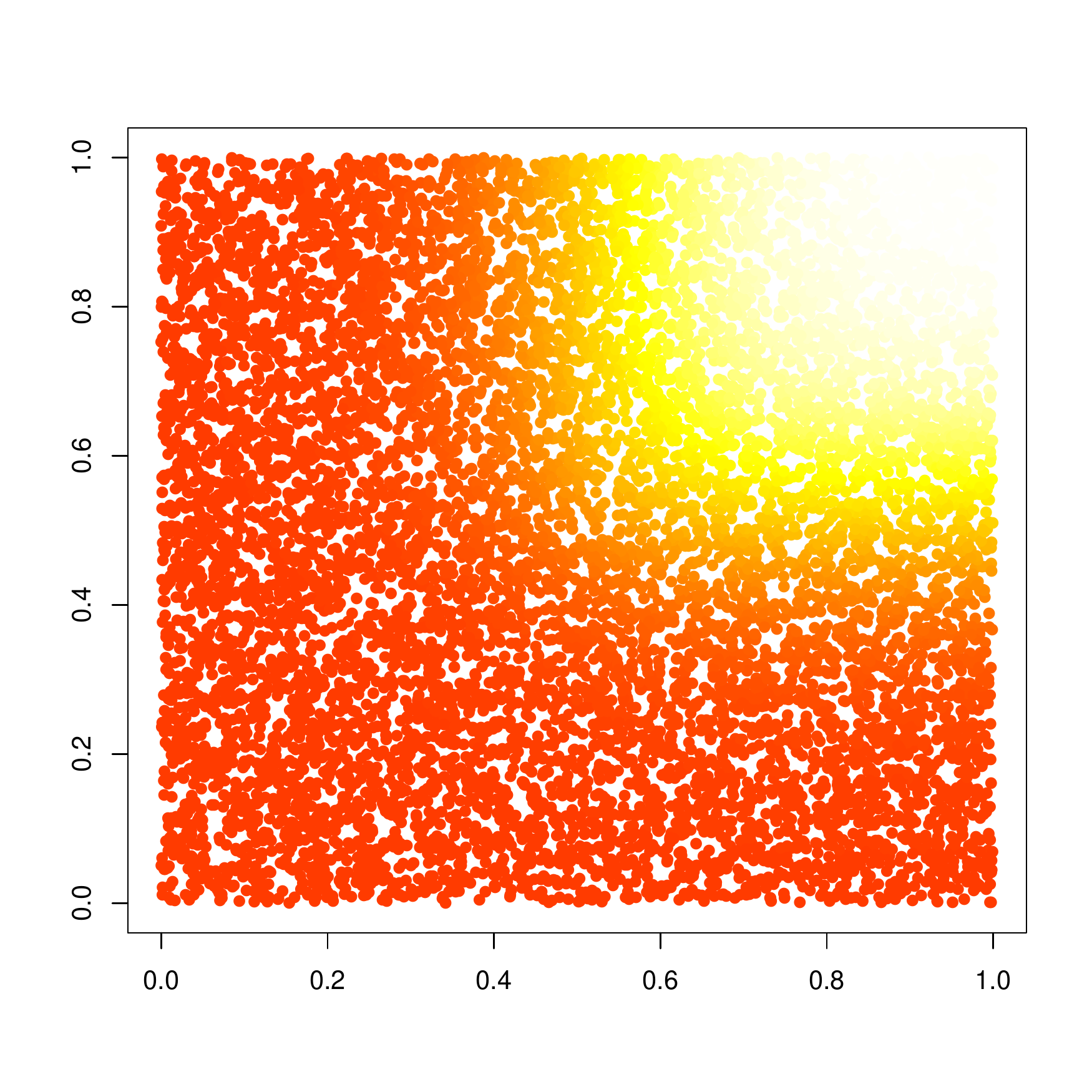} &
\includegraphics[width = 0.25\textwidth]{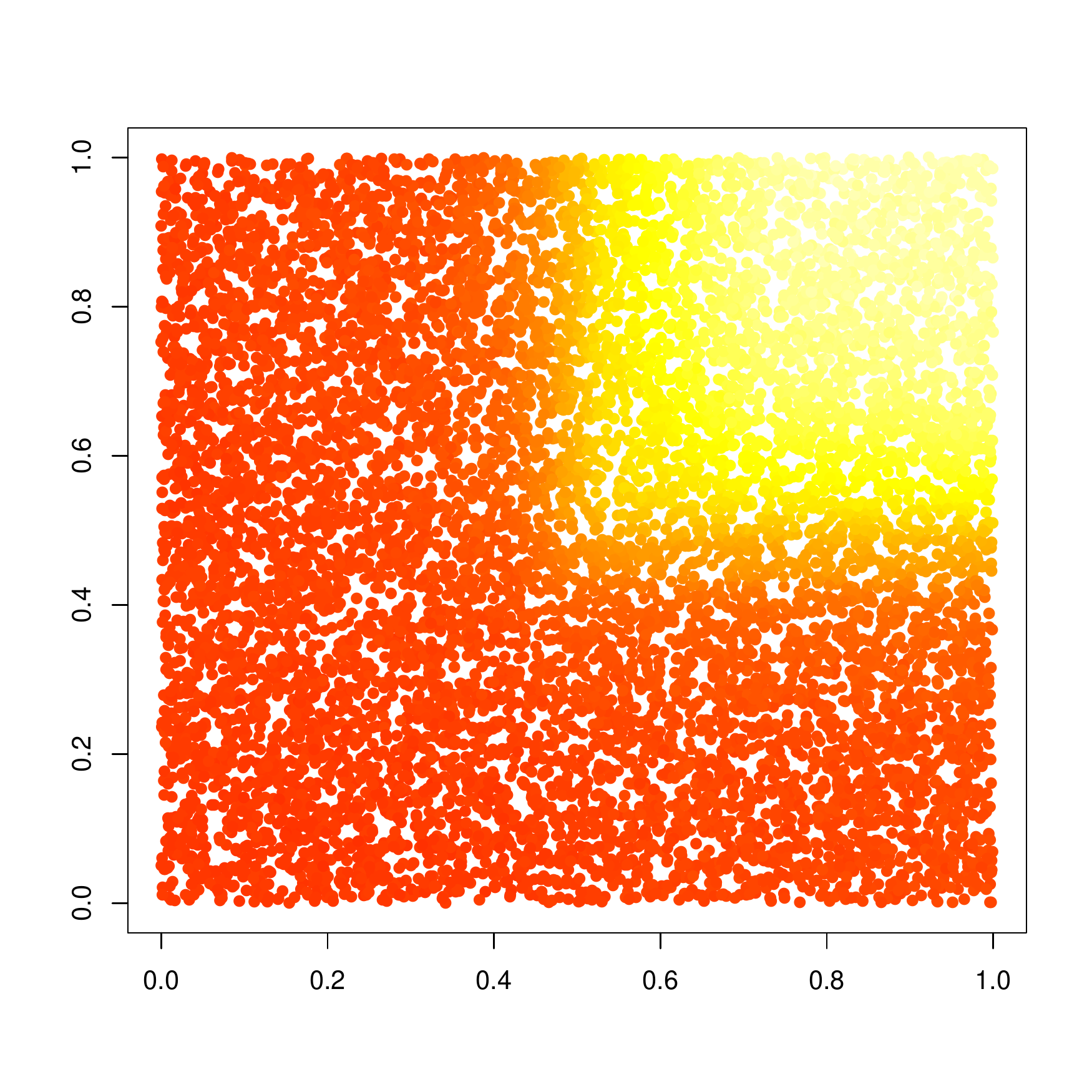} &
\includegraphics[width = 0.25\textwidth]{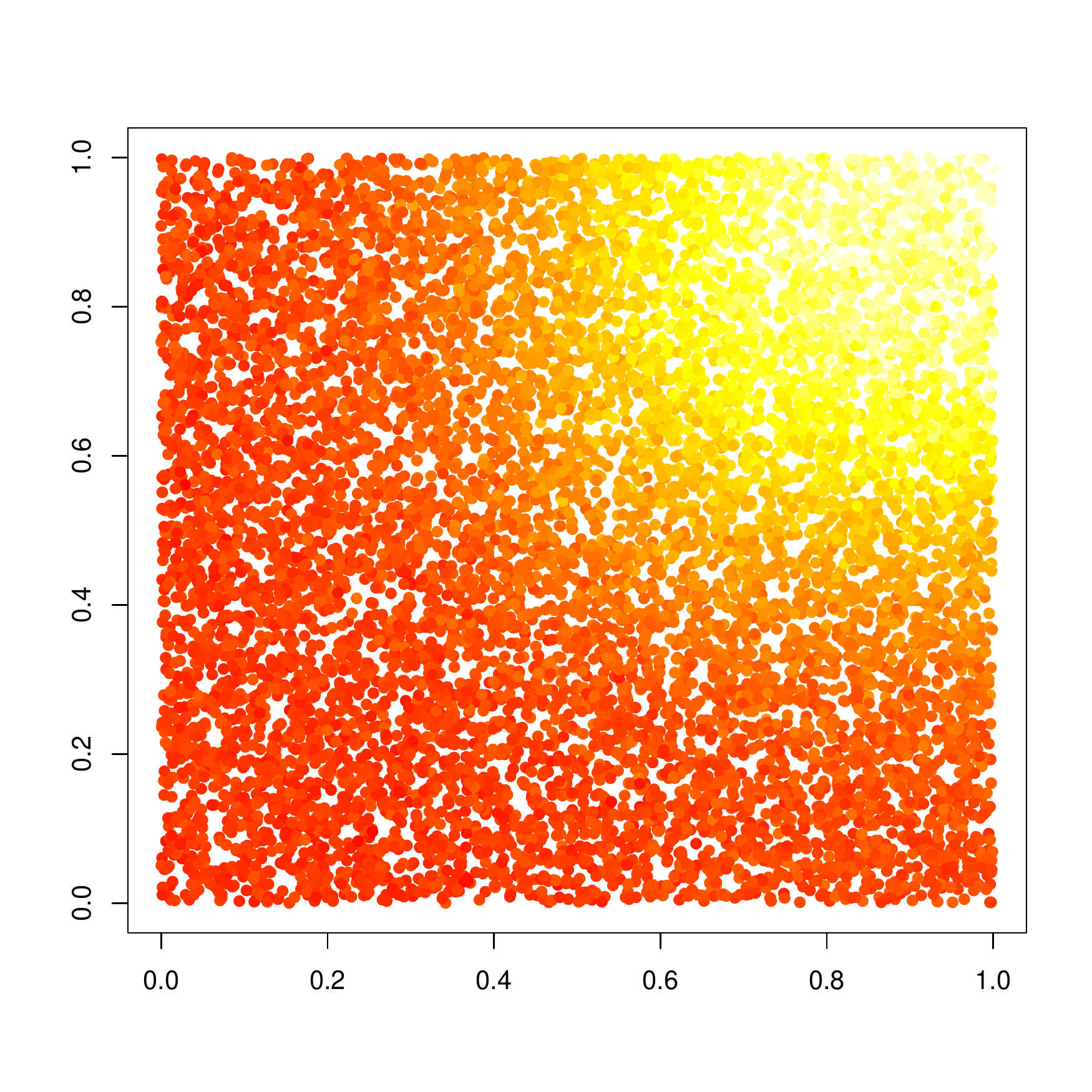}  \\
{\begin{sideways}\parbox{0.25\textwidth}{\centering $d = 20$}\end{sideways}} &
\includegraphics[width = 0.25\textwidth]{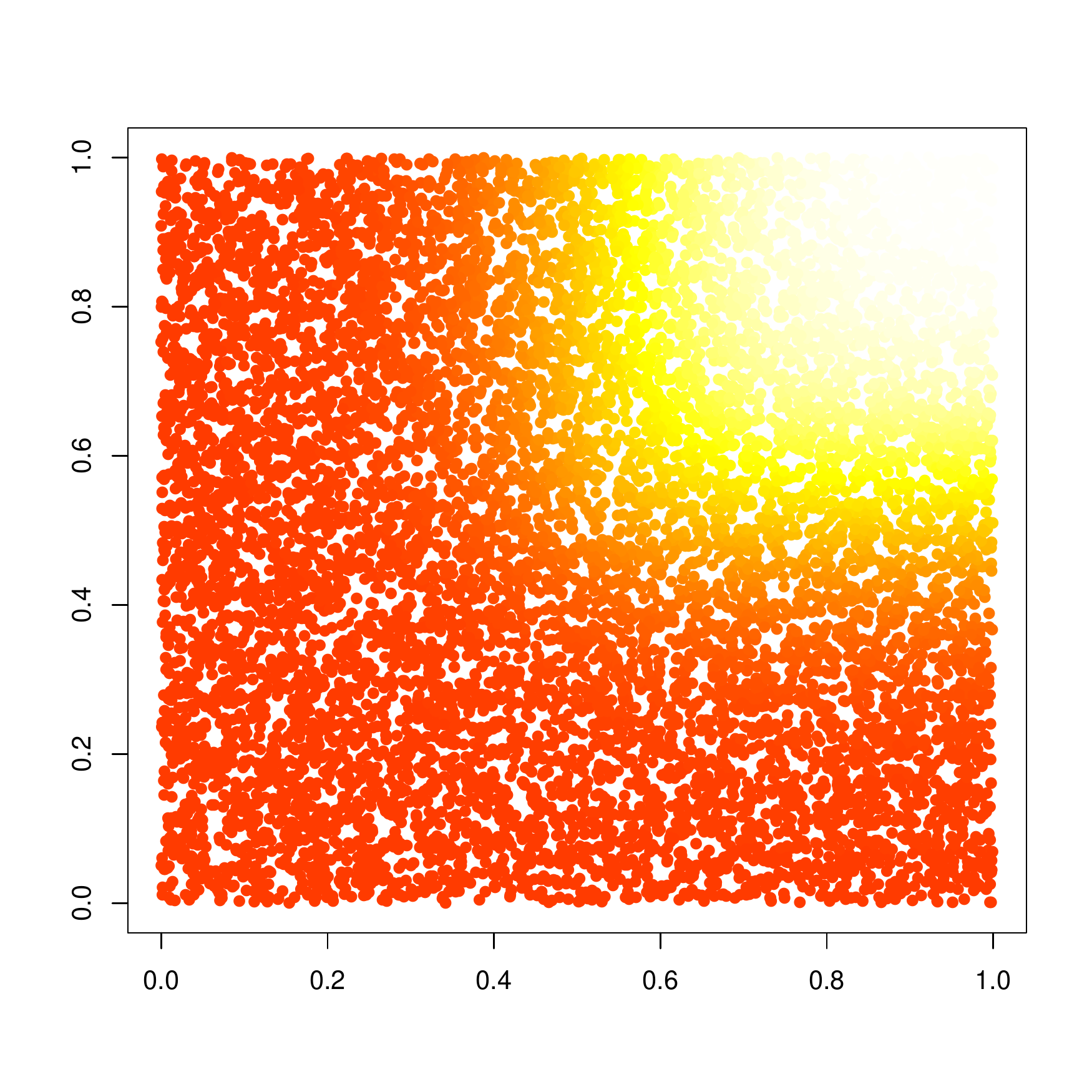} &
\includegraphics[width = 0.25\textwidth]{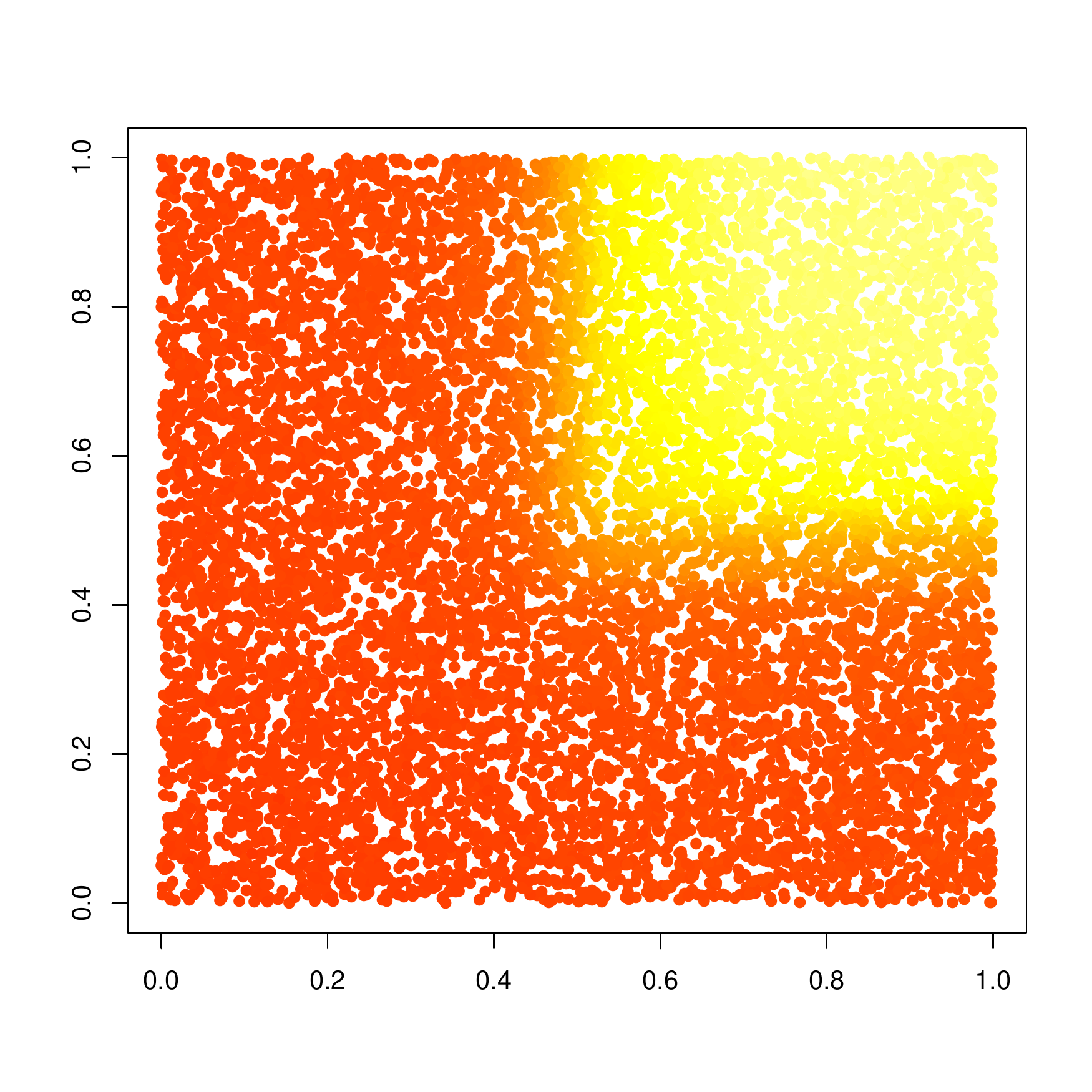} &
\includegraphics[width = 0.25\textwidth]{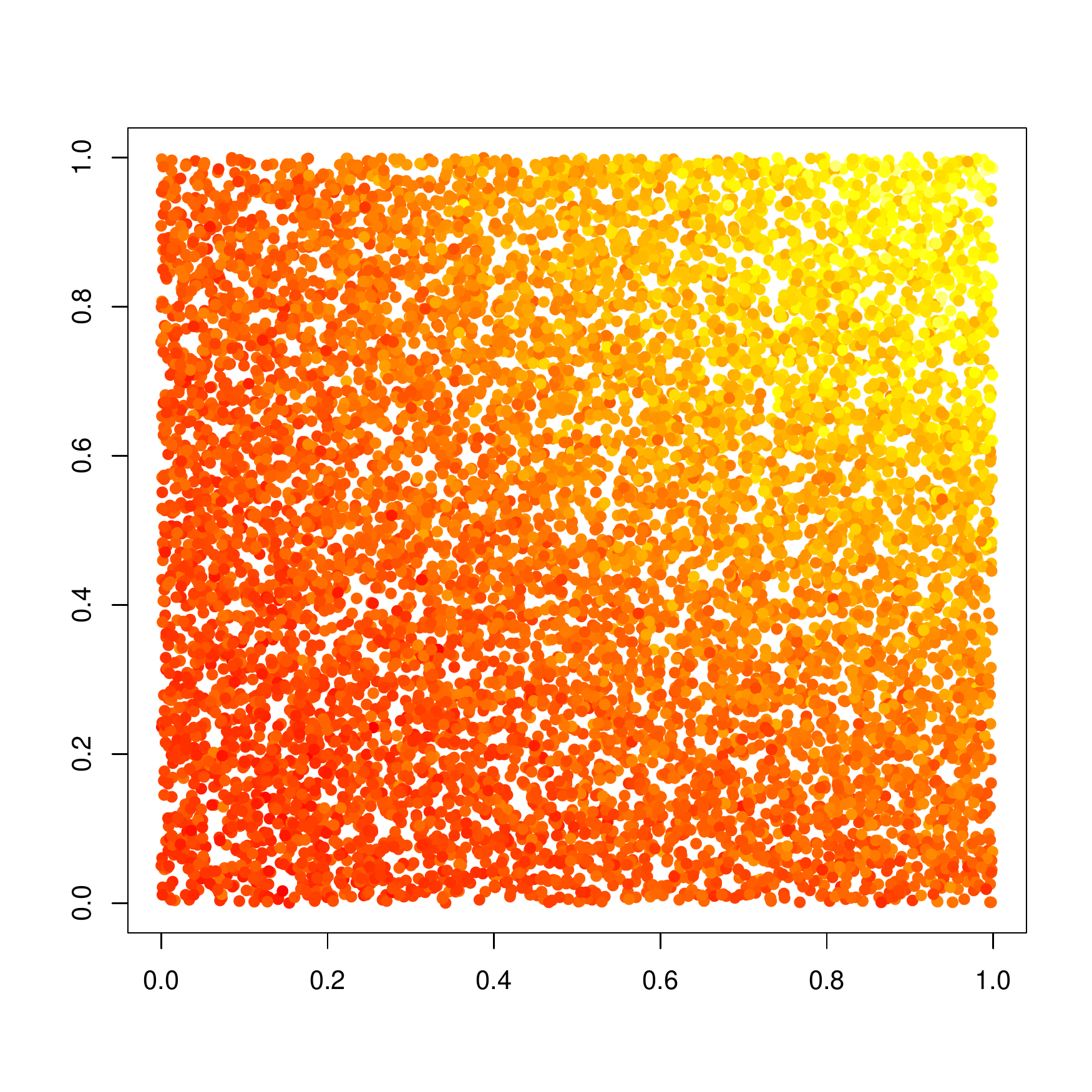}  \\
& True effect $\tau(x)$ & Causal forest & $k^*$-NN
\end{tabular}
\caption{The true treatment effect $\tau(X_i)$ at 10,000 random test examples $X_i$, along with estimates $\htau(X_i)$ produced by a causal forest and optimally-tuned $k$-NN, on data drawn according to \eqref{eq:tau_setup} with $d = 6, \, 20$. The test points are plotted according to their first two coordinates; the treatment effect is denoted by color, from dark (low) to light (high). On this simulation instance, causal forests and $k^*$-NN had a mean-squared error of 0.03 and 0.13 respectively for $d = 6$, and of 0.05 and 0.62 respectively for $d = 20$. The optimal tuning choices for $k$-NN were $k^* = 39$ for $d = 6$, and $k^* = 24$ for $d = 20$.}
\label{fig:2d_visu}
\end{figure}

Finally, in the setting \eqref{eq:tau_setup}, Table \ref{tab:tau_simu} shows that causal forests still
achieve an order of magnitude improvement over $k$-NN in terms of mean-squared error when
$d > 2$, but struggle more in terms of coverage.
This appears to largely be a bias effect: especially as $d$ gets larger, the random forest is dominated by bias instead of variance and so the confidence intervals are not centered.
Figure \ref{fig:2d_visu} illustrates this phenomenon: although the causal forest faithfully captures the qualitative aspects of the true $\tau$-surface, it does not exactly match its shape, especially in the upper-right corner where $\tau$ is largest.
Our theoretical results guarantee that this effect will go away as $n \rightarrow \infty$.
Figure \ref{fig:2d_visu} also helps us understand why $k$-NN performs so poorly in terms of mean-squared error: its predictive surface is both badly biased and noticeably ``grainy,'' especially for $d = 20$.
It suffers from bias not only at the boundary where the treatment effect is largest,
but also where the slope of the treatment effect is high in the interior.

These results highlight the promise of causal forests for accurate estimation of heterogeneous treatment effects, all while emphasizing avenues for further work. An immediate challenge is to control the bias of causal forests to achieve better coverage. Using more powerful splitting rules is a good way to reduce bias by enabling the trees to focus more closely on the coordinates with the greatest signal. The study of splitting rules for trees designed to estimate causal effects is still in its infancy and improvements may be possible.

A limitation of the present simulation study is that we manually chose whether to use double-sample forests
or propensity forests, depending on which procedure seemed more appropriate in each problem setting.
An important challenge for future work is to design splitting rules that can automatically choose
which characteristic of the training data to split on.
A principled and automatic rule for choosing $s$ would also be valuable. 

We present additional simulation results in the supplementary material. Appendix \ref{sec:more_simu}
has extensive simulations in the setting of Table \ref{tab:tau0_simu} while varying both $s$ and $n$;
and also considers a simulation setting where the signal is spread out over many different features, meaning
that forests have less upside over baseline methods. Finally, in Appendix \ref{sec:honesty}, we study
the effect of honesty versus adaptivity on forest predictive error.

\section{Discussion}

This paper proposed a class of non-parametric methods for
heterogeneous treatment effect estimation
that allow for data-driven feature selection
all while maintaining the benefits of classical methods, i.e.,
asymptotically normal and unbiased point estimates with valid confidence intervals.
Our causal forest estimator can be thought of as an adaptive nearest neighbor method,
where the data determines which dimensions are most important to consider in selecting nearest neighbors.
Such adaptivity seems essential for modern large-scale applications with many features.

In general, the challenge in using adaptive methods as the basis for valid statistical inference is that selection bias can be difficult to quantify; see \citet{berk2013valid}, \citet{chernozhukov2015valid}, \citet{taylor2015statistical}, and references therein for recent advances.
In this paper, pairing ``honest'' trees with the subsampling mechanism of random forests enabled us to accomplish this goal in a simple yet principled way.
In our simulation experiments, our method provides dramatically better mean-squared error than classical methods while achieving nominal coverage rates in moderate sample sizes.

A number of important extensions and refinements are left open.
Our current results only provide pointwise confidence intervals for $\tau(x)$; extending our theory to the setting of global functional estimation seems like a promising avenue for further work.
Another challenge is that nearest-neighbor non-parametric estimators typically suffer from bias at the boundaries of the support of the feature space.
A systematic approach to trimming at the boundaries, and possibly correcting for bias, would improve the coverage of the confidence intervals. 
In general, work can be done to identify methods that produce accurate variance estimates even in more challenging circumstances, e.g., with small samples or a large number of covariates, or to identify when variance estimates are unlikely to be reliable.

\section*{Acknowledgment}

We are grateful for helpful feedback from Brad Efron, Trevor Hastie, Guido Imbens, Guenther Walther,
as well as the associate editor, two anonymous referees, and seminar participants at
Atlantic Causal Inference Conference,
Berkeley Haas,
Berkeley Statistics,
California Econometrics Conference,
Cambridge,
Carnegie Mellon,
COMPSTAT,
Cornell,
Columbia Business School,
Columbia Statistics,
CREST,
EPFL,
ISAT/DARPA Machine Learning for Causal Inference,
JSM,
London Business School,
Microsoft Conference on Digital Economics,
Merck,
MIT IDSS,
MIT Sloan,
Northwestern,
SLDM,
Stanford GSB,
Stanford Statistics,
University College London,
University of Bonn,
University of Chicago Booth,
University of Chicago Statistics,
University of Illinois Urbana--Champaign,
University of Pennsylvania,
University of Southern California,
University of Washington,
Yale Econometrics, and
Yale Statistics.
S. W. was partially supported by a B. C. and E. J. Eaves Stanford Graduate Fellowship.

{
\setstretch{1}
\setlength{\bibsep}{0.2pt plus 0.3ex}
\bibliographystyle{plainnat}
\bibliography{references}

\begin{thebibliography}{77}
\providecommand{\natexlab}[1]{#1}
\providecommand{\url}[1]{\texttt{#1}}
\expandafter\ifx\csname urlstyle\endcsname\relax
  \providecommand{\doi}[1]{doi: #1}\else
  \providecommand{\doi}{doi: \begingroup \urlstyle{rm}\Url}\fi

\bibitem[Assmann et~al.(2000)Assmann, Pocock, Enos, and
  Kasten]{assmann2000subgroup}
Susan~F Assmann, Stuart~J Pocock, Laura~E Enos, and Linda~E Kasten.
\newblock Subgroup analysis and other (mis) uses of baseline data in clinical
  trials.
\newblock \emph{The Lancet}, 355\penalty0 (9209):\penalty0 1064--1069, 2000.

\bibitem[Athey and Imbens(2016)]{athey2015machine}
Susan Athey and Guido Imbens.
\newblock Recursive partitioning for heterogeneous causal effects.
\newblock \emph{Proceedings of the National Academy of Sciences}, 113\penalty0
  (27):\penalty0 7353--7360, 2016.

\bibitem[Athey et~al.(2016)Athey, Tibshirani, and Wager]{athey2016generalized}
Susan Athey, Julie Tibshirani, and Stefan Wager.
\newblock Generalized random forests.
\newblock \emph{arXiv preprint arXiv:1610.01271}, 2016.

\bibitem[Barndorff-Nielsen and Sobel(1966)]{barndorff1966distribution}
Ole Barndorff-Nielsen and Milton Sobel.
\newblock On the distribution of the number of admissible points in a vector
  random sample.
\newblock \emph{Theory of Probability \& Its Applications}, 11\penalty0
  (2):\penalty0 249--269, 1966.

\bibitem[Berk et~al.(2013)Berk, Brown, Buja, Zhang, and Zhao]{berk2013valid}
Richard Berk, Lawrence Brown, Andreas Buja, Kai Zhang, and Linda Zhao.
\newblock Valid post-selection inference.
\newblock \emph{The Annals of Statistics}, 41\penalty0 (2):\penalty0 802--837,
  2013.

\bibitem[Beygelzimer and Langford(2009)]{beygelzimer2009offset}
Alina Beygelzimer and John Langford.
\newblock The offset tree for learning with partial labels.
\newblock In \emph{Proceedings of the 15th International Conference on
  Knowledge Discovery and Data Mining}, pages 129--138. ACM, 2009.

\bibitem[Beygelzimer et~al.(2013)Beygelzimer, Kakadet, Langford, Arya, Mount,
  and Li]{beygelzimer2013fnn}
Alina Beygelzimer, Sham Kakadet, John Langford, Sunil Arya, David Mount, and
  Shengqiao Li.
\newblock \emph{{FNN}: Fast Nearest Neighbor Search Algorithms and
  Applications}, 2013.
\newblock URL \url{http://CRAN.R-project.org/package=FNN}.
\newblock {R} package version 1.1.

\bibitem[Bhattacharya and Dupas(2012)]{bhattacharya2012inferring}
Debopam Bhattacharya and Pascaline Dupas.
\newblock Inferring welfare maximizing treatment assignment under budget
  constraints.
\newblock \emph{Journal of Econometrics}, 167\penalty0 (1):\penalty0 168--196,
  2012.

\bibitem[Biau(2012)]{biau2012analysis}
G{\'e}rard Biau.
\newblock Analysis of a random forests model.
\newblock \emph{The Journal of Machine Learning Research}, 13:\penalty0
  1063--1095, 2012.

\bibitem[Biau and Devroye(2010)]{biau2010layered}
G{\'e}rard Biau and Luc Devroye.
\newblock On the layered nearest neighbour estimate, the bagged nearest
  neighbour estimate and the random forest method in regression and
  classification.
\newblock \emph{Journal of Multivariate Analysis}, 101\penalty0 (10):\penalty0
  2499--2518, 2010.

\bibitem[Biau et~al.(2008)Biau, Devroye, and Lugosi]{biau2008consistency}
G{\'e}rard Biau, Luc Devroye, and G{\'a}bor Lugosi.
\newblock Consistency of random forests and other averaging classifiers.
\newblock \emph{The Journal of Machine Learning Research}, 9:\penalty0
  2015--2033, 2008.

\bibitem[Biau et~al.(2010)Biau, C{\'e}rou, and Guyader]{biau2010rate}
G{\'e}rard Biau, Fr{\'e}d{\'e}ric C{\'e}rou, and Arnaud Guyader.
\newblock On the rate of convergence of the bagged nearest neighbor estimate.
\newblock \emph{The Journal of Machine Learning Research}, 11:\penalty0
  687--712, 2010.

\bibitem[Billingsley(2008)]{billingsley2008probability}
Patrick Billingsley.
\newblock \emph{Probability and Measure}.
\newblock John Wiley \& Sons, 2008.

\bibitem[Breiman(1996)]{breiman1996bagging}
Leo Breiman.
\newblock Bagging predictors.
\newblock \emph{Machine Learning}, 24\penalty0 (2):\penalty0 123--140, 1996.

\bibitem[Breiman(2001{\natexlab{a}})]{breiman2001random}
Leo Breiman.
\newblock Random forests.
\newblock \emph{Machine Learning}, 45\penalty0 (1):\penalty0 5--32,
  2001{\natexlab{a}}.

\bibitem[Breiman(2001{\natexlab{b}})]{breiman2001statistical}
Leo Breiman.
\newblock Statistical modeling: The two cultures (with comments and a rejoinder
  by the author).
\newblock \emph{Statistical Science}, 16\penalty0 (3):\penalty0 199--231,
  2001{\natexlab{b}}.

\bibitem[Breiman(2004)]{breiman2004consistency}
Leo Breiman.
\newblock Consistency for a simple model of random forests.
\newblock \emph{Statistical Department, University of California at Berkeley.
  Technical Report}, \penalty0 (670), 2004.

\bibitem[Breiman et~al.(1984)Breiman, Friedman, Stone, and
  Olshen]{breiman1984classification}
Leo Breiman, Jerome Friedman, Charles~J Stone, and Richard~A Olshen.
\newblock \emph{Classification and Regression Trees}.
\newblock CRC press, 1984.

\bibitem[B{\"u}hlmann and Yu(2002)]{buhlmann2002analyzing}
Peter B{\"u}hlmann and Bin Yu.
\newblock Analyzing bagging.
\newblock \emph{The Annals of Statistics}, 30\penalty0 (4):\penalty0 927--961,
  2002.

\bibitem[Buja and Stuetzle(2006)]{buja2006observations}
Andreas Buja and Werner Stuetzle.
\newblock Observations on bagging.
\newblock \emph{Statistica Sinica}, 16\penalty0 (2):\penalty0 323, 2006.

\bibitem[Chen and Hall(2003)]{chen2003effects}
Song~Xi Chen and Peter Hall.
\newblock Effects of bagging and bias correction on estimators defined by
  estimating equations.
\newblock \emph{Statistica Sinica}, 13\penalty0 (1):\penalty0 97--110, 2003.

\bibitem[Chernozhukov et~al.(2015)Chernozhukov, Hansen, and
  Spindler]{chernozhukov2015valid}
Victor Chernozhukov, Christian Hansen, and Martin Spindler.
\newblock Valid post-selection and post-regularization inference: An
  elementary, general approach.
\newblock \emph{Annual Review of Economics}, 7\penalty0 (1):\penalty0 649--688,
  2015.

\bibitem[Chipman et~al.(2010)Chipman, George, and McCulloch]{chipman2010bart}
Hugh~A Chipman, Edward~I George, and Robert~E McCulloch.
\newblock {BART}: {B}ayesian additive regression trees.
\newblock \emph{The Annals of Applied Statistics}, 4\penalty0 (1):\penalty0
  266--298, 2010.

\bibitem[Cook et~al.(2004)Cook, Gebski, and Keech]{cook2004subgroup}
David~I Cook, Val~J Gebski, and Anthony~C Keech.
\newblock Subgroup analysis in clinical trials.
\newblock \emph{Medical Journal of Australia}, 180\penalty0 (6):\penalty0
  289--292, 2004.

\bibitem[Crump et~al.(2008)Crump, Hotz, Imbens, and
  Mitnik]{crump2008nonparametric}
Richard~K Crump, V~Joseph Hotz, Guido~W Imbens, and Oscar~A Mitnik.
\newblock Nonparametric tests for treatment effect heterogeneity.
\newblock \emph{The Review of Economics and Statistics}, 90\penalty0
  (3):\penalty0 389--405, 2008.

\bibitem[Dehejia(2005)]{dehejia2005program}
Rajeev~H Dehejia.
\newblock Program evaluation as a decision problem.
\newblock \emph{Journal of Econometrics}, 125\penalty0 (1):\penalty0 141--173,
  2005.

\bibitem[Denil et~al.(2014)Denil, Matheson, and De~Freitas]{denil2014narrowing}
Misha Denil, David Matheson, and Nando De~Freitas.
\newblock Narrowing the gap: Random forests in theory and in practice.
\newblock In \emph{Proceedings of The 31st International Conference on Machine
  Learning}, pages 665--673, 2014.

\bibitem[Duan(2011)]{duan2011bootstrap}
Jiangtao Duan.
\newblock \emph{Bootstrap-Based Variance Estimators for a Bagging Predictor.}
\newblock PhD thesis, North Carolina State University, 2011.

\bibitem[Dud{\'\i}k et~al.(2011)Dud{\'\i}k, Langford, and
  Li]{langford2011doubly}
Miroslav Dud{\'\i}k, John Langford, and Lihong Li.
\newblock Doubly robust policy evaluation and learning.
\newblock In \emph{Proceedings of the 28th International Conference on Machine
  Learning}, pages 1097--1104, 2011.

\bibitem[Efron(2014)]{efron2013estimation}
Bradley Efron.
\newblock Estimation and accuracy after model selection (with discussion).
\newblock \emph{Journal of the American Statistical Association}, 109\penalty0
  (507), 2014.

\bibitem[Efron and Stein(1981)]{efron1981jackknife}
Bradley Efron and Charles Stein.
\newblock The jackknife estimate of variance.
\newblock \emph{The Annals of Statistics}, 9\penalty0 (3):\penalty0 586--596,
  1981.

\bibitem[Foster et~al.(2011)Foster, Taylor, and Ruberg]{foster2011subgroup}
Jared~C Foster, Jeremy~MG Taylor, and Stephen~J Ruberg.
\newblock Subgroup identification from randomized clinical trial data.
\newblock \emph{Statistics in Medicine}, 30\penalty0 (24):\penalty0 2867--2880,
  2011.

\bibitem[Friedman and Hall(2007)]{friedman2007bagging}
Jerome~H Friedman and Peter Hall.
\newblock On bagging and nonlinear estimation.
\newblock \emph{Journal of Statistical Planning and Inference}, 137\penalty0
  (3):\penalty0 669--683, 2007.

\bibitem[Green and Kern(2012)]{green2012modeling}
Donald~P Green and Holger~L Kern.
\newblock Modeling heterogeneous treatment effects in survey experiments with
  {B}ayesian additive regression trees.
\newblock \emph{Public Opinion Quarterly}, 76\penalty0 (3):\penalty0 491--511,
  2012.

\bibitem[H{\'a}jek(1968)]{hajek1968asymptotic}
Jaroslav H{\'a}jek.
\newblock Asymptotic normality of simple linear rank statistics under
  alternatives.
\newblock \emph{The Annals of Mathematical Statistics}, 39\penalty0
  (2):\penalty0 325--346, 1968.

\bibitem[Hastie et~al.(2009)Hastie, Tibshirani, and
  Friedman]{hastie2009elements}
Trevor Hastie, Robert Tibshirani, and Jerome Friedman.
\newblock \emph{The Elements of Statistical Learning}.
\newblock New York: Springer, 2009.

\bibitem[Hill and Su(2013)]{hill2013assessing}
Jennifer Hill and Yu-Sung Su.
\newblock Assessing lack of common support in causal inference using bayesian
  nonparametrics: Implications for evaluating the effect of breastfeeding on
  children’s cognitive outcomes.
\newblock \emph{The Annals of Applied Statistics}, 7\penalty0 (3):\penalty0
  1386--1420, 2013.

\bibitem[Hill(2011)]{hill2011bayesian}
Jennifer~L Hill.
\newblock Bayesian nonparametric modeling for causal inference.
\newblock \emph{Journal of Computational and Graphical Statistics}, 20\penalty0
  (1), 2011.

\bibitem[Hirano and Porter(2009)]{hirano2009asymptotics}
Keisuke Hirano and Jack~R Porter.
\newblock Asymptotics for statistical treatment rules.
\newblock \emph{Econometrica}, 77\penalty0 (5):\penalty0 1683--1701, 2009.

\bibitem[Hirano et~al.(2003)Hirano, Imbens, and Ridder]{hirano2003efficient}
Keisuke Hirano, Guido~W Imbens, and Geert Ridder.
\newblock Efficient estimation of average treatment effects using the estimated
  propensity score.
\newblock \emph{Econometrica}, 71\penalty0 (4):\penalty0 1161--1189, 2003.

\bibitem[Hoeffding(1948)]{hoeffding1948class}
Wassily Hoeffding.
\newblock A class of statistics with asymptotically normal distribution.
\newblock \emph{The Annals of Mathematical Statistics}, 19\penalty0
  (3):\penalty0 293--325, 1948.

\bibitem[Imai and Ratkovic(2013)]{imai2013estimating}
Kosuke Imai and Marc Ratkovic.
\newblock Estimating treatment effect heterogeneity in randomized program
  evaluation.
\newblock \emph{The Annals of Applied Statistics}, 7\penalty0 (1):\penalty0
  443--470, 2013.

\bibitem[Imbens and Rubin(2015)]{imbens2015causal}
Guido~W Imbens and Donald~B Rubin.
\newblock \emph{Causal Inference in Statistics, Social, and Biomedical
  Sciences}.
\newblock Cambridge University Press, 2015.

\bibitem[Jaeckel(1972)]{jaeckel1972infinitesimal}
Louis~A Jaeckel.
\newblock \emph{The Infinitesimal Jackknife}.
\newblock 1972.

\bibitem[Kleinberg et~al.(2015)Kleinberg, Ludwig, Mullainathan, and
  Obermeyer]{kleinberg2015prediction}
Jon Kleinberg, Jens Ludwig, Sendhil Mullainathan, and Ziad Obermeyer.
\newblock Prediction policy problems.
\newblock \emph{American Economic Review}, 105\penalty0 (5):\penalty0 491--95,
  2015.

\bibitem[Lee(2009)]{lee2009non}
Myoung-jae Lee.
\newblock Non-parametric tests for distributional treatment effect for randomly
  censored responses.
\newblock \emph{Journal of the Royal Statistical Society: Series B (Statistical
  Methodology)}, 71\penalty0 (1):\penalty0 243--264, 2009.

\bibitem[Liaw and Wiener(2002)]{liaw2002classification}
Andy Liaw and Matthew Wiener.
\newblock Classification and regression by random{F}orest.
\newblock \emph{R News}, 2\penalty0 (3):\penalty0 18--22, 2002.
\newblock URL \url{http://CRAN.R-project.org/doc/Rnews/}.

\bibitem[Lin and Jeon(2006)]{lin2006random}
Yi~Lin and Yongho Jeon.
\newblock Random forests and adaptive nearest neighbors.
\newblock \emph{Journal of the American Statistical Association}, 101\penalty0
  (474):\penalty0 578--590, 2006.

\bibitem[Mallows(1973)]{mallows1973some}
Colin~L Mallows.
\newblock Some comments on {Cp}.
\newblock \emph{Technometrics}, 15\penalty0 (4):\penalty0 661--675, 1973.

\bibitem[Manski(2004)]{manski2004statistical}
Charles~F Manski.
\newblock Statistical treatment rules for heterogeneous populations.
\newblock \emph{Econometrica}, 72\penalty0 (4):\penalty0 1221--1246, 2004.

\bibitem[McCaffrey et~al.(2004)McCaffrey, Ridgeway, and
  Morral]{mccaffrey2004propensity}
Daniel~F McCaffrey, Greg Ridgeway, and Andrew~R Morral.
\newblock Propensity score estimation with boosted regression for evaluating
  causal effects in observational studies.
\newblock \emph{Psychological Methods}, 9\penalty0 (4):\penalty0 403, 2004.

\bibitem[Meinshausen(2006)]{meinshausen2006quantile}
Nicolai Meinshausen.
\newblock Quantile regression forests.
\newblock \emph{The Journal of Machine Learning Research}, 7:\penalty0
  983--999, 2006.

\bibitem[Mentch and Hooker(2016)]{mentch2014ensemble}
Lucas Mentch and Giles Hooker.
\newblock Quantifying uncertainty in random forests via confidence intervals
  and hypothesis tests.
\newblock \emph{Journal of Machine Learning Research}, 17\penalty0
  (26):\penalty0 1--41, 2016.

\bibitem[Neyman(1923)]{neyman1923applications}
Jersey Neyman.
\newblock Sur les applications de la th{\'e}orie des probabilit{\'e}s aux
  experiences agricoles: Essai des principes.
\newblock \emph{Roczniki Nauk Rolniczych}, 10:\penalty0 1--51, 1923.

\bibitem[Politis et~al.(1999)Politis, Romano, and Wolf]{politis1999subsampling}
Dimitris~N. Politis, Joseph~P. Romano, and Michael Wolf.
\newblock \emph{Subsampling}.
\newblock Springer Series in Statistics. Springer New York, 1999.

\bibitem[Rosenbaum and Rubin(1983)]{rosenbaum1983central}
Paul~R Rosenbaum and Donald~B Rubin.
\newblock The central role of the propensity score in observational studies for
  causal effects.
\newblock \emph{Biometrika}, 70\penalty0 (1):\penalty0 41--55, 1983.

\bibitem[Rubin(1974)]{rubin1974estimating}
Donald~B Rubin.
\newblock Estimating causal effects of treatments in randomized and
  nonrandomized studies.
\newblock \emph{Journal of Educational Psychology}, 66\penalty0 (5):\penalty0
  688, 1974.

\bibitem[Samworth(2012)]{samworth2012optimal}
Richard~J Samworth.
\newblock Optimal weighted nearest neighbour classifiers.
\newblock \emph{The Annals of Statistics}, 40\penalty0 (5):\penalty0
  2733--2763, 2012.

\bibitem[Schick(1986)]{schick1986asymptotically}
Anton Schick.
\newblock On asymptotically efficient estimation in semiparametric models.
\newblock \emph{The Annals of Statistics}, pages 1139--1151, 1986.

\bibitem[Scornet et~al.(2015)Scornet, Biau, and Vert]{scornet2015consistency}
Erwan Scornet, G\'erard Biau, and Jean-Philippe Vert.
\newblock Consistency of random forests.
\newblock \emph{The Annals of Statistics}, 43\penalty0 (4):\penalty0
  1716--1741, 2015.

\bibitem[Sexton and Laake(2009)]{sexton2009standard}
Joseph Sexton and Petter Laake.
\newblock Standard errors for bagged and random forest estimators.
\newblock \emph{Computational Statistics \& Data Analysis}, 53\penalty0
  (3):\penalty0 801--811, 2009.

\bibitem[Signorovitch(2007)]{signorovitch2007identifying}
James~Edward Signorovitch.
\newblock \emph{Identifying Informative Biological Markers in High-Dimensional
  Genomic Data and Clinical Trials}.
\newblock PhD thesis, Harvard University, 2007.

\bibitem[Strobl et~al.(2007)Strobl, Boulesteix, Zeileis, and
  Hothorn]{strobl2007bias}
Carolin Strobl, Anne-Laure Boulesteix, Achim Zeileis, and Torsten Hothorn.
\newblock Bias in random forest variable importance measures: Illustrations,
  sources and a solution.
\newblock \emph{BMC Bioinformatics}, 8\penalty0 (1):\penalty0 25, 2007.

\bibitem[Su et~al.(2009)Su, Tsai, Wang, Nickerson, and Li]{su2009subgroup}
Xiaogang Su, Chih-Ling Tsai, Hansheng Wang, David~M Nickerson, and Bogong Li.
\newblock Subgroup analysis via recursive partitioning.
\newblock \emph{The Journal of Machine Learning Research}, 10:\penalty0
  141--158, 2009.

\bibitem[Taddy et~al.(2016)Taddy, Gardner, Chen, and
  Draper]{taddy2014heterogeneous}
Matt Taddy, Matt Gardner, Liyun Chen, and David Draper.
\newblock A nonparametric {B}ayesian analysis of heterogeneous treatment
  effects in digital experimentation.
\newblock \emph{Journal of Business \& Economic Statistics}, \penalty0
  (just-accepted), 2016.

\bibitem[Taylor and Tibshirani(2015)]{taylor2015statistical}
Jonathan Taylor and Robert~J Tibshirani.
\newblock Statistical learning and selective inference.
\newblock \emph{Proceedings of the National Academy of Sciences}, 112\penalty0
  (25):\penalty0 7629--7634, 2015.

\bibitem[Tian et~al.(2014)Tian, Alizadeh, Gentles, and
  Tibshirani]{tian2014simple}
Lu~Tian, Ash~A Alizadeh, Andrew~J Gentles, and Robert Tibshirani.
\newblock A simple method for estimating interactions between a treatment and a
  large number of covariates.
\newblock \emph{Journal of the American Statistical Association}, 109\penalty0
  (508):\penalty0 1517--1532, 2014.

\bibitem[Van~der Vaart(2000)]{van2000asymptotic}
Aad~W. Van~der Vaart.
\newblock \emph{Asymptotic Statistics}.
\newblock Number~3. Cambridge Univ Pr, 2000.

\bibitem[Wager(2014)]{wager2014asymptotic}
Stefan Wager.
\newblock Asymptotic theory for random forests.
\newblock \emph{arXiv preprint arXiv:1405.0352}, 2014.

\bibitem[Wager and Walther(2015)]{wager2015uniform}
Stefan Wager and Guenther Walther.
\newblock Adaptive concentration of regression trees, with application to
  random forests.
\newblock \emph{arXiv preprint arXiv:1503.06388}, 2015.

\bibitem[Wager et~al.(2014)Wager, Hastie, and Efron]{wager2014confidence}
Stefan Wager, Trevor Hastie, and Bradley Efron.
\newblock Confidence intervals for random forests: The jackknife and the
  infinitesimal jackknife.
\newblock \emph{The Journal of Machine Learning Research}, 15, 2014.

\bibitem[Wasserman and Roeder(2009)]{wasserman2009high}
Larry Wasserman and Kathryn Roeder.
\newblock High-dimensional variable selection.
\newblock \emph{The Annals of Statistics}, 37\penalty0 (5A):\penalty0
  2178--2201, 2009.

\bibitem[Weisberg and Pontes(2015)]{weisberg2015post}
Herbert~I Weisberg and Victor~P Pontes.
\newblock Post hoc subgroups in clinical trials: Anathema or analytics?
\newblock \emph{Clinical Trials}, 12\penalty0 (4):\penalty0 357--364, 2015.

\bibitem[Westreich et~al.(2010)Westreich, Lessler, and
  Funk]{westreich2010propensity}
Daniel Westreich, Justin Lessler, and Michele~J Funk.
\newblock Propensity score estimation: neural networks, support vector
  machines, decision trees ({CART}), and meta-classifiers as alternatives to
  logistic regression.
\newblock \emph{Journal of Clinical Epidemiology}, 63\penalty0 (8):\penalty0
  826--833, 2010.

\bibitem[Willke et~al.(2012)Willke, Zheng, Subedi, Althin, and
  Mullins]{willke2012concepts}
Richard~J Willke, Zhiyuan Zheng, Prasun Subedi, Rikard Althin, and C~Daniel
  Mullins.
\newblock From concepts, theory, and evidence of heterogeneity of treatment
  effects to methodological approaches: a primer.
\newblock \emph{BMC Medical Research Methodology}, 12\penalty0 (1):\penalty0
  185, 2012.

\bibitem[Zeileis et~al.(2008)Zeileis, Hothorn, and Hornik]{zeileis2008model}
Achim Zeileis, Torsten Hothorn, and Kurt Hornik.
\newblock Model-based recursive partitioning.
\newblock \emph{Journal of Computational and Graphical Statistics}, 17\penalty0
  (2):\penalty0 492--514, 2008.

\bibitem[Zhu et~al.(2015)Zhu, Zeng, and Kosorok]{zhu2015reinforcement}
Ruoqing Zhu, Donglin Zeng, and Michael~R Kosorok.
\newblock Reinforcement learning trees.
\newblock \emph{Journal of the American Statistical Association}, 110\penalty0
  (512):\penalty0 1770--1784, 2015.

\end{thebibliography}
}

\newpage

\begin{appendix}

\newgeometry{margin=1.5cm}
{\thispagestyle{empty}
\setlength{\tabcolsep}{0.2em}
\begin{table}[p]
\centering
\resizebox{\textwidth}{!}{
\begin{tabular}{||ccc|ccc||ccc|ccc||ccc|ccc||}
  \hline
\hline
n & d & s & MSE & Coverage & Variance & n & d & s & MSE & Coverage & Variance & n & d & s & MSE & Coverage & Variance \\ 
  \hline
\hline
1000 & 2 & 100 & 0.16 (1) & 0.29 (3) & 0.01 (0) & 2000 & 4 & 200 & 0.14 (1) & 0.25 (3) & 0.01 (0) & 5000 & 6 & 500 & 0.05 (0) & 0.52 (3) & 0.01 (0) \\ 
  1000 & 2 & 200 & 0.09 (1) & 0.69 (5) & 0.03 (0) & 2000 & 4 & 400 & 0.06 (1) & 0.70 (4) & 0.02 (0) & 5000 & 6 & 1000 & 0.03 (0) & 0.75 (3) & 0.01 (0) \\ 
  1000 & 2 & 250 & 0.07 (1) & 0.79 (4) & 0.03 (0) & 2000 & 4 & 500 & 0.05 (1) & 0.81 (3) & 0.02 (0) & 5000 & 6 & 1250 & 0.02 (0) & 0.79 (3) & 0.01 (0) \\ 
  1000 & 2 & 333 & 0.07 (1) & 0.89 (3) & 0.04 (0) & 2000 & 4 & 667 & 0.04 (0) & 0.90 (2) & 0.03 (0) & 5000 & 6 & 1667 & 0.02 (0) & 0.86 (2) & 0.02 (0) \\ 
  1000 & 2 & 500 & 0.07 (1) & 0.94 (2) & 0.07 (0) & 2000 & 4 & 1000 & 0.04 (0) & 0.95 (1) & 0.04 (0) & 5000 & 6 & 2500 & 0.02 (0) & 0.92 (1) & 0.02 (0) \\ 
  1000 & 2 & 667 & 0.08 (1) & 0.90 (4) & 0.08 (1) & 2000 & 4 & 1333 & 0.05 (0) & 0.96 (1) & 0.06 (0) & 5000 & 6 & 3333 & 0.03 (0) & 0.96 (1) & 0.03 (0) \\ 
   \hline
1000 & 3 & 100 & 0.25 (2) & 0.16 (2) & 0.01 (0) & 2000 & 5 & 200 & 0.17 (1) & 0.18 (2) & 0.01 (0) & 5000 & 8 & 500 & 0.06 (1) & 0.42 (3) & 0.01 (0) \\ 
  1000 & 3 & 200 & 0.13 (2) & 0.53 (5) & 0.02 (0) & 2000 & 5 & 400 & 0.07 (1) & 0.65 (6) & 0.02 (0) & 5000 & 8 & 1000 & 0.03 (0) & 0.69 (2) & 0.01 (0) \\ 
  1000 & 3 & 250 & 0.11 (2) & 0.66 (6) & 0.03 (0) & 2000 & 5 & 500 & 0.05 (1) & 0.75 (5) & 0.02 (0) & 5000 & 8 & 1250 & 0.03 (0) & 0.73 (3) & 0.01 (0) \\ 
  1000 & 3 & 333 & 0.09 (1) & 0.81 (5) & 0.04 (0) & 2000 & 5 & 667 & 0.05 (0) & 0.84 (3) & 0.03 (0) & 5000 & 8 & 1667 & 0.03 (0) & 0.81 (2) & 0.01 (0) \\ 
  1000 & 3 & 500 & 0.08 (1) & 0.90 (3) & 0.06 (0) & 2000 & 5 & 1000 & 0.04 (0) & 0.91 (1) & 0.04 (0) & 5000 & 8 & 2500 & 0.03 (0) & 0.88 (2) & 0.02 (0) \\ 
  1000 & 3 & 667 & 0.08 (1) & 0.73 (5) & 0.04 (0) & 2000 & 5 & 1333 & 0.05 (0) & 0.92 (2) & 0.05 (0) & 5000 & 8 & 3333 & 0.03 (0) & 0.92 (1) & 0.03 (0) \\ 
   \hline
1000 & 4 & 100 & 0.32 (1) & 0.12 (1) & 0.01 (0) & 2000 & 6 & 200 & 0.22 (1) & 0.12 (1) & 0.01 (0) & 10000 & 2 & 1000 & 0.01 (0) & 0.95 (1) & 0.01 (0) \\ 
  1000 & 4 & 200 & 0.16 (1) & 0.43 (3) & 0.03 (0) & 2000 & 6 & 400 & 0.09 (1) & 0.52 (5) & 0.02 (0) & 10000 & 2 & 2000 & 0.01 (0) & 0.96 (1) & 0.02 (0) \\ 
  1000 & 4 & 250 & 0.13 (1) & 0.58 (4) & 0.03 (0) & 2000 & 6 & 500 & 0.07 (1) & 0.68 (4) & 0.02 (0) & 10000 & 2 & 2500 & 0.02 (0) & 0.96 (0) & 0.02 (0) \\ 
  1000 & 4 & 333 & 0.09 (1) & 0.76 (3) & 0.04 (0) & 2000 & 6 & 667 & 0.05 (0) & 0.82 (3) & 0.03 (0) & 10000 & 2 & 3333 & 0.02 (0) & 0.96 (0) & 0.02 (0) \\ 
  1000 & 4 & 500 & 0.07 (1) & 0.91 (2) & 0.06 (0) & 2000 & 6 & 1000 & 0.04 (0) & 0.89 (2) & 0.03 (0) & 10000 & 2 & 5000 & 0.03 (0) & 0.97 (0) & 0.04 (0) \\ 
  1000 & 4 & 667 & 0.07 (1) & 0.79 (3) & 0.04 (1) & 2000 & 6 & 1333 & 0.05 (0) & 0.94 (1) & 0.05 (0) & 10000 & 2 & 6667 & 0.04 (0) & 0.98 (0) & 0.06 (0) \\ 
   \hline
1000 & 5 & 100 & 0.34 (2) & 0.10 (1) & 0.01 (0) & 2000 & 8 & 200 & 0.24 (1) & 0.12 (1) & 0.01 (0) & 10000 & 3 & 1000 & 0.01 (0) & 0.84 (1) & 0.01 (0) \\ 
  1000 & 5 & 200 & 0.16 (2) & 0.41 (6) & 0.02 (0) & 2000 & 8 & 400 & 0.08 (0) & 0.61 (4) & 0.02 (0) & 10000 & 3 & 2000 & 0.01 (0) & 0.91 (1) & 0.01 (0) \\ 
  1000 & 5 & 250 & 0.12 (2) & 0.59 (5) & 0.03 (0) & 2000 & 8 & 500 & 0.06 (0) & 0.78 (2) & 0.02 (0) & 10000 & 3 & 2500 & 0.01 (0) & 0.92 (1) & 0.01 (0) \\ 
  1000 & 5 & 333 & 0.09 (1) & 0.80 (4) & 0.04 (0) & 2000 & 8 & 667 & 0.05 (0) & 0.85 (1) & 0.02 (0) & 10000 & 3 & 3333 & 0.02 (0) & 0.94 (1) & 0.02 (0) \\ 
  1000 & 5 & 500 & 0.07 (1) & 0.89 (3) & 0.06 (0) & 2000 & 8 & 1000 & 0.04 (0) & 0.91 (1) & 0.04 (0) & 10000 & 3 & 5000 & 0.02 (0) & 0.95 (0) & 0.03 (0) \\ 
  1000 & 5 & 667 & 0.07 (1) & 0.77 (4) & 0.04 (0) & 2000 & 8 & 1333 & 0.04 (0) & 0.89 (3) & 0.04 (0) & 10000 & 3 & 6667 & 0.03 (0) & 0.97 (0) & 0.04 (0) \\ 
   \hline
1000 & 6 & 100 & 0.41 (3) & 0.07 (1) & 0.01 (0) & 5000 & 2 & 500 & 0.02 (0) & 0.86 (2) & 0.01 (0) & 10000 & 4 & 1000 & 0.02 (0) & 0.73 (3) & 0.01 (0) \\ 
  1000 & 6 & 200 & 0.22 (3) & 0.31 (4) & 0.02 (0) & 5000 & 2 & 1000 & 0.02 (0) & 0.92 (1) & 0.02 (0) & 10000 & 4 & 2000 & 0.02 (0) & 0.85 (2) & 0.01 (0) \\ 
  1000 & 6 & 250 & 0.17 (3) & 0.48 (6) & 0.03 (0) & 5000 & 2 & 1250 & 0.02 (0) & 0.94 (1) & 0.02 (0) & 10000 & 4 & 2500 & 0.02 (0) & 0.87 (1) & 0.01 (0) \\ 
  1000 & 6 & 333 & 0.12 (2) & 0.68 (7) & 0.04 (0) & 5000 & 2 & 1667 & 0.03 (0) & 0.95 (1) & 0.03 (0) & 10000 & 4 & 3333 & 0.02 (0) & 0.90 (1) & 0.01 (0) \\ 
  1000 & 6 & 500 & 0.08 (2) & 0.89 (3) & 0.05 (0) & 5000 & 2 & 2500 & 0.04 (0) & 0.96 (0) & 0.05 (0) & 10000 & 4 & 5000 & 0.02 (0) & 0.93 (1) & 0.02 (0) \\ 
  1000 & 6 & 667 & 0.07 (1) & 0.75 (6) & 0.04 (0) & 5000 & 2 & 3333 & 0.05 (0) & 0.97 (0) & 0.06 (0) & 10000 & 4 & 6667 & 0.02 (0) & 0.95 (0) & 0.03 (0) \\ 
   \hline
1000 & 8 & 100 & 0.51 (2) & 0.06 (0) & 0.01 (0) & 5000 & 3 & 500 & 0.02 (0) & 0.75 (3) & 0.01 (0) & 10000 & 5 & 1000 & 0.02 (0) & 0.65 (3) & 0.01 (0) \\ 
  1000 & 8 & 200 & 0.29 (1) & 0.20 (2) & 0.02 (0) & 5000 & 3 & 1000 & 0.02 (0) & 0.89 (2) & 0.01 (0) & 10000 & 5 & 2000 & 0.02 (0) & 0.79 (2) & 0.01 (0) \\ 
  1000 & 8 & 250 & 0.23 (1) & 0.31 (2) & 0.03 (0) & 5000 & 3 & 1250 & 0.02 (0) & 0.91 (1) & 0.02 (0) & 10000 & 5 & 2500 & 0.02 (0) & 0.83 (2) & 0.01 (0) \\ 
  1000 & 8 & 333 & 0.16 (1) & 0.53 (4) & 0.04 (0) & 5000 & 3 & 1667 & 0.02 (0) & 0.94 (1) & 0.02 (0) & 10000 & 5 & 3333 & 0.02 (0) & 0.87 (1) & 0.01 (0) \\ 
  1000 & 8 & 500 & 0.10 (1) & 0.86 (2) & 0.06 (0) & 5000 & 3 & 2500 & 0.03 (0) & 0.96 (1) & 0.03 (0) & 10000 & 5 & 5000 & 0.02 (0) & 0.92 (1) & 0.02 (0) \\ 
  1000 & 8 & 667 & 0.08 (1) & 0.70 (5) & 0.04 (0) & 5000 & 3 & 3333 & 0.03 (0) & 0.97 (0) & 0.05 (0) & 10000 & 5 & 6667 & 0.02 (0) & 0.94 (0) & 0.02 (0) \\ 
   \hline
2000 & 2 & 200 & 0.05 (0) & 0.64 (4) & 0.01 (0) & 5000 & 4 & 500 & 0.03 (0) & 0.61 (3) & 0.01 (0) & 10000 & 6 & 1000 & 0.02 (0) & 0.62 (3) & 0.00 (0) \\ 
  2000 & 2 & 400 & 0.03 (0) & 0.88 (1) & 0.02 (0) & 5000 & 4 & 1000 & 0.02 (0) & 0.84 (2) & 0.01 (0) & 10000 & 6 & 2000 & 0.02 (0) & 0.75 (2) & 0.01 (0) \\ 
  2000 & 2 & 500 & 0.03 (0) & 0.92 (1) & 0.03 (0) & 5000 & 4 & 1250 & 0.02 (0) & 0.88 (2) & 0.01 (0) & 10000 & 6 & 2500 & 0.02 (0) & 0.79 (2) & 0.01 (0) \\ 
  2000 & 2 & 667 & 0.04 (0) & 0.95 (1) & 0.04 (0) & 5000 & 4 & 1667 & 0.02 (0) & 0.91 (1) & 0.02 (0) & 10000 & 6 & 3333 & 0.02 (0) & 0.84 (1) & 0.01 (0) \\ 
  2000 & 2 & 1000 & 0.05 (0) & 0.97 (1) & 0.06 (0) & 5000 & 4 & 2500 & 0.03 (0) & 0.95 (1) & 0.03 (0) & 10000 & 6 & 5000 & 0.02 (0) & 0.90 (1) & 0.01 (0) \\ 
  2000 & 2 & 1333 & 0.06 (0) & 0.98 (0) & 0.08 (0) & 5000 & 4 & 3333 & 0.03 (0) & 0.96 (1) & 0.04 (0) & 10000 & 6 & 6667 & 0.02 (0) & 0.93 (1) & 0.02 (0) \\ 
   \hline
2000 & 3 & 200 & 0.09 (1) & 0.40 (5) & 0.01 (0) & 5000 & 5 & 500 & 0.04 (0) & 0.56 (4) & 0.01 (0) & 10000 & 8 & 1000 & 0.03 (0) & 0.56 (2) & 0.00 (0) \\ 
  2000 & 3 & 400 & 0.04 (1) & 0.78 (4) & 0.02 (0) & 5000 & 5 & 1000 & 0.02 (0) & 0.81 (2) & 0.01 (0) & 10000 & 8 & 2000 & 0.02 (0) & 0.70 (2) & 0.01 (0) \\ 
  2000 & 3 & 500 & 0.04 (1) & 0.85 (3) & 0.02 (0) & 5000 & 5 & 1250 & 0.02 (0) & 0.86 (3) & 0.01 (0) & 10000 & 8 & 2500 & 0.02 (0) & 0.74 (2) & 0.01 (0) \\ 
  2000 & 3 & 667 & 0.04 (1) & 0.90 (2) & 0.03 (0) & 5000 & 5 & 1667 & 0.02 (0) & 0.90 (2) & 0.02 (0) & 10000 & 8 & 3333 & 0.02 (0) & 0.80 (2) & 0.01 (0) \\ 
  2000 & 3 & 1000 & 0.04 (0) & 0.94 (1) & 0.05 (0) & 5000 & 5 & 2500 & 0.02 (0) & 0.94 (1) & 0.02 (0) & 10000 & 8 & 5000 & 0.02 (0) & 0.87 (1) & 0.01 (0) \\ 
  2000 & 3 & 1333 & 0.05 (0) & 0.96 (1) & 0.06 (0) & 5000 & 5 & 3333 & 0.03 (0) & 0.96 (1) & 0.03 (0) & 10000 & 8 & 6667 & 0.02 (0) & 0.91 (1) & 0.02 (0) \\ 
   \hline
\hline
\end{tabular}

}
\caption{Simulations for the generative model described in \eqref{eq:tau0_setup}, while varying $s$ and $n$. The results presented in Table \ref{tab:tau0_simu} are a subset of these results. The numbers in parentheses indicate the (rounded) estimates for standard sampling error for the last printed digit, obtained by aggregating performance over 10 simulation replications. Mean-squared error (MSE) and coverage denote performance for estimating $\tau$ on a random test set; the variance column denotes the mean variance estimate obtained by the infinitesimal jackknife on the test set. Target coverage is 0.95. We always grew $B = n$ trees.}
\label{tab:bigsimu}
\end{table}
}
\restoregeometry

\section{Additional Simulations}
\label{sec:more_simu}

In Table \ref{tab:bigsimu}, we expand on the simulation results given in Table \ref{tab:tau0_simu}, and present results for data generated according to \eqref{eq:tau0_setup} all while varying $s$ and $n$. As expected, we observe a bias-variance trade-off: the causal forest is more affected by bias when $s$ is small relative to $n$, and by variance when $s$ is larger relative to $n$.
Reassuringly, we observe that our confidence intervals obtain close-to-nominal coverage when the mean-squared error matches the average variance estimate $\hsigma^2(X)$ generated by the infinitesimal jackknife, corroborating the hypothesis that failures in coverage mostly arise when the causal forest is bias- instead of variance-dominated.

\begin{table}[t]
\centering
\begin{tabular}{|rr|ccc|ccc|}
\hline
 & & \multicolumn{3}{c|}{mean-squared error} & \multicolumn{3}{c|}{coverage} \\
 \hline
$q$ & $d$ & CF & 10-NN & 100-NN & CF & 10-NN & 100-NN \\ 
  \hline
  2 & 6 & 0.04 (0) & 0.24 (0) & 0.13 (0) & 0.92 (1) & 0.92 (0) & 0.59 (1) \\ 
  4 & 6 & 0.06 (0) & 0.22 (0) & 0.07 (0) & 0.87 (1) & 0.93 (0) & 0.72 (1) \\ 
  6 & 6 & 0.08 (0) & 0.22 (0) & 0.05 (0) & 0.75 (1) & 0.93 (0) & 0.78 (1) \\ \hline
  2 & 12 & 0.04 (0) & 0.38 (0) & 0.34 (0) & 0.86 (1) & 0.88 (0) & 0.45 (0) \\ 
  4 & 12 & 0.08 (0) & 0.30 (0) & 0.18 (0) & 0.76 (1) & 0.90 (0) & 0.51 (1) \\ 
  6 & 12 & 0.12 (0) & 0.26 (0) & 0.12 (0) & 0.59 (1) & 0.91 (0) & 0.59 (1) \\ 
   \hline
\end{tabular}
\caption{Results for a data-generating design where we vary both the number of signal features $q$ and the number of ambient features $d$. All simulations have $n = 5,000$, $B = 2,000$ and a minimum leaf size of 1, and are aggregated over 20 simulation replicates.}
\label{tab:supp}
\end{table}

Finally, all our experiments relied on settings with strong, low-dimensional structure
that forests could pick up on to improve over $k$-NN matching. This intuition is
formally supported by, e.g., the theory developed by \citet{biau2012analysis}.
Here, we consider how forests perform when the signal is spread out over a larger number of
features, and so forests have less upside over other methods. We find that---as expected---they
do not improve much over baselines.
Specifically, we generate data with a treatment effect function
\begin{equation}
\tau(x) = \frac{4}{q} \sum_{j = 1}^q \p{\frac{1}{1 + e^{-12 \p{x_j - 0.5}}} - \frac{1}{2}}, 
\end{equation}
where we vary both the number of signal dimensions $q$ and ambient dimensions $d$.
As seen in Table \ref{tab:supp}, forests vastly improve over $k$-NN in terms of mean-squared
error when $q$ is much smaller than $d$, but that this advantage decreases when $d$ and $q$
are comparable; and actually do worse than $k$-NN when we have a dense signal with $d = q = 6$.
When the signal is dense, all surveyed methods have bad coverage except for 10-NN which, as always,
simply has very wide intervals.

\section{Is Honesty Necessary for Consistency?}
\label{sec:honesty}

Our honesty assumption is the largest divergence between our framework
and main-stream applications of random forests. Following \citet{breiman2001random}, almost
all practical implementations of random forests are not honest. Moreover, there has been a stream
of recent work providing theoretical guarantees for adaptive random forests:
\citet{scornet2015consistency} establish risk consistency under assumptions on the
data-generating function, i.e., they show that the test-set error of forests
converges asymptotically to the Bayes risk,
\citet{mentch2014ensemble} provide results about uncertainty quantification, and
\citet{wager2015uniform} find that adaptive trees with growing leaves are in general consistent
under fairly weak conditions and provide bounds on the decay rate of their bias.

\begin{figure}
\centering
\includegraphics[width=\figw]{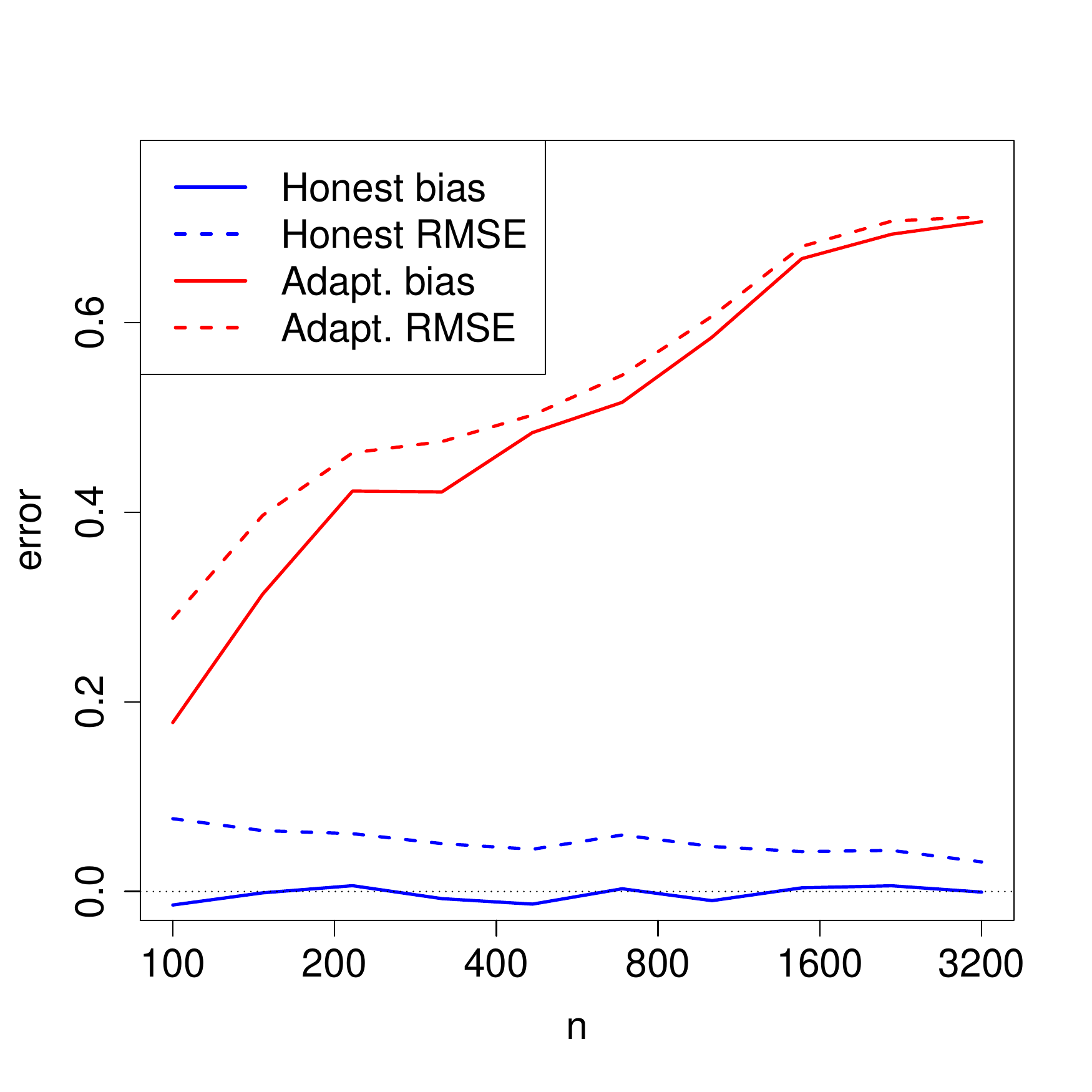}
\caption{Comparison of the performance of honest and adaptive causal forests when predicting at
$x_0 = (0, \, 0, \, \ldots, \, 0)$, which is a corner of the support of the features $X_i$. Both forests have
$B = 500$ trees, and use a leaf-size of $k = 1$. We use a subsample size $s = n^{0.8}$ for adaptive
forests and $s = 2 \, n^{0.8}$ for honest forests. All results are averaged over 40 replications; we report
both bias and root-mean-squared error (RMSE).}
\label{fig:honesty}
\end{figure}

However, if we want pointwise centered asymptotic Gaussianity results, then honesty
appears to be necessary. Consider the following simple example, where there is no treatment
heterogeneity---and in fact $X$ is independent of $W$ and $Y$. We are in a randomized controlled
trial, where $X \sim \text{Uniform}([0, \, 1]^p)$ with $p = 10$ and $W \sim \text{Bernoulli}(0.5)$.
The distribution of $Y$ is $Y_i = 2 W_i A_i + \varepsilon_i$, where $\varepsilon_i \sim \nn\p{0, \, 0.1^2}$
and $A_i \sim \text{Bernoulli}(0.05)$. Thus, the treatment effect is $\tau(x) = 0.1$ for all $x \in [0, \, 1]^p$.
Our goal is to estimate the treatment effect $\tau(x_0)$ at $x_0 = (0, \, 0, \, \ldots, \, 0)$.

Results from running both honest and adaptive forests are shown in Figure \ref{fig:honesty}. We see that
honest forests are unbiased regardless of $n$, and their mean-squared error decreases with
sample size, as expected. Adaptive forests, in contrast, perform remarkably badly. They have
bias that far exceeds the intrinsic sampling variation; and, furthermore, this bias \emph{increases} with $n$.
What is happening here is that CART trees aggressively seek to separate outliers (``$Y_i \approx 1$'') from
the rest of the data (``$Y_i \approx 0$'') and, in doing so, end up over-representing outliers in the corners
of the feature space. As $n$ increases, it appears that adaptive forests have more opportunities to
push outliers into corners of features space and so the bias worsens. This phenomenon is not restricted
to causal forests; an earlier technical report \citep{wager2014asymptotic} observed the same phenomenon
in the context of plain regression forests.  Honest trees do not have this problem, as we do not know
where the $\ii$-sample outliers will be when placing splits using only the $\jj$-sample.
Thus, it appears that adaptive CART forests are pointwise biased in corners of $x$-space.

Finally, we note that this bias phenomenon does not contradict existing consistency results
in the literature. \citet{wager2015uniform} prove that this bias phenomenon discussed above
can be averted if we use a minimum leaf-size $k$ that grows with $n$ (in contrast, Figure \ref{fig:honesty}
uses $k = 1$). However, their bounds on the bias decays slower than the sampling variance of random forests,
and so their results cannot be used to get centered confidence intervals.

Meanwhile, \citet{scornet2015consistency} prove that forests are risk-consistent at an average test point, and,
in fact, the test set error of adaptive forests does decay in the setting of Figure \ref{fig:honesty} 
as the sample size $n$ grows (although
honest forests still maintain a lower test set error). The reason test set error can go to zero despite the
bias phenomenon in Figure \ref{fig:honesty} is that, when $n$ gets large,
almost all test points will be far enough from corners that they will not be affected by the phenomenon
from Figure \ref{fig:honesty}.

\subsection{Adaptive versus Honest Predictive Error}

The discussion above implies that the theorems proved in this paper are not valid for
adaptive forests. That being said, it still remains interesting to ask whether our use
of honest forest hurts us in terms of mean-squared error at a random test point, as
in the formalism considered by, e.g., \citet{scornet2015consistency}.
In this setting, \citet{denil2014narrowing} showed that honesty can hurt the performance
of regression forests on some classic datasets from the UCI repository; however, in a causal
inference setting, we might be concerned that the risk of overfitting with adaptive forests
is higher since our signals of interest are often quite weak.

We compare the performance of honest and adaptive forests in the setting of
Table \ref{tab:tau0_simu}, with $d = 8$. Here, if we simply run adaptive forests
out-of-the-box with the usual minimum leaf size parameter $k = 1$, they do
extremely badly; in fact, they do worse than 50 nearest neighbors. However, if
we are willing to increase the minimum leaf size, their performance improves.

\begin{figure}[t]
\centering
\includegraphics[width=\figw]{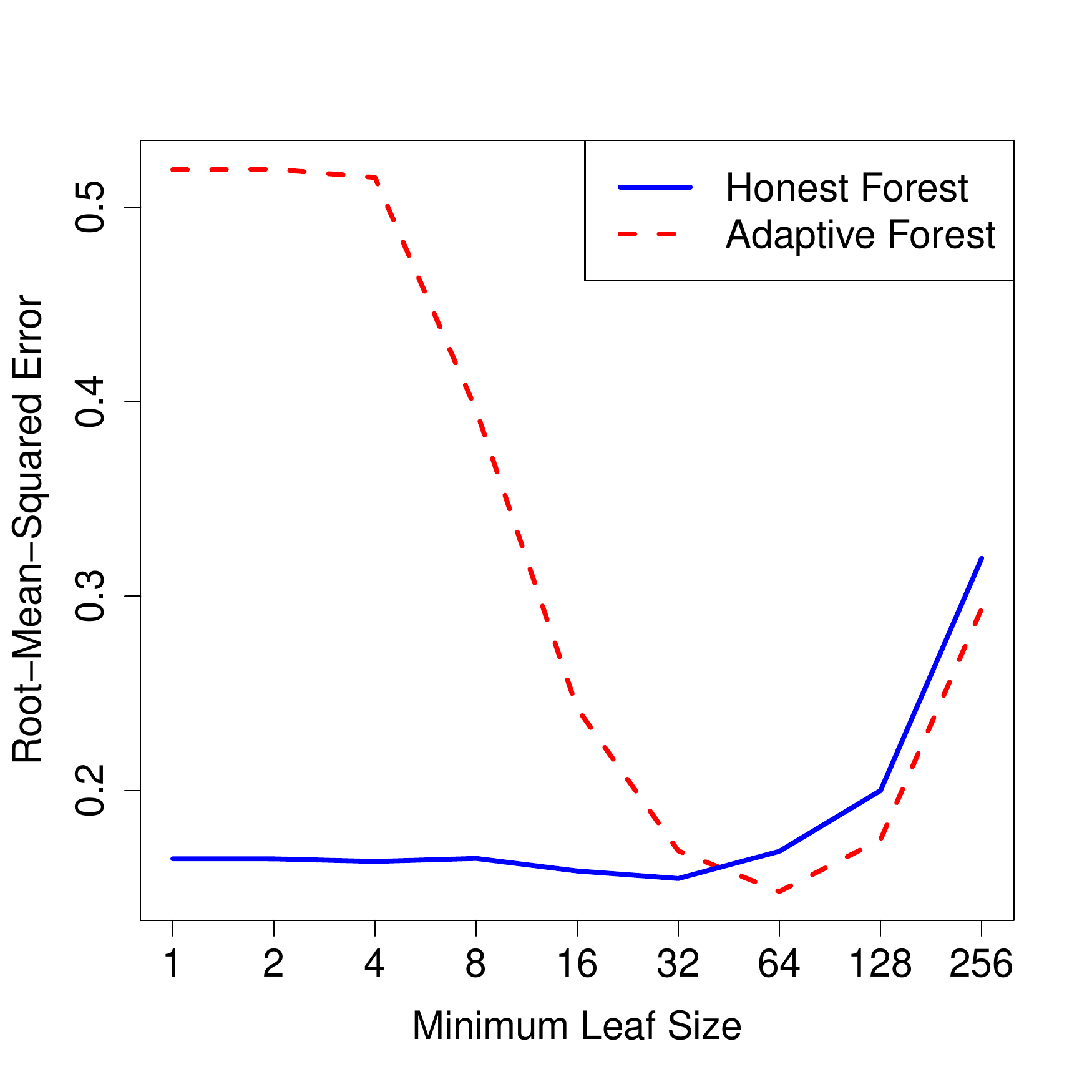}
\caption{Comparison of the root-mean-squared error of honest and adaptive forests in the setting of Table \ref{tab:tau0_simu}, with $d = 8$. Honest forests use $s = 2500$ (i.e., $\abs{\ii} = 1250$) while adaptive forests use $s = 1250$, such that both methods grow trees of the same depth. Both forests have $B = 500$, and results are averaged over 100 simulation replications.}
\label{fig:adaptMSE}
\end{figure}

Figure \ref{fig:adaptMSE} depicts the root-mean-squared error for both adaptive and
honest forests over a wide range of choices for the minimum leaf size parameter $k$.
We see that, at their best, both methods do comparably. However, honest forests
attain good performance over a wide range of choices for $k$, including our default
choice $k = 1$, whereas adaptive forests are extremely sensitive to choosing a good value
of $k$.
We also note that the optimum $k = 64$ for adaptive forests is quite far from
standard choices advocated in practice; such as $k = 5$ recommended by 
\citet{breiman2001random} for regression forests, $k = 7$ in the \texttt{cforest} function used by
\citet{strobl2007bias}, or $k = 10$ recommended by \citet{meinshausen2006quantile}
for quantile regression. Thus, it appears that accurately tuning adaptive forests
in this setting may present a challenge and, overall, a practitioner may prefer honest
forests even based on their mean-squared error characteristics alone.

%
%

\section{Proofs}
\label{sec:proofs}

{\bf Notation.} Throughout the appendix we use the following notation to describe asymptotic scalings:
$f(s) \sim g(s)$ means that $\lim_{s \rightarrow \infty} f(s)/g(s) = 1$,
$f(s) \gtrsim g(s)$ means that $\liminf_{s \rightarrow \infty} f(s)/g(s) \geq 1$ and $f(s) \lesssim g(s)$ is analogous,
$f(s) = \oo(g(s))$ means that $f(s) \lesssim C \, g(s)$ for some $C > 0$,
$f(s) = \Omega(g(s))$ means that $f(s) \gtrsim c \, g(s)$ for some $c > 0$,
and finally $f(s) =  o(g(s))$ means that $\limsup_{s \rightarrow \infty} f(s)/g(s) = 0$.

\begin{proof}[Proof of Theorem \ref{theo:intro}]
Given the statements of Theorem \ref{theo:gauss} and Theorem \ref{theo:ij}, it only remains to show that \eqref{eq:intro_gauss} holds with $\EE{\hmu_n(x)}$ replaced with $\mu(x)$. To do so, it suffices to show that $\abs{\EE{\hmu_n\p{x}} - \mu(x)}/\sigma_n\p{x} \rightarrow 0$; the rest follows from Slutsky's lemma.
Now, recall that by Theorem \ref{theo:bias},
$$ \abs{\EE{\hmu_n\p{x}} - \mu(x)} = \oo\p{n^{-\frac{\beta}{2} \frac{\log\p{\p{1 - \alpha}^{-1}}}{\pi^{-1} d \, \log\p{\alpha^{-1}}}}}. $$
Meanwhile, from Theorem \ref{theo:pnn} and the proof of Theorem \ref{theo:gauss}, we see that
$$ \sigma_n^2(x) \gtrsim C_{f, \, d} \, \frac{s}{n} \, \frac{\Var{T}}{\log(s)^d}. $$
By honesty of $T$, \smash{$\Var{T} \gtrsim \Var{Y \cond X = x} / \abs{\cb{i : X_i \in L(x)}} \geq \Var{Y \cond X = x}/ (2k)$}, and so
$$ \sigma_n^2(x) \gtrsim \frac{C_{f, \, d}}{2k} \, \frac{s}{n} \, \frac{\Var{Y \cond X = x}}{\log(s)^d} = \Omega\p{n^{\beta - 1 - \varepsilon}} $$
for any $\varepsilon > 0$. It follows that
$$ \frac{\abs{\EE{\hmu_n\p{x}} - \mu(x)}}{\sigma_n\p{x}} = \oo\p{n^{\frac{1}{2}\p{1 + \varepsilon - \beta \, \p{1 +  \frac{\log\p{\p{1 - \alpha}^{-1}}}{\pi^{-1} d \, \log\p{\alpha^{-1}}}}}}}. $$
The right-hand-side bound converges to 0 for some small enough $\varepsilon > 0$ provided that
$$ \beta > \p{1 +  \frac{\log\p{\p{1 - \alpha}^{-1}}}{\pi^{-1} d \, \log\p{\alpha^{-1}}}}^{-1} = 1 - \p{1 + \frac{d}{\pi} \,  \frac{\log\p{\alpha^{-1}}}{\log\p{\p{1 - \alpha}^{-1}}}}^{-1} = \beta_{\min}. $$
\end{proof}

\subsection{Bounding the Bias of Regression Trees}
\label{sec:proof_bias}

\begin{proof}[Proof of Lemma \ref{lemm:diameter}]
Let $c(x)$ be the number of splits leading to the leaf $L(x)$, and let $c_j(x)$ be the number of these splits along the $j$-th coordinate. By regularity, we know that $s \alpha^{c(x)} \leq 2k - 1$, and so $c(x) \geq \log(s/(2k - 1)) / \log(\alpha^{-1})$. Thus, because the tree is a random split tree, $c_j(x)$ we have the following stochastic lower bound for $c_j(x)$:
\begin{equation}
\label{eq:diam_pf_1}
c_j(x) \overset{d}{\geq} \text{Binom}\p{\frac{\log\p{s/(2k - 1)}}{\log\p{\alpha^{-1}}}; \, \frac{\pi}{d}}.
\end{equation}
By Chernoff's inequality, it follows that
\begin{align}
\label{eq:diam_pf_2}
\PP{c_j\p{x} \leq \frac{\pi}{d} \frac{\log\p{s/(2k - 1)}}{\log\p{\alpha^{-1}}} \p{1 - \eta}}
&\leq \exp\sqb{-\frac{\eta^2}{2} \, \frac{\log\p{s/(2k - 1)}}{\pi^{-1} d\log\p{\alpha^{-1}}}} \\
\notag
&= \p{\frac{s}{2k - 1}}^{-\frac{\eta^2}{2} \, \frac{1}{\pi^{-1} d \log\p{\alpha^{-1}}}}.
\end{align}
Meanwhile, again by regularity, we might expect that
\smash{$\diam_j\p{L\p{x}}$} should less than \smash{$\p{1 - \alpha}^{c_j\p{x}}$}, at least for large $n$.
This condition would hold directly if the regularity condition from Definition \ref{defi:regular} were framed in terms of Lebesgue measure instead of the number of training examples in the leaf; our task is to show that it still holds approximately in our current setup.

Using the methods developed in \citet{wager2015uniform}, and in particular their Lemma 12, we can verify that, with high probability and simultaneously for all but the last $\oo\p{\log\log n}$ parent nodes above $L(x)$,  the number of training examples inside the node divided by $n$ is within a factor $1 + o(1)$ of the Lebesgue measure of the node. From this, we conclude that, for large enough $s$, with probability greater than $1 - 1/s$
$$ \diam_j\p{L\p{x}} \leq \p{1 - \alpha + o(1)}^{\p{1 + o(1)} c_j\p{x}}, $$
or, more prosaically, that
$$ \diam_j\p{L\p{x}} \leq \p{1 - \alpha}^{0.991 \, c_j\p{x}}. $$
Combining this results with the above Chernoff bound yields the desired inequality.
Here, replacing $0.991$ with $0.99$ in the bound lets us ignore the $1/s$ asymptotic failure probability of the concentration result used above.

Finally, we note that with double-sample trees, all the ``$s$'' terms in the above argument need to be replaced by ``$s/2$''; this additional factor 2, however, does not affect the final result.
\end{proof}

\begin{proof}[Proof of Theorem \ref{theo:bias}]
We begin with two observations. First, by honesty,
$$ \EE{T\p{x; \, Z}} - \EE{Y \cond X = x} = \EE{\EE{Y \cond X \in L\p{x}} - \EE{Y \cond X = x}}. $$
Second, by Lipschitz continuity of the conditional mean function, 
$$ \abs{\EE{Y \cond X \in L\p{x}} - \EE{Y \cond X = x}} \leq C \diam\p{L\p{x}}, $$
where $C$ is the Lipschitz constant. Thus, in order to bound the bias under both Lipschitz and honesty assumptions, it suffices to bound the average diameter of the leaf $L(x)$.

To do so, we start by plugging in $\eta = \sqrt{\log((1 - \alpha)^{-1}}$ in the bound from Lemma \ref{lemm:diameter}. Thanks our assumption that $\alpha \leq 0.2$, we see that $\eta \leq 0.48$ and so $0.99 \cdot (1 - \eta) \geq 0.51$; thus, a union bound gives us that, for large enough $s$,
$$\PP{\diam\p{L(x)} \geq \sqrt{d} \p{\frac{s}{2k - 1}}^{-0.51 \, \frac{\log\p{\p{1 - \alpha}^{-1}}}{\log\p{\alpha^{-1}}} \, \frac{\pi}{d}}} \leq d \, \p{\frac{s}{2k - 1}}^{-\frac{1}{2} \frac{\log\p{\p{1 - \alpha}^{-1}}}{\log\p{\alpha^{-1}}} \, \frac{\pi}{d}}.$$
The Lipschitz assumption lets us bound the bias on the event that that $\diam(L(x))$ satisfies this bound. Thus, for large $s$, we find that
\begin{align*}
 \abs{\EE{T\p{x; \, Z}} - \EE{Y \cond X = x}} &\lesssim  d \, \p{\frac{s}{2k - 1}}^{-\frac{1}{2} \frac{\log\p{\p{1 - \alpha}^{-1}}}{\log\p{\alpha^{-1}}} \, \frac{\pi}{d}} \\
 &\hspace{-17mm}\times \p{\sup_{x \in [0, \, 1]^d} \cb{\EE{Y \cond X = x}} - \inf_{x \in [0, \, 1]^d}\cb{\EE{Y \cond X = x}}},
\end{align*}
where $\sup_x\EE{Y \cond X = x} - \inf_x\EE{Y \cond X = x} = \oo\p{1}$ because
the conditional mean function is Lipschitz continuous.
Finally, since a forest is just an average of trees, the above result also holds for $\hmu(x)$.
\end{proof}

\subsection{Bounding the Incrementality of Regression Trees}
\label{sec:proof_incr}

\begin{proof}[Proof of Lemma \ref{lemm:pnn}]
First, we focus on the case where $f$ is constant, i.e., the features $X_i$ have a uniform distribution over $[0, \, 1]^d$. To begin, recall that the $S_i$ denote selection weights
$$ T\p{x; \, Z} = \sum_{i = 1}^s S_i Y_i \ \text{ where } \
S_i =
\begin{cases}
\abs{\cb{i : X_i \in L(x; \, Z)}}^{-1} & \text{ if } X_i \in L(x; \, Z), \\
0 & \text{ else,}
\end{cases}$$
where $L(x; \, Z)$ denotes the leaf containing $x$.
We also define the quantities 
$$P_i = 1\p{\cb{\text{$X_i$ is a $k$-PNN of $x$}}}.$$
Because $T$ is a $k$-PNN predictor, $P_i = 0$ implies that $S_i = 0$, and, moreover, $\abs{\cb{i : X_i \in L(x; \, Z)}} \geq k$; thus, we can verify that
$$ \EE{S_1 \cond Z_1} \leq \frac{1}{k} \, \EE{P_1 \cond Z_1}. $$
The bulk of the proof involves showing that
\begin{equation}
\label{eq:wbound}
\PP{\EE{P_1 \cond Z_1} \geq \frac{1}{s^2}} \lesssim k \, \frac{2^{d+1} \log\p{s}^d}{\p{d - 1}!} \, \frac{1}{s};
\end{equation}
by the above argument, this immediately implies that
\begin{equation}
\label{eq:sbound}
\PP{\EE{S_1 \cond Z_1} \geq \frac{1}{k \, s^2}}
\lesssim k \,  \frac{2^{d+1} \log\p{s}^d}{\p{d - 1}!} \, \frac{1}{s}.
\end{equation}
Now, by exchangeability of the indices $i$, we know that
$$\EE{\EE{S_1 \cond Z_1}} = \EE{S_1} = \frac{1}{s} \EE{\sum_{i = 1}^s S_i} = \frac{1}{s}, $$ moreover, we can verify that
$$ \PP{\EE{S_1 \cond Z_1} \geq \frac{1}{k s^2}}\EE{\EE{S_1 \cond Z_1} \cond \EE{S_1 \cond Z_1} \geq \frac{1}{k s^2}} \sim \frac{1}{s}. $$
By Jensen's inequality, we then see that
\begin{align*}
\EE{\EE{S_1 \cond Z_1}^2}
&\geq \PP{\EE{S_1 \cond Z_1} \geq \frac{1}{ks^2}}\EE{\EE{S_1 \cond Z_1} \cond \EE{S_1 \cond Z_1} \geq \frac{1}{ks^2}}^2 \\
&\sim \frac{s^{-2}}{\PP{\EE{S_1 \cond Z_1} \geq \frac{1}{k s^2}}}
\end{align*}
which, paired with \eqref{eq:sbound}, implies that
$$ \EE{\EE{S_1 \cond Z_1}^2}
\gtrsim \frac{\p{d - 1}!}{2^{d+1} \log\p{s}^d} \, \frac{1}{k \, s}. $$
This is equivalent to \eqref{eq:pnn_bound} because $\EE{\EE{S_1 \cond Z_1}}^2 = 1/s^2$ is negligibly small.

We now return to establishing \eqref{eq:wbound}. Recall that $X_1, \, ..., \, X_s$ are independently and uniformly distributed over $[0, \, 1]^d$, and that we are trying to find points that are $k$-PNNs of a prediction point $x$. For now, suppose that $x = 0$. We know that $X_1$ is a $k$-PNN of 0 if and only if there are at most $2k-2$ other points $X_i$ such that $X_{ij} \leq X_{1j}$ for all $j = 1, \, ..., \, d$ (because the $X$ have a continuous density, there will almost surely be no ties). Thus,
\begin{align}
\label{eq:replace_unif}
\EE[x = 0]{P_1 \cond Z_1}
&= \PP{\text{Binomial}\p{s - 1; \, \prod_{j = 1}^d X_{1j}} \leq 2k - 2} \\
\notag
&\leq \binom{s - 1}{2k - 2} \p{1 - \prod_{j = 1}^d X_{1j}}^{s - 2k + 1} \\
\notag
&\leq s^{2k - 2}  \p{1 - \prod_{j = 1}^d X_{1j}}^{s - 2k + 1},
\end{align}
where the second inequality can be understood as a union bound over all sets of
$s - 2k + 1$ Bernoulli realizations that could be simultaneously 0.
We can check that $X_{1j} \eqd e^{-E_j}$ where $E_j$ is a standard exponential random variable, and so
$$ \EE[x = 0]{P_1 \cond Z_1} \, \overset{d}{\leq} \, s^{2k - 2} \p{1 - \exp\left[-\sum_{j = 1}^d E_j\right]}^{s - 2k + 1}, $$
where $A \, \overset{d}{\leq} \, B$ means that $B - A \geq 0$ almost surely.
Thus,
\begin{align*}
&\PP[x = 0]{\EE{P_1 \cond Z_1} \geq \frac{1}{s^2}} \\
&\ \ \ \ \leq \PP{s^{\p{2k - 2}}\p{1 - \exp\left[-\sum_{j = 1}^d E_j\right]}^{s - 2k + 1} \geq \frac{1}{s^2}} \\
& \ \ \ \ =\PP{\exp\left[-\sum_{j = 1}^d E_j\right] \leq 1 - \p{\frac{1}{s^{2k}}}^\frac{1}{s - 2k + 1}} \\
& \ \ \ \ =\PP{\sum_{j = 1}^d E_j \geq - \log\p{1 - \exp\left[-2k \, \frac{\log\p{s}}{s - 2k + 1}\right]}}.
\end{align*}
Notice that this quantity goes to zero as $s$ gets large. The sum of $d$ standard exponential random variables has a gamma distribution with shape $d$ and scale 1, and
$$  \PP{\sum_{j = 1}^d E_j \geq c} = \frac{\Gamma\p{d, \, c}}{\p{d - 1}!}, $$
where $\Gamma$ is the upper incomplete gamma function. It is well known that
$$ \lim_{c \rightarrow \infty} \frac{\Gamma\p{d, \, c}}{c^{d - 1} \, e^{-c}} = 1, $$
and so
\begin{align*}
&\PP[x = 0]{\EE{P_1 \cond Z_1} \geq \frac{1}{s^2}} \\
&\ \ \ \ \ \ \ \ \lesssim \frac{\p{- \log\p{1 - \exp\left[-2k \, \frac{\log\p{s}}{s - 2k + 1}\right]}}^{d - 1}  \, \p{1 - \exp\left[-2k \, \frac{\log\p{s}}{s - 2k + 1}\right]}}{\p{d - 1}!}.
\end{align*}
We can check that
$$ 1 - \exp\left[-2k \, \frac{\log\p{s}}{s - 2k + 1}\right] \sim 2k \, \frac{\log\p{s}}{s}, $$
letting us simplify the above expression to
\begin{equation}
\label{eq:zero_bound}
\PP[x = 0]{\EE{P_1 \cond Z_1} \geq \frac{1}{s^2}} \lesssim \frac{2k}{\p{d - 1}!} \, \frac{\log\p{s}^d}{s}.
\end{equation}
We thus have obtained a tight expression for our quantity of interested for a prediction point at $x = 0$.

In the case $x \neq 0$, the ambient space around $x$ can be divided into $2^d$ quadrants. In order to check whether $X_i$ is a PNN, we only need to consider other points in the same quadrant, as no point in a different quadrant can prevent $X_i$ from being a PNN. Now, index the quadrants by $l = 1, \, ..., \, 2^d$, and let $v_l$ be the volume of the $l$-th quadrant. By applying \eqref{eq:zero_bound} on the $l$-th quadrant alone, we see that the probability of $\EE{P_1 \cond Z_1} \geq \frac{1}{s^2}$ given that $X_1$ is in the $l$-the quadrant is asymptotically bounded on the order of
$$ \frac{2k}{\p{d - 1}!} \, \frac{\log\p{s}^d}{v_l \, s}. $$
Summing over all quadrants, we find that
\begin{align*}
\PP[x = 0]{\EE{P_1 \cond Z_1} \geq \frac{1}{s^2}}
&\lesssim \sum_{\{l : v_l > 0\}} v_l \,  \frac{2k}{\p{d - 1}!} \, \frac{\log\p{s}^d}{v_l \, s} \\
&= \left|\{l : v_l > 0\}\right|  \frac{2k}{\p{d - 1}!} \, \frac{\log\p{s}^d}{s} \\
&\leq  k\, \frac{2^{d+1}}{\p{d - 1}!} \, \frac{\log\p{s}^d}{s},
\end{align*}
thus establishing \eqref{eq:wbound}.
Finally, to generalize to bounded densities $f$, we note that if $f(x) \leq C$ for all $x \in [0, \, 1]^d$, then
$$ \EE[x = 0]{P_1 \cond Z_1} \leq \PP{\text{Binomial}\p{s - 1; \, C \prod_{j = 1}^d X_{1j}} \leq 2k - 2}, $$
and the previous argument goes though.
\end{proof}

\begin{proof}[Proof of Theorem \ref{theo:pnn}]

Our main task is to show that if $T$ is a regular tree and $\Var{Y \cond X = x} > 0$, then
\begin{equation}
\label{eq:v1bound}
\Var{\EE{T\p{x; \, Z} \cond Z_1}} \gtrsim \Var{\EE{S_1 \cond Z_1}} \Var{Y \cond X = x}.
\end{equation}
Given this result, Lemma \ref{lemm:pnn} then implies that
\begin{align*}
\Var{\EE{T\p{x; \, Z} \cond Z_1}}
&\gtrsim \frac{1}{k} \,\frac{\nu\p{s}}{s} \, \Var{Y \cond X = x}.
\end{align*}
Moreover, by Theorem \ref{theo:bias}, we know that
\begin{align*}
&\EE{Y_{i} \cond X_{i} \in L\p{x; \, Z}} \rightarrow_p \EE{Y \cond X = x}, \eqand\\
&\EE{Y_{i}^2 \cond X_{i} \in L\p{x; \, Z}} \rightarrow_p \EE{Y^2 \cond X = x},
\end{align*}
and so
$$k \Var{T\p{x; \, Z}} \leq \abs{\cb{i : X_i \in L\p{x; \, Z}}} \cdot \Var{T\p{x; \, Z}} \rightarrow_p \Var{Y \cond X = x}, $$
because $k$ remains fixed while the leaf size gets smaller.
Thus, we conclude that
$$ \frac{\Var{\proj{T}\p{x; \, Z}}}{\Var{T\p{x; \, Z}}} \gtrsim k \, \frac{s \, \Var{\EE{T\p{x; \, Z} \cond Z_1}}}{\Var{Y \cond X = x}} \gtrsim \nu\p{s}, $$
as claimed.

Now, in order to verify \eqref{eq:v1bound}, we first recall that by Lemma \ref{lemm:pnn}
\begin{equation}
\label{eq:target_rate}
\Var{\EE{S_1 \cond Z_1}} = \Omega\p{ \frac{1}{s \, \log\p{s}^{d}}}. 
\end{equation}
Thus, any terms that decay faster than the right-hand-side rate can safely be ignored in establishing \eqref{eq:v1bound}.
We begin by verifying that we can take the leaf $L(x)$ containing $x$ to have a small diameter $\diam(L(x))$.
Define the truncated tree predictor
$$ T'\p{x; \, Z} = T\p{x; \, Z} \, 1\p{\cb{\diam\p{L(x)} \leq s^{-\omega}}}, \where
\omega = \frac{1}{2} \,  \frac{\pi}{d} \,  \frac{\log\p{\p{1 - \alpha}^{-1}}}{\log\p{\alpha^{-1}}}, $$
and define similarly the truncated selection variables \smash{$S_i' = S_i \, 1(\cb{\diam\p{L(x)} \leq s^{-\omega}})$}. Now, thanks to the ANOVA decomposition (3rd line), we see that
\begin{align*}
&\Var{\EE{T'\p{x; \, Z} \cond Z_1} - \EE{T\p{x; \, Z} \cond Z_1}} \\
&\ \ \ \ \ \ = \Var{\EE{T\p{x; \, Z} \, 1\p{\cb{\diam\p{L(x)} > s^{-\omega}}}\cond Z_1}} \\
&\ \ \ \ \ \ \leq \frac{1}{s} \Var{T\p{x; \, Z} \, 1\p{\cb{\diam\p{L(x)} > s^{-\omega}}}} \\
&\ \ \ \ \ \ \leq \frac{\sup_{x\in [0, \, 1]^d}\cb{\EE{Y^2 \cond X = x}}}{s} \, \PP{\diam\p{L(x)} > s^{-\omega}},
\end{align*}
where the $\sup$ term is bounded by Lipschitz-continuity of the second moment of $Y$.
Thus, by Lemma \ref{lemm:diameter}, the variance of the difference between $\EE{T' \cond Z_1}$ and $\EE{T \cond Z_1}$ decays faster than the target rate \eqref{eq:target_rate}, and so
$$ \Var{\EE{T\p{x; \, Z} \cond Z_1}} \sim \Var{\EE{T'\p{x; \, Z} \cond Z_1}}, $$
provided that $T'$ satisfies \eqref{eq:v1bound}, as we will see it does.
By the same argument, we also note that
$$ \Var{\EE{S_1' \cond Z_1}} \sim \Var{\EE{S_1 \cond Z_1}}. $$
We can now proceed to analyze $T'$ instead of $T$.

Recall that our goal is to provide a lower bound on the variance of the expectation
of $T'(x; \, Z)$ conditionally on $Z_1$. First, an elementary decomposition shows that
\begin{align*}
&\Var{\EE{T'(x; \, Z) \cond Z_1}} \\
&\ \ \ \ \ \ \ \ = \Var{\EE{T'(x; \, Z) \cond X_1}} + \Var{\EE{T'(x; \, Z) \cond X_1, \, Y_1} - \EE{T'(x; \, Z) \cond X_1}} \\
&\ \ \ \ \ \ \ \ \geq \Var{\EE{T'(x; \, Z) \cond X_1, \, Y_1} - \EE{T'(x; \, Z) \cond X_1}},
\end{align*}
and so it suffices to provide a lower bound for the latter term. Next we note that, thanks to honesty
as in Definition \ref{defi:honest}, part (a),
and i.i.d. sampling,
$$ \EE{T'(x; \, Z) \cond X_1, \, Y_1} - \EE{T'(x; \, Z) \cond X_1} = \EE{S_1' \cond X_1}\p{Y_1 - \EE{Y_1 \cond X_1}}. $$
Because of honesty and our Lipschitz assumption, the above implies that
\begin{align*}
&\Var{\EE{T'(x; \, Z) \cond X_1, \, Y_1} - \EE{T'(x; \, Z) \cond X_1}} \\
&\ \ \ \ \ \ \ \ = \Var{\EE{S_1' \cond X_1}\p{Y_1 - \mu(x)}} + \oo\p{\EE{S_1'^2} s^{-2\omega}},
\end{align*}
where we note that the error term decays as $s^{-(1 + 2\omega)}$, which will be prove to
be negligible relative to the main term.
Finally, we can verify that
\begin{align}
\label{eq:expansion_1}
&\Var{\EE{S'_1 \cond X_1} \p{Y_1 - \mu(x)}}   \\
\notag
 &\ \ \ \ \ \ \ \  =\EE{\EE{S'_1 \cond X_1}^2 \, \EE{\p{Y_1 - \mu\p{x}}^2 \cond X_1}} 
 - \EE{\EE{S'_1 \cond X_1}\EE{Y_1 - \mu(x) \cond X_1}}^2.
\end{align}
Now, because the first two conditional moments of $Y$ given $X$ are Lipschitz, and since $\EE{S'_1 \cond X_1}$ is 0 for $\Norm{X_1 - x}_2 > s^{-\omega}$ thanks to our truncating argument, we see that
\begin{align*}
\EE{\EE{S'_1 \cond X_1}^2 \, \EE{\p{Y_1 - \mu\p{x}}^2 \cond X_1}} 
&\sim \EE{\EE{S'_1 \cond X_1}^2} \Var{Y \cond X = x} \\
&\sim  \EE{\EE{S_1 \cond X_1}^2} \Var{Y \cond X = x}.
\end{align*}
Meanwhile, the second term in the expansion \eqref{eq:expansion_1} is of order $1/s^2$ and thus negligible.
To recap, we have shown that a version of \eqref{eq:v1bound} holds with $T$ replaced by $T'$;
and so \eqref{eq:v1bound} must also hold thanks to the previously established coupling result.
\end{proof}

\begin{proof}[Proof of Corollary \ref{coro:doublesample}]
For a double sample tree, we start by noting that
\begin{align*}
\Var{\proj{T}} &= s \Var{\EE{T \cond Z_1}} = s \Var{\EE{1\p{\cb{1 \in \ii}} T \cond Z_1} + \EE{1\p{\cb{1 \not\in \ii}} T \cond Z_1}} \\
& \geq \frac{s}{2} \Var{\EE{1\p{\cb{1 \in \ii}} T \cond Z_1}} - s \Var{\EE{1\p{\cb{1 \not\in \ii}} T \cond Z_1}} \\
& \sim \frac{s}{8} \Var{\EE{T \cond Z_1} \cond 1 \in \ii} - \frac{s}{4} \Var{\EE{ T \cond Z_1} \cond 1 \not\in \ii},
\end{align*}
where to verify the last line we note that $\PP{1 \in \ii \cond Z_1} = \lfloor s/2 \rfloor$ regardless of $Z_1$.
Now, an immediate application of Theorem \ref{theo:pnn} shows us that
\begin{equation*}
\left\lfloor s/2 \right\rfloor \Var{\EE{T \cond Z_1} \cond 1 \in \ii} \gtrsim {C_{f, \, d}} \big / {\log\p{s}^{d}} \, \Var{T},
\end{equation*}
which corresponds to the rate we seek to establish.
Meanwhile, by standard results going back to \citet{hoeffding1948class},
$$ \left\lceil s/2\right\rceil \Var{\EE{ T \cond Z_1} \cond 1 \not\in \ii} \leq \Var{\EE{T \cond \cb{Z_j : j \not\in \ii}} \cond \ii}; $$
then, Lemma \ref{lemm:diameter} and the argument used to establish Theorem \ref{theo:bias} imply that 
$$ \Var{\EE{T \cond \cb{Z_j : j \not\in \ii}} \cond \ii} = \oo\p{s^{-\frac{\log\p{(1 - \alpha)^{-1}}}{\log\p{\alpha^{-1}}}\frac{\pi}{d}}}, $$
and so the term arising under the $1 \not\in \ii$ condition is negligibly small.
\end{proof}

\subsection{Properties of Subsampled Incremental Base Learners}
\label{sec:proof_ss}

The results presented in this section rely heavily on the Efron-Stein ANOVA decomposition, summarized here for convenience. Suppose we have any symmetric function $T :\Omega^n \rightarrow \RR$, and suppose that $Z_1, \, ..., \, Z_n$ are independent and identically distributed on $\Omega$ such that $\Var{T(Z_1, \, ..., \, Z_n)} < \infty$. Then \citet{efron1981jackknife} show that there exist functions $T_1, \, ..., \, T_n$ such that
\begin{equation}
\label{eq:anova}
T\p{Z_1, \, ..., \, Z_n} = \EE{T} + \sum_{i = 1}^n T_1\p{Z_i} + \sum_{i < j} T_2\p{Z_i, \, Z_j} + ... + T_n\p{Z_1, \, ..., \, Z_n},
\end{equation}
and that all $2^n - 1$ random variables on the right side of the above expression are all mean-zero and uncorrelated. It immediately follows that
\begin{equation}
\label{eq:anova_var}
\Var{T} = \sum_{k = 1}^n \binom{n}{k} V_k, \where V_k = \Var{T_k\p{Z_1, \, ..., \, Z_k}}.
\end{equation}
For our purposes, it is also useful to note that the H\'ajek projection $\proj{T}$ can be written as
$$ \proj{T}\p{Z_1, \, ..., \, Z_n} = \EE{T} + \sum_{i = 1}^n T_1\p{Z_i}, \eqand \Var{\proj{T}} = n \, V_1. $$
Thus, the ANOVA decomposition provides a convenient abstract framework for analyzing our quantities of interest.

\begin{proof}[Proof of Lemma \ref{lemm:hajek}]
Applying the ANOVA decomposition to the individual trees $T$, we see that a random forest estimator $\hmu(x)$ of the form \eqref{eq:rf_defi} can equivalently be written as
\begin{align*}
&\hmu\p{x; \, Z_1, \, ..., \, Z_n} = \EE{T}
+ \binom{n}{s}^{-1} \Bigg(\binom{n - 1}{s - 1} \sum_{i = 1}^n T_1\p{Z_i} \\
&\ \ \ \ \ \ \ \  + \binom{n - 2}{s - 2} \sum_{i < j} T_2\p{Z_i, \, Z_j} 
 + ... + \sum_{i_1 < ... < i_s} T_s\p{Z_{i_1}, \, ..., \, Z_{i_s}}\Bigg).
\end{align*}
The above formula holds because each training point $Z_i$ appears in $\binom{n - 1}{s - 1}$ out of $\binom{n}{s}$ possible subsamples of size $s$, each pair $\p{Z_i, \, Z_j}$ appears is $\binom{n - 2}{s - 2}$ subsets, etc.

Now, we can also show that the H\'ajek projection of $\hmu$ is
$$ \proj{\hmu}\p{x; \, Z_1, \, ..., \, Z_n} = \EE{T} + \frac{s}{n} \sum_{i = 1}^n T_1\p{Z_i} . $$
As with all projections \citep{van2000asymptotic},
$$ \EE{\p{\hmu(x) - \proj{\hmu}(x)}^2} = \Var{\hmu(x) - \proj{\hmu}(x)}. $$
Recall that the $T_k\p{\cdot}$ are all pairwise uncorrelated. Thus, using the notation $s_k = s \cdot \p{s - 1}  \cdots  \p{s - k}$ it follows that
\begin{align*}
\EE{\p{\hmu(x) - \proj{\hmu}(x)}^2}
&= \sum_{k = 2}^s \p{\frac{s_k}{n_k}}^2 \binom{n}{k}  V_k, \\
&= \sum_{k = 2}^s \frac{s_k}{n_k} \, \binom{s}{k}  V_k, \\
&\leq \frac{s_2}{n_2} \sum_{k = 2}^s \binom{s}{k}  V_k, \\
&\leq \frac{s_2}{n_2}\Var{T},
\end{align*}
where on the last line we used \eqref{eq:anova_var}. We recover the stated result by noticing that $s_2/n_2 \leq s^2/n^2$ for all $2 \leq s \leq n$.
\end{proof}

\begin{proof}[Proof of Theorem \ref{theo:gauss}]
Using notation from the previous lemma, let $\sigma_n^2 = s^2/n \, V_1$ be the variance of $\proj{\hmu}$. We know that
$$\sigma_n^2 = \frac{s}{n} \, sV_1 \leq \frac{s}{n} \, \Var{T}, $$
and so $\sigma_n \rightarrow 0$ as desired. Now, by Theorem \ref{theo:pnn} or Corollary \ref{coro:doublesample}
combined with Lemma \ref{lemm:hajek}, we find that
\begin{align*}
\frac{1}{\sigma_n^2} \EE{\p{\hmu\p{x} - \proj{\hmu}\p{x}}^2}
&\leq \p{\frac{s}{n}}^2 \, \frac{\Var{T}}{\sigma_n^2} \\
&= \frac{s}{n} \, \Var{T} \big/ \Var{\proj{T}} \\
&\lesssim \frac{s}{n} \, \frac{\log\p{s}^{d}}{C_{f, \, d}/4} \\
&\rightarrow 0
\end{align*}
by hypothesis.
Thus, by Slutsky's lemma, it suffices to show that \eqref{eq:gauss} holds for the H\'ajek projection of the random forest $\proj{\hmu}(x)$.

By our definition of $\sigma_n$, all we need to do is check that $\proj{\hmu}$ is asymptotically normal. One way to do so is using the
Lyapunov central limit theorem \citep[e.g.,][]{billingsley2008probability}. Writing
$$ \proj{\hmu}(x) = \frac{s}{n} \sum_{i = 1}^n \p{\EE{T \cond Z_i} - \EE{T}}, $$
it suffices to check that
\begin{equation}
\label{eq:lyapunov}
\limn \sum_{i = 1}^n \EE{\abs{\EE{T \cond Z_i} - \EE{T}}^{2 + \delta}} \, \big/ \, \p{\sum_{i = 1}^n \Var{\EE{T \cond Z_i}}}^{1 + \delta/2} = 0.
\end{equation}
Using notation from Section \ref{sec:incremental}, we write $T = \sum_{i = 1}^n S_i Y_i$.
Thanks to honesty, we can verify that for any index $i > 1$, $Y_i$ is independent of $S_i$
conditionally on $X_i$ and $Z_1$, and so
$$ \EE{T \cond Z_1} - \EE{T} = \EE{S_1 \p{Y_1 - \EE{Y_1 \cond X_1}} \cond Z_1} + \p{\EE{\sum_{i = 1}^n S_i \EE{Y_i \cond X_i} \cond Z_1} - \EE{T}}. $$
Note that the two right-hand-side terms above are both mean-zero.
By Jensen's inequality, we also have that
\begin{align*}
&2^{-\p{1 + \delta}} \, \EE{\abs{\EE{T \cond Z_1} - \EE{T}}^{2 + \delta}}   
\leq \EE{\abs{\EE{S_1 \p{Y_1 - \EE{Y_1 \cond X_1}} \cond Z_1}}^{2 + \delta}} \\
& \ \ \ \ \ + \EE{\abs{ \EE{\sum_{i = 1}^n S_i \EE{Y_i \cond X_i} \cond Z_1}- \EE{T}}^{2 + \delta}}.
\end{align*}
Now, again by honesty, $\EE{S_1 \cond Z_1} = \EE{S_1 \cond X_1}$,
and so our uniform $(2+\delta)$-moment bounds on the distribution of $Y_i$ conditional on $X_i$ implies that
\begin{align*}
 \EE{\abs{\EE{S_1 \p{Y_1 - \EE{Y_1 \cond X_1}} \cond Z_1}}^{2 + \delta}}
&= \EE{\EE{S_1 \cond X_1}^{2+\delta} \abs{Y_1 - \EE{Y_1 \cond X_1}}^{2 + \delta}}\\
&\leq M \EE{\EE{S_1 \cond X_1}^{2+\delta}} \leq M \EE{\EE{S_1 \cond X_1}^{2}},
\end{align*}
because $S_1 \leq 1$.
Meanwhile, because $\EE{Y \cond X = x}$ is Lipschitz,
we can define $u := \sup \cb{\abs{\EE{Y \cond X = x}} : x \in [0, \, 1]^d}$, and see that
\begin{align*}
&\EE{\abs{ \EE{\sum_{i = 1}^n S_i \EE{Y_i \cond X_i} \cond Z_1}- \EE{T}}^{2 + \delta}}
 \leq (2u)^\delta \Var{\EE{\sum_{i = 1}^n S_i \EE{Y_i \cond X_i} \cond Z_1}} \\
&\ \ \ \ \ \ \ \ \leq 2^{1+\delta}u^{2 + \delta} \p{\EE{\EE{S_1 \cond Z_1}^2} + \Var{(n - 1) \EE{S_2 \cond Z_1}}}
\leq (2u)^{2 + \delta} \EE{\EE{S_1 \cond X_1}^2}.
\end{align*}
Thus, the condition we need to check simplifies to
$$ \limn n \, \EE{\EE{S_1 \cond X_1}^{2}} \, \big/ \, \p{n \Var{\EE{T \cond Z_1}}}^{1 + \delta/2} = 0. $$
Finally, as argued in the proofs of Theorem \ref{theo:pnn} and Corollary \ref{coro:doublesample},
\begin{equation*}
\Var{\EE{T \cond Z_1}} = \Omega\p{ \EE{\EE{S_1 \cond X_1}^2} \Var{Y \cond X = x}}.
\end{equation*}
Because $\Var{Y \cond X = x} > 0$ by assumption, we can use Lemma \ref{lemm:pnn}
to conclude our argument, noting that
\begin{align*}
\p{n \, \EE{\EE{S_1 \cond X_1}^2}}^{-\delta/2} 
\lesssim \p{\frac{C_{f, \, d}}{2k} \frac{n}{s \log(s)^d}}^{-\delta/2},
\end{align*}
which goes to 0 thanks to our assumptions on the scaling of $s$ (and the factor 2 comes from potentially using
a double-sample tree).
\end{proof}

\begin{proof}[Proof of Theorem \ref{theo:ij}]
Let $F$ denote the distribution from which we drew $Z_1, \, ..., \, Z_n$. Then, the variance $\sigma_n^2$ of the H\'ajek projection of $\hmu(x)$ is
\begin{align*}
\sigma_n^2
&= \sum_{i = 1}^n \p{\EE[Z \sim F]{\hmu(x) \cond Z_i} - \EE[Z \sim F]{\hmu(x)}}^2 \\
&= \frac{s^2}{n^2}\sum_{i = 1}^n \p{\EE[Z \sim F]{T \cond Z_i} - \EE[Z \sim F]{T}}^2,
\end{align*}
whereas we can check that the infinitesimal jackknife as defined in \eqref{eq:hvij} is equal to
$$ \hVIJ =  \frac{n - 1}{n} \p{\frac{n}{n - s}}^2 \frac{s^2}{n^2} \sum_{i = 1}^n \p{\EE[Z^* \subset \hF]{T \cond Z_1^* = Z_i} - \EE[Z^* \subset \hF]{T}}^2, $$
where $\hF$ is the empirical distribution on $\{Z_1, \, ..., \, Z_n\}$. Recall that we are sampling the $Z^*$ from $\hF$ without replacement.

It is useful to write our expression of interest $\hVIJ$ using the H\'ajek projection $\proj{T}$ of $T$:
\begin{align*}
&\hVIJ =  \frac{n - 1}{n} \p{\frac{n}{n - s}}^2 \frac{s^2}{n^2} \sum_{i = 1}^n \p{A_i + R_i}^2, \where \\
&\ \ \ \ \ A_i = \EE[Z^* \subset \hF]{\proj{T} \cond Z_1^* = Z_i} - \EE[Z^* \subset \hF]{\proj{T}} \eqand \\
&\ \ \ \ \ R_i = \EE[Z^* \subset \hF]{T - \proj{T} \cond Z_1^* =  Z_i}  - \EE[Z^* \subset \hF]{T - \proj{T}}.
\end{align*}
As we show in Lemma \ref{lemm:tech_main}, the main effects $A_i$ give us $\sigma_n^2$, in that
\begin{equation}
\label{eq:ij_main}
\frac{1}{\sigma_n^2}  \frac{s^2}{n^2} \sum_{i = 1}^n A_i^2 \rightarrow_p 1.
\end{equation}
Meanwhile, Lemma \ref{lemm:tech_resid} establishes that the $B_i$ all satisfy
\begin{equation}
\label{eq:ij_resid}
\EE{R_i^2} \lesssim \frac{2}{n}\Var{T\p{x;, \, Z_1, \, ..., \, Z_{s}}},
\end{equation}
and so
\begin{align*}
\EE{\frac{s^2}{n^2} \sum_{i = 1}^n R_i^2}
&\lesssim \frac{2\, s^2}{n^2} \Var{T\p{x;, \, Z_1, \, ..., \, Z_{s}}} \\
&\lesssim \frac{2\, s}{n} \log\p{n}^{d} \sigma^2_n.
\end{align*}
Because all terms are positive and $s  \log\p{n}^{d} / {n}$ goes to zero by hypothesis, Markov's inequality implies that
\begin{equation*}
\frac{1}{\sigma_n^2}  \frac{s^2}{n^2} \sum_{i = 1}^n R_i^2 \rightarrow_p 0.
\end{equation*}
Using Cauchy-Schwarz to bound the cross terms of the form $A_iR_i$, and noting that $\limn n (n - 1)/(n - s)^2 = 1$, we can thus conclude that $\hVIJ / \sigma_n^2$ converges in probability to 1.
\end{proof}

\begin{lemm}
\label{lemm:tech_main}
Under the conditions of Theorem \ref{theo:ij}, \eqref{eq:ij_main} holds.
\proof
We can write
\begin{align*}
A_i
&= \EE[Z^* \subset \hF]{\proj{T} \cond Z_1^* = Z_i} - \EE[Z^* \subset \hF]{\proj{T}} \\
&= \p{1 - \frac{s}{n}} T_1\p{Z_i} + \p{\frac{s - 1}{n - 1} - \frac{s}{n}} \sum_{j \neq i} T_1\p{Z_j},
\end{align*}
and so our sum of interest is asymptotically unbiased for $\sigma^2_n$:
\begin{align*}
\EE{\frac{n - 1}{n} \p{\frac{n}{n - s}}^2 \frac{s^2}{n^2} \sum_{i = 1}^n A_i^2}
&= \frac{s^2}{n} \EE{T_1\p{Z}^2} \\
&= \frac{s}{n} \Var{\proj{T}\p{Z_1, \, ..., \, Z_s}} \\
&= \sigma^2_n.
\end{align*}
Finally, to establish concentration, we first note that the above calculation also implies that
$\sigma_n^{-2} \frac{s^2}{n^2} \sum_{i = 1}^n \p{A_i - T_1(Z_i)}^2 \rightarrow_p 0$.
Meanwhile, following the argumentation in the proof of Theorem \ref{theo:gauss},
we can apply $(2+\delta)$-moment bounds on $Y_i - \EE{Y_i \cond X_i}$ to verify that
$$ \lim_{u \rightarrow \infty} \lim_{n \rightarrow \infty} \p{\EE{T_1^2(Z_1)} - \EE{\min\cb{u, \,  T_1^2(Z_1)}}} = 0, $$
and so we obtain can apply a truncation-based argument to derive a weak law of large
numbers for triangular arrays for $\sigma_n^{-2} \frac{s^2}{n^2} \sum_{i = 1}^n T_1^2(Z_i)$.
\endproof
\end{lemm}

\begin{lemm}
\label{lemm:tech_resid}
Under the conditions of Theorem \ref{theo:ij}, \eqref{eq:ij_resid} holds.
\proof
Without loss of generality, we establish \eqref{eq:ij_resid} for $R_1$. Using the ANOVA decomposition \eqref{eq:anova}, we can write our term of interest as
\begin{align*}
R_1 &= \EE[Z^* \subset \hF]{T - \proj{T} \cond Z_1^* =  Z_1}  - \EE[Z^* \subset \hF]{T - \proj{T}} \\
&= \p{\frac{s - 1}{n - 1} - \binom{s}{2} \big/ \binom{n}{2}} \sum_{i = 2}^n T_2 \p{Z_1, \, Z_i} \\
&\ \ \ \ \ \ \ \ + \p{\binom{s-1}{2} \big/ \binom{n - 1}{2} - \binom{s}{2} \big/ \binom{n}{2}} \sum_{2 \leq i < j \leq n} T_2 \p{Z_i, \,  Z_j} \\
&+  \p{\binom{s-1}{2} \big/ \binom{n - 1}{2} - \binom{s}{3} \big/ \binom{n}{3}} \sum_{2 \leq i < j \leq n} T_3 \p{Z_1, \, Z_i, \,  Z_j} \\
&\ \ \ \ \ \ \ \ + \p{\binom{s-1}{3} \big/ \binom{n - 1}{3} - \binom{s}{3} \big/ \binom{n}{3}} \sum_{2 \leq i < j < k \leq n} T_3 \p{Z_i, \,  Z_j, \, Z_k} \\
&+ \ldots
\end{align*}
Because all the terms in the ANOVA expansion are mean-zero and uncorrelated, we see using notation from \eqref{eq:anova_var} that
\begin{align*}
\EE{R_1^2}
&= \p{n - 1} \p{\frac{s - 1}{n - 1} - \binom{s}{2} \big/ \binom{n}{2}}^2 V_2 \\
&\ \ \ \ \ \ \ \ +\binom{n - 1}{2}\p{\binom{s-1}{2} \big/ \binom{n - 1}{2} - \binom{s}{2} \big/ \binom{n}{2}}^2 V_2\\
&+  \binom{n - 1}{2}\p{\binom{s-1}{2} \big/ \binom{n - 1}{2} - \binom{s}{3} \big/ \binom{n}{3}}^2 V_3 \\
&\ \ \ \ \ \ \ \ + \binom{n - 1}{3}\p{\binom{s-1}{3} \big/ \binom{n - 1}{3} - \binom{s}{3} \big/ \binom{n}{3}}^2 V_3 \\
&+ \ldots
\end{align*}
Recall that
$$ \sum_{k = 1}^s \binom{s}{k} V_k = \Var{T\p{Z_1, \, ..., \, Z_s}}. $$
The above sum is maximized when all the variance is contained in second-order terms, and $\binom{s}{2} V_2 = \Var{T}$. This implies that
\begin{align*}
\EE{R_1^2}
& \lesssim \p{n - 1} \p{\frac{s - 1}{n - 1} - \binom{s}{2} \big/ \binom{n}{2}}^2 \binom{s}{2}^{-1} \Var{T\p{Z_1, \, ..., \, Z_s}} \\
&\sim \frac{2}{n} \Var{T\p{Z_1, \, ..., \, Z_s}},
\end{align*}
thus completing the proof.
\endproof
\end{lemm}

\begin{proof}[Proof of Proposition \ref{prop:finite_sample}]
Let $N^*_i$ denote whether or not the $i$-training example was used for a subsample, as in \eqref{eq:hvij_plain}. For trivial trees
$$ T(x; \, \xi, \, Z_{i_1}, \, ..., \, Z_{i_s}) = \frac{1}{s} \sum_{j = 1}^s Y_{i_j} $$
we can verify that for any $i = 1, \, ..., \, n$, \smash{$\EE[*]{\hmu^*}\EE{N^*_i} = s/n \ \bY$},
\begin{align*}
&\EE[*]{\hmu^* \, N^*_1}
= \frac{s}{n} \p{\frac{Y_i}{s} + \frac{s - 1}{s} \, \frac{n \, \bY - Y_i}{n - 1}} 
= \frac{1}{n} \, \frac{n - s}{n - 1} \, Y_i + \frac{s - 1}{n - 1} \, \bY, \text{ and  } \\
 &\Cov[*]{\hmu^*, \, N^*_i}
 = \frac{1}{n} \, \frac{n - s}{n - 1} \, Y_i + \p{\frac{s - 1}{n - 1} - \frac{s}{n}} \bY 
 = \frac{1}{n - 1} \, \frac{n - s}{n} \p{Y_i - \bY}.
\end{align*}
Thus, we find that
\begin{align*}
\hVIJ
&= \frac{n - 1}{n} \p{\frac{n}{n - s}}^2 \sum_{i = 1}^n \Cov[*]{\hmu^*, \, N^*_i}^2 \\
&= \frac{1}{n \, \p{n - 1}} \sum_{i = 1}^n \p{Y_i - \bY}^2 \\
&= \hV_{simple},
\end{align*}
as we sought to verify.
\end{proof}

\subsection{Extension to Causal Forests}
\label{sec:proof_causal}

\begin{proof}[Proof of Theorem \ref{theo:cmp_forest}]
Our argument mirrors the proof of Theorem \ref{theo:intro}.
The main steps involve bounding the bias of causal forests with an analogue to Theorem \ref{theo:bias} and their incrementality using an analogue to Theorem \ref{theo:pnn}. In general, we find that the same arguments as used with regression forests go through, but the constants in the results get worse by a factor $\varepsilon$ depending on the amount of overlap \eqref{eq:overlap}.
Given these results, the subsampling-based argument from Section \ref{sec:hajek} can be reproduced almost verbatim, and the final proof of Theorem \ref{theo:cmp_forest} is identical to that of Theorem \ref{theo:intro} presented at the beginning of Appendix \ref{sec:proofs}.

{\bf Bias.} \sloppy{Under the conditions of Lemma \ref{lemm:diameter}, suppose that \smash{$\EE{Y^{(0)} \cond X = x}$} and \smash{$\EE{Y^{(1)} \cond X = x}$} are Lipschitz continuous, that the trees $\Gamma$ comprising the random forest are honest, and, moreover, that the overlap condition \eqref{eq:overlap} holds for some $\varepsilon > 0$.}
These conditions also imply that \smash{$|\mathbb{E}[Y^{(0)} \cond X = x]|, \, |\mathbb{E}[Y^{(1)} \cond X = x]| \leq M$} for some constant $M$, for all $x \in [0, \, 1]^d$.
Then, provided that $\alpha \leq 0.2$, the bias of the random forest at $x$ is bounded by
$$ \abs{\EE{\htau\p{x}} - \tau\p{x}} \lesssim 2M \, d \, \p{\frac{\varepsilon \, s}{2k - 1}}^{-\frac{1}{2} \frac{\log\p{\p{1 - \alpha}^{-1}}}{\log\p{\alpha^{-1}}} \, \frac{\pi}{d}}. $$
To establish this claim, we first seek with an analogue to Lemma \ref{lemm:diameter}, except now $s$ in \eqref{eq:diam_pf_1} is replaced by $s_{\min}$, i.e., the minimum of the number of cases (i.e., observations with $W_i = 1$) or controls (i.e., observations with $W_i = 0$) in the sample. A straight-forward computation then shows that $s_{\min}/s \gtrsim \varepsilon$, and that a variant of  \eqref{eq:diam_pf_2} where we replace $s$ with $\varepsilon s$ still holds for large $s$.
Next, to bound the bias itself, we start by applying unconfoundedness as in \eqref{eq:unconf_use}; then, the argument of Theorem \ref{theo:bias} goes through without modifications, provided we replace every instance of ``$s$'' with ``$\varepsilon s$''.

{\bf Incrementality.} Suppose that the conditions of Lemma \ref{lemm:pnn} hold and that $\Gamma$ is an honest $k$-regular causal tree in the sense of Definitions \ref{defi:honest_cmp} and \ref{defi:regular_cmp}. Suppose moreover that \smash{$\EE{Y^{(0/1)} \cond X = x}$} and \smash{$\Var{Y^{(0/1)} \cond X = x}$} are all Lipschitz continuous at $x$,
and that $\Var{Y \cond X = x} > 0$. Suppose, finally, that the overlap condition \eqref{eq:overlap} holds with $\varepsilon > 0$. Then $T$ is $\nu\p{s}$-incremental at $x$ with
$$ \nu\p{s} = \varepsilon \, {C_{f, \, d}} \, \big / \, {\log\p{s}^{d}}, $$
where $C_{f, \, d}$ is the constant from Lemma \ref{lemm:pnn}.

To prove this claim, we again focus on the case where $f(x) = 1$, in which case we use $C_{f, \, d} = 2^{-(d + 1)}\p{d - 1}!$.
We begin by setting up notation as in the proof of Lemma \ref{lemm:pnn}. We write our causal tree as $\Gamma\p{x; \, Z} = \sum_{i = 1}^s S_i Y_i$, where
$$ S_i =
\begin{cases}
\abs{\cb{i : X_i \in L(x; \, Z)}, \, W_i = 1}^{-1} & \text{ if } X_i \in L(x; \, Z) \eqand W_i = 1, \\
-\abs{\cb{i : X_i \in L(x; \, Z)}, \, W_i = 0}^{-1} & \text{ if } X_i \in L(x; \, Z) \eqand W_i = 0, \\
0 & \text{ else,}
\end{cases}$$
where $L(x; \, Z)$ denotes the leaf containing $x$, and let
$$P_i^W = 1\p{\cb{\text{$X_i$ is a $k$-PNN of $x$ among points with treatment status $W_i$}}}.$$
Finally, in a break from Lemma \ref{lemm:pnn}, define $w_{\min}(x; \, Z)$ as the minority class within the leaf $L(x; \, Z)$; more formally,
$$w_{\min} = 1\p{\cb{\abs{\cb{i : X_i \in L(x; \, Z)}, \, W_i = 1} \leq \abs{\cb{i : X_i \in L(x; \, Z)}, \, W_i = 0}}}. $$
By regularity of $\Gamma$, we know that the leaf $L(x; \, Z)$ can contain at most $2k - 1 $ examples from its minority class, and so $P_i^W = 0$ and $W = w_{\min}$ together imply that $S_i = 0$. Thus, we can verify that
$$ \EE{\abs{S_1} \, 1\p{\cb{W_1 = w_{\min}}} \cond Z_1} \leq \frac{1}{k} \, \EE{P_1^W \cond Z_1}. $$
We are now ready to use the same machinery as before. The random variables $P_1^W$ now satisfy
$$ \PP{\EE{P_1^W \cond Z_1} \geq \frac{1}{s^2 \, \PP{W = W_1}^2}} \lesssim k \, \frac{2^{d+1} \log\p{s}^d}{\p{d - 1}!} \, \frac{1}{s \, \PP{W = W_1}}; $$
by the above argument and $\varepsilon$-overlap \eqref{eq:overlap}, this immediately implies that
$$ \PP{\EE{\abs{S_1} \, 1\p{\cb{W_1 = w_{\min}} \cond Z_1}} \geq \frac{1}{k \, \varepsilon^2 \, s^2}}
\lesssim k \,  \frac{2^{d+1} \log\p{s}^d}{\p{d - 1}!} \, \frac{1}{\varepsilon s}. $$
By construction, we know that
$$ \EE{ \EE{\abs{S_1} \, 1\p{\cb{W_1 = w_{\min}}} \cond Z_1}} =  \EE{\abs{S_1} \, 1\p{\cb{W_1 = w_{\min}}}} = \frac{1}{s}, $$
which by the same argument as before implies that
$$ \EE{\EE{\abs{S_1} \, 1\p{\cb{W_1 = w_{\min}}} \cond Z_1}^2}
\gtrsim \frac{\p{d - 1}!}{2^{d+1} \log\p{s}^d} \, \frac{\varepsilon}{k \, s}. $$
By monotonicity, we then conclude that
$$ \EE{\EE{S_1 \cond Z_1}^2} =  \EE{\EE{\abs{S_1} \cond Z_1}^2}
\gtrsim \frac{\p{d - 1}!}{2^{d+1} \log\p{s}^d} \, \frac{\varepsilon}{k \, s}. $$
The second part of the proof follows from a straight-forward adaptation of Theorem \ref{theo:pnn}.
\end{proof}
\end{appendix}

\end{document}